\documentclass[12pt]{article}
\usepackage{amsmath}
\usepackage{amsfonts}
\usepackage{makeidx}
\usepackage{amssymb}
\usepackage{times}
\usepackage{anysize}

\marginsize{3cm}{3cm}{2cm}{2cm}
\newtheorem{satz}{Theorem}[section]
\newtheorem{definition}[satz]{Definition}
\newtheorem{lemma}[satz]{Lemma}
\newtheorem{koro}[satz]{Corollary}
\newtheorem{bemerkung}[satz]{Remark}

\newtheorem{proposition}[satz]{Proposition}
\newtheorem{notation}[satz]{Notation}
\newenvironment{proof}{\par\noindent {\it Proof:} \hspace{7pt}}{\hfill\hbox{\vrule width 7pt depth 0pt height 7pt}
\par\vspace{10pt}}

\newcommand{\dbigcup}{\mathop{\displaystyle \bigcup }}
\input{tcilatex}
\begin{document}

\title{Heat Production of Non--Interacting Fermions Subjected to Electric
Fields}
\author{J.-B. Bru \and W. de Siqueira Pedra \and C. Hertling}

\maketitle

\begin{abstract}
Electric resistance in conducting media is related to \emph{heat} (or \emph{%
entropy}) production in presence of electric fields. In this paper, by using
Araki's relative entropy for states, we mathematically define and analyze
the heat production of free fermions within external potentials. More
precisely, we investigate the heat production of the non-autonomous $C^{\ast
}$--dynamical system obtained from the fermionic second quantization of a
discrete Schr\"{o}%
\-%
dinger operator with bounded static potential in presence of an electric
field that is time-- and space--dependent. It is a first preliminary step
towards a mathematical description of transport properties of fermions from
thermal considerations. This program will be carried out in several papers.
The regime of small and slowly varying in space electric fields is important
in this context, and is studied the present paper. We use tree--decay bounds
of the $n$--point, $n\in 2\mathbb{N}$, correlations of the many--fermion
system to analyze this regime. We verify below the 1st law of thermodynamics
for the system under consideration. The latter implies, for systems doing no
work, that the heat produced by the electromagnetic field is exactly the
increase of the internal energy resulting from the modification of the
(infinite volume) state of the fermion system. The identification of heat
production with an energy increment is, among other things, technically
convenient. We initially focus our study on non--interacting (or free)
fermions, but our approach will be later applied to weakly interacting
fermions.
\end{abstract}

\tableofcontents%

\section{Introduction}

Ohm and Joule's laws, respectively derived in 1827 and 1840, are among the
most resilient laws of (classical) electricity theory. In standard
textbooks, the microscopic theory presented to explain Ohm's law is based on
the Drude model proposed in 1900, before the emergence of quantum mechanics.
In this model, the motion of electrons and ions is treated classically and
the interaction between these two species is modeled by perfectly elastic
random collisions. This quite elementary model explains very well DC-- and
AC--conductivities in metals, qualitatively. There are also improvements of
the Drude model taking into account quantum corrections. Nevertheless, to
our knowledge, there is no rigorous microscopic (complete) description of
the phenomenon of linear conductivity from first principles of quantum
mechanics. It is a highly non--trivial question. Indeed, problems are in
this case doubled because the electric resistance of conductors results from
both the presence of disorder in the host material and interactions between
charge carriers.

Rigorous quantum many--body theory is a notoriously difficult subject. The
hurdles that have to be overcome in order to arrive at important new
mathematical results involve many different fields of mathematics such as
probability theory, operator algebras, differential equations or functional
analysis. Disorder leads us to consider random Schr\"{o}dinger operators
like the celebrated Anderson model. It is an advanced and relatively mature
branch of mathematics. For instance, it is known that, in general, the
one--dimensional Anderson model only has purely point spectrum with a
complete set of localized eigenstates (Anderson localization) and it is thus
believed that no steady current can exist in this case. For more details,
see, e.g., \cite{Werner Kirsch}. Nevertheless, even in absence of
interactions, there are, to our knowledge, only few mathematical results on
transport properties of such models that yield Ohm's law in some form.

Indeed, Klein, Lenoble and M\"{u}ller introduced for the first time in \cite%
{Annale} the concept of a \textquotedblleft conduc%
\-%
tivity measure\textquotedblright\ for a system of non--interacting fermions
subjected to a random potential. More precisely, the authors considered the
Anderson tight--binding model in presence of a time--dependent spatially
homogeneous electric field that is adiabatically switched on. See also \cite%
{jfa} for further details on linear response theory of such a model. The
fermionic nature of charge carriers -- electrons or holes in crystals -- was
implemented by choosing the Fermi--Dirac distribution as the initial%
\footnote{%
This corresponds to $t\rightarrow -\infty $ in their approach.} density
matrix of particles. In \cite{Annale} only systems at zero temperature with
Fermi energy lying in the localization regime are considered, but it is
shown in \cite{JMP-autre} that a conductivity measure can also be defined
without the localization assumption and at any positive temperature. Their
study can thus be seen as a mathematical derivation of Ohm's law for
space--homogeneous electric fields having a specific time behavior. \cite%
{Cornean} is another mathematical result on free fermions proving Ohm's law
for graphene--like materials subjected to space--homogeneous and
time--periodic electric fields. Observe however that Joule's law and heat
production are not considered in \cite{Annale,JMP-autre,Cornean}.

We propose in a companion paper a different approach to the conducti%
\-%
vity measure\ based on a natural thermodynamic principle, the positivity of
the heat (or entropy) production, together with the Bochner--Schwartz
theorem \cite[Theorem IX.10]{ReedSimonII}. Our aim is to derive both Ohm and
Joule's laws for the Fourier components of time--dependent electric fields
from the analysis of the heat production in a realistic many--fermion
system. We first focus our study on \emph{non--interacting} (or free)
fermions in presence of disorder, here a static external potential, while
keeping in mind its possible extension to interacting fermions. Indeed, the
possibility of naturally extending results to systems with interaction is
one of the main advantages and motivations of the approach we propose here.
This will be discussed in more details in subsequent papers. Therefore,
although there is no interaction between fermions, we do \emph{not} restrict
our analyses to the one--particle Hilbert space to study transport
properties. Instead, our approach is based on the algebraic formulation of
many--fermion systems on lattices.

As observed by J. P. Joule in its original paper \cite{Joulesup}, the electric resistance is
associated with a heat production in the conducting system. Therefore, the
first step is to rigorously define and analyze the concept of \emph{heat}
production induced by electric fields on the fermion system. This study is
the main subject of the present paper. At constant temperature, the heat
production is, by definition, a quantity that is proportional to the \emph{%
entropy} production. The proportionality coefficient is the temperature of
the system. In order to give a precise mathematical definition of this
quantity, we use in Section \ref{Section Heat Production as Entropy} Araki's
relative entropy \cite{Araki1,Araki2,OhyaPetz} which, in our case, turns out
to be finite for all times. The latter uses the concept of spacial
derivative operators \cite{Connes}, see Section \ref{Section Quantum
Relative Entropy}. Part of the paper is devoted to recover the 1st law of
thermodynamics for the system under consideration, implying that the heat
production generated by the electromagnetic field is exactly the increase of
the \emph{internal} energy resulting from the modification of the (infinite
volume) state of the system. An increment of internal energy of the system
is defined here as being the increase of total energy minus the increase of
potential energy associated with the external electric field. See Sections %
\ref{Section Linear Response}. The 1st law of thermodynamics is an important
outcome in our context because it leads to more explicit expressions for the
heat production. Moreover, the increase of \emph{total} energy (i.e.,
internal plus potential energy) of the \emph{infinite} system obeys a
principle of conservation and is exactly the work performed by the electric
field on the charged particles. See Section \ref{Section Linear Response}.
This is well--known for dynamics on $C^{\ast }$--algebras generated by\
time--dependent bounded symmetric derivations. See for instance discussions
in \cite[Section 5.4.4.]{BratteliRobinson}. Here, we prove a version of that
result for our particular \emph{unbounded} case.

Note that Ohm's law corresponds to a linear\emph{\ }response to electric
fields. We thus rescale the strength of the electromagnetic potential by a
real parameter $\eta \in \mathbb{R}$ and will eventually take the limit $%
\eta \rightarrow 0$ (in a subsequent paper). Understanding the behavior of
the heat production as a function of $\eta $ is a necessary step in order to
obtain Ohm and Joule's laws. By using the fact mentioned above that the heat
production can be expressed in terms of an energy increment (Section \ref%
{Section Linear Response}), it can be shown that the heat production is a
real analytic function of the scaling parameter $\eta $. The coefficients of
the (absolutely convergent) power series in $\eta $ for the heat production
have the following important property: They behave, at any order $k\in
\mathbb{N}$, like the volume of the support (in space) of the applied
electric field, as physically expected. Such a behavior permits us, in
particular, to define densities (like heat production per unit volume).
Remark that naive bounds only predict that the $k$--the coefficient of the
power series for the heat production should behave like the $k$--power of
the volume of the support of the applied electric field. However, the heat
production is proven to behave like $\eta ^{2}$ times the volume of the
support of the applied electric field, provided $|\eta |$ is sufficiently
small. This is done in Section \ref{section Energy Increments as Power
Series}. See also Section \ref{Section Heat Production as Power series}.
Moreover, this result makes possible the study of non--quadratic (resp.
non--linear) corrections to Joule's law (resp. Ohm's law).

To obtain the properties described above for the power series in $\eta $
representing the heat production, we use a pivotal ingredient, namely \emph{%
tree--decay bounds} on multi--commutators. These bounds are derived in
Section \ref{section Tree--decay Bounds} and are useful to analyze
multi--commutators of monomials in annihilation and creation operators. They
will also be necessary in subsequent papers.

To conclude, our main assertions are Theorems \ref{main 1 copy(1)} and \ref%
{Thm Heat production as power series}, and Corollary \ref{tree bound main
copy(1)}. This paper is organized as follows:

\begin{itemize}
\item In Section \ref{Section main results} we describe non--autonomous $%
C^{\ast }$--dynamical systems for (free) fermions associated to discrete Schr%
\"{o}dinger operators with bounded (static) potentials in presence of an
electric field that is time-- and space--dependent.

\item Section \ref{Section Heat Production} introduces the concept of heat
production and discusses its main properties.

\item Section \ref{section Tree--decay Bounds} is devoted to tree--decay
bounds for expectation values of multi--commutators.

\item All technical proofs related to Section \ref{Section Heat Production}
are postponed to Section \ref{sect technical proofs}.

\item Section \ref{Section appendix} is an appendix containing two parts:
Section \ref{Section Quantum Relative Entropy} is a concise overview on the
quantum relative entropy \cite{Araki1,Araki2,OhyaPetz}. In Section \ref%
{Section finite volume system} it is shown that all properties of the
infinite system we use here result from the corresponding ones of finite
systems, at large volume. Note that Section \ref{Section finite volume
system} is not really used in other sections and has a supplementary
character, only.
\end{itemize}

\begin{notation}[Generic constants]
\label{remark constant}\mbox{
}\newline
To simplify notation, we denote by $D$ any generic positive and finite
constant. These constants do not need to be the same from one statement to
another.
\end{notation}

\section{$C^{\ast }$--Dynamical Systems for Free Fermions\label{sect 2.1
copy(1)}\label{Section main results}}

\subsection{CAR $C^{\ast }$--Algebra\label{sect 2.1}}

The host material for conducting fermions is assumed to be a cubic crystal.
Other crystal families could also be studied in the same way, but, for
simplicity, we refrain from considering them. The unit of length is chosen
such that the lattice spacing is exactly $1$. We thus use the $d$%
--dimensional cubic lattice $\mathfrak{L}:=\mathbb{Z}^{d}$ ($d\in \mathbb{N}$%
) to represent the crystal and we define $\mathcal{P}_{f}(\mathfrak{L}%
)\subset 2^{\mathfrak{L}}$ to be the set of all \emph{finite} subsets of $%
\mathfrak{L}$.

Within this framework, we consider an \emph{infinite} system of charged
fermions. To simplify notation we only consider spinless fermions with \emph{%
negative} charge. The cases of particles with spin and/or positively charged
particles can be treated by exactly the same methods. We denote by $\mathcal{%
U}$ the CAR algebra of the infinite system. More precisely, the (separable) $%
C^{\ast }$--algebra $\mathcal{U}$ is the inductive limit of the finite
dimensional $C^{\ast }$--algebras $\{\mathcal{U}_{\Lambda }\}_{\Lambda \in
\mathcal{P}_{f}(\mathfrak{L})}$ with identity $\mathbf{1}$ and generators $%
\{a_{x}\}_{x\in \Lambda }$ satisfying the canonical anti--commutation
relations: For any $x,y\in \mathfrak{L}$,
\begin{equation}
a_{x}a_{y}+a_{y}a_{x}=0\ ,\qquad a_{x}a_{y}^{\ast }+a_{y}^{\ast
}a_{x}=\delta _{x,y}\mathbf{1}\ .  \label{CAR}
\end{equation}

\subsection{Dynamics in Presence of Static External Potentials\label{Section
dynamics}}

It is widely accepted that electric resistance of conductors results from
both the presence of disorder in the host material and interactions between
charge carriers. Here, we only consider effects of disorder for
non--interacting fermions. That means physically that the particles obey the
Pauli exclusion principle but do not interact with each other via some
mutual force. This setup corresponds for example to the case of low electron
densities. Our approach can be applied to weakly interacting fermions on the
lattice, but the analysis would be -- from the technical point of view --
much more demanding of course.

Disorder in the crystal will be modeled in subsequent papers by a random
external potential coming from a probability space $(\Omega ,\mathfrak{A}%
_{\Omega },\mathfrak{a}_{\Omega })$ with $\Omega :=[-1,1]^{\mathfrak{L}}$.
In the present work, however, all studies are performed at any fixed $\omega
\in \Omega $ and all the results will be uniform with respect to (w.r.t.)
the choice of $\omega \in \Omega $. Note that, for any $\omega \in \Omega $,
$V_{\omega }\in \mathcal{B}(\ell ^{2}(\mathfrak{L}))$ is by definition the
self--adjoint multiplication operator with the function $\omega :\mathfrak{L}%
\rightarrow \lbrack -1,1]$. The static external potential $V_{\omega }$ is
of order $\mathcal{O}(1)$ and we rescale its strength by an additional
parameter $\lambda \in \mathbb{R}_{0}^{+}$ (i.e., $\lambda \geq 0$), see (%
\ref{rescaled}).

For any function $\omega \in \Omega $, we define the dynamics of the lattice
fermion system via a strongly continuous (quasi--free) group of
automorphisms of the $C^{\ast }$--algebra $\mathcal{U}$. To set up this time
evolution, we first define annihilation and creation operators of (spinless)
fermions with wave functions $\psi \in \ell ^{2}(\mathfrak{L})$ by
\begin{equation}
a(\psi ):=\sum\limits_{x\in \mathfrak{L}}\overline{\psi (x)}a_{x}\in
\mathcal{U}\ ,\quad a^{\ast }(\psi ):=\sum\limits_{x\in \mathfrak{L}}\psi
(x)a_{x}^{\ast }\in \mathcal{U}\ .  \label{creation operators}
\end{equation}%
These operators are well--defined because of (\ref{CAR}). Indeed,
\begin{equation}
\Vert a(\psi )\Vert ^{2},\Vert a^{\ast }(\psi )\Vert ^{2}=\left\Vert \psi
\right\Vert _{2}^{2}\ ,\qquad \psi \in \ell ^{2}(\mathfrak{L})\ ,
\label{norm cont of a}
\end{equation}%
and thus, the anti--linear (resp. linear) map $\psi \mapsto a(\psi )$ (resp.
$\psi \mapsto a^{\ast }(\psi )$) from $\ell ^{2}(\mathfrak{L})$\ to $%
\mathcal{U}$ is norm--continuous. Clearly, $a^{\ast }(\psi )=a(\psi )^{\ast
} $ for all $\psi \in \ell ^{2}(\mathfrak{L})$.

Now, for any function $\omega \in \Omega $ and strength $\lambda \in \mathbb{%
R}_{0}^{+}$ of the static (external) potential, we define the free dynamics
via the unitary group $\{\mathrm{U}_{t}^{(\omega ,\lambda )}\}_{t\in \mathbb{%
R}}$, where%
\begin{equation}
\mathrm{U}_{t}^{(\omega ,\lambda )}:=\exp (-it(\Delta _{\mathrm{d}}+\lambda
V_{\omega }))\in \mathcal{B}(\ell ^{2}(\mathfrak{L}))\ .  \label{rescaled}
\end{equation}%
Here, $\Delta _{\mathrm{d}}\in \mathcal{B}(\ell ^{2}(\mathfrak{L}))$ is (up
to a minus sign) the usual $d$--dimensional discrete Laplacian:%
\begin{equation}
\lbrack \Delta _{\mathrm{d}}(\psi )](x):=2d\psi (x)-\sum\limits_{z\in
\mathfrak{L},\text{ }|z|=1}\psi (x+z)\ ,\text{\qquad }x\in \mathfrak{L},\
\psi \in \ell ^{2}(\mathfrak{L})\ .  \label{discrete laplacian}
\end{equation}%
In particular, for an independent identically distributed (i.i.d.) random
potential $V_{\omega }$, $(\Delta _{\mathrm{d}}+\lambda V_{\omega })$ is the
Anderson tight--binding model acting on the Hilbert space $\ell ^{2}(%
\mathfrak{L})$. [Note that we could add some constant (chemical) potential
to the discrete Laplacian without changing our proofs.]

For all $\omega \in \Omega $ and $\lambda \in \mathbb{R}_{0}^{+}$, the
condition%
\begin{equation}
\tau _{t}^{(\omega ,\lambda )}(a(\psi ))=a((\mathrm{U}_{t}^{(\omega ,\lambda
)})^{\ast }\psi )\ ,\text{\qquad }t\in \mathbb{R},\ \psi \in \ell ^{2}(%
\mathfrak{L})\ ,  \label{rescaledbis}
\end{equation}%
uniquely defines a family $\tau ^{(\omega ,\lambda )}:=\{\tau _{t}^{(\omega
,\lambda )}\}_{t\in {\mathbb{R}}}$ of (Bogoliubov) automorphisms of $%
\mathcal{U}$, see \cite[Theorem 5.2.5]{BratteliRobinson}. The one--parameter
group $\tau ^{(\omega ,\lambda )}$ is strongly continuous and we denote its
(unbounded) generator by $\delta ^{(\omega ,\lambda )}$.

\subsection{Electromagnetic Fields\label{Section Electromagnetic Fields}}

The electromagnetic potential is defined by a compactly supported
time--dependent vector potential $\mathbf{A}\in \mathbf{C}_{0}^{\infty }$,
where%
\begin{eqnarray*}
\mathbf{C}_{0}^{\infty } &:=&\underset{l\in \mathbb{R}^{+}}{\dbigcup }\Big \{\mathbf{A}:\mathbb{R}\times {\mathbb{R}}^{d}\mapsto ({\mathbb{R}}^{d})^{\ast }\ |\ \exists \mathbf{B}\in
C_{0}^{\infty }(\mathbb{R}\times {\mathbb{R}}^{d};({\mathbb{R}}^{d})^{\ast })
\\
&&\qquad \qquad \qquad \qquad \qquad \text{with }\mathbf{A}(t,x)=\mathbf{B}(t,x)\mathbf{1}\big [x\in \lbrack -l,l^{d}] \big] \Big\}\ .
\end{eqnarray*}%
Here, $({\mathbb{R}}^{d})^{\ast }$ is the set of one--forms\footnote{%
In a strict sense, one should take the dual space of the tangent spaces $T({%
\mathbb{R}}^{d})_{x}$, $x\in {\mathbb{R}}^{d}$.} on ${\mathbb{R}}^{d}$ that
take values in $\mathbb{R}$. In other words, as $\left[ -l,l\right] ^{d}$ is
a compact subset of ${\mathbb{R}}^{d}$, $\mathbf{C}_{0}^{\infty }$ is the
union%
\begin{equation*}
\mathbf{C}_{0}^{\infty }=\underset{l\in \mathbb{R}^{+}}{\mathop{%
\displaystyle \bigcup } }C_{0}^{\infty }(\mathbb{R}\times \left[ -l,l\right]
^{d};({\mathbb{R}}^{d})^{\ast })
\end{equation*}
of the space of smooth compactly supported functions $\mathbb{R}\times \left[
-l,l\right] ^{d}\rightarrow ({\mathbb{R}}^{d})^{\ast }$ for $l\in \mathbb{R}%
^{+}$. The smoothness of $\mathbf{A}$ is not really necessary at this stage
but will be technically convenient in subsequent papers. Here, only the
continuously differentiability of the map $t\mapsto \mathbf{A}(t,\cdot )$ is
really crucial to define below the electric field and the non--autonomous
dynamics.

Since $\mathbf{A}\in \mathbf{C}_{0}^{\infty }$, $\mathbf{A}(t,x)=0$ for all $%
t\leq t_{0}$, where $t_{0}\in \mathbb{R}$ is some initial time. We use the
Weyl gauge (also named temporal gauge) for the electromagnetic field and as
a consequence,%
\begin{equation}
E_{\mathbf{A}}(t,x):=-\partial _{t}\mathbf{A}(t,x)\ ,\quad t\in \mathbb{R},\
x\in \mathbb{R}^{d}\ ,  \label{V bar 0}
\end{equation}%
is the electric field associated with $\mathbf{A}$.

Note that the time $t_{1}\geq t_{0}$ when the electric field is turned off
can be chosen as arbitrarily large and one recovers the DC--regime by taking
$t_{1}>>1$. However, for electric fields slowly varying in time, charge
carriers have time to move and significantly change the charge density,
producing an additional, self--generated, internal electric field. This
contribution is not taken into account in our model.

Finally, observe that space--dependent electromagnetic potentials imply
magnetic fields which interact with fermion spins. We neglect this
contribution because such a term will become negligible for electromagnetic
potentials slowly varying in space. This justifies the assumption of
fermions with zero--spin. In any case, our study can be performed for
non--zero fermion spins exactly in the same way. We omit this generalization
for simplicity.

\subsection{Dynamics in Presence of Time-Dependent Electromagnetic Fields
\label{section Dynamics}}

Recall that we only consider \emph{negatively} charged fermions. We choose
units such that the charge of fermions is $-1$. The (minimal) coupling of
the vector potential $\mathbf{A}\in \mathbf{C}_{0}^{\infty }$ to the fermion
system is achieved through a redefinition of the discrete Laplacian. Indeed,
we define the time--dependent self--adjoint operator $\Delta _{\mathrm{d}}^{(%
\mathbf{A})}\in \mathcal{B}(\ell ^{2}(\mathfrak{L}))$ by
\begin{equation}
\langle \mathfrak{e}_{x},\Delta _{\mathrm{d}}^{(\mathbf{A})}\mathfrak{e}%
_{y}\rangle =\exp \left( -i\int\nolimits_{0}^{1}\left[ \mathbf{A}(t,\alpha
y+(1-\alpha )x)\right] (y-x)\mathrm{d}\alpha \right) \langle \mathfrak{e}%
_{x},\Delta _{\mathrm{d}}\mathfrak{e}_{y}\rangle  \label{eq discrete lapla A}
\end{equation}%
for all $x,y\in \mathfrak{L}$, where $\langle \cdot ,\cdot \rangle $ is here
the canonical scalar product in $\ell ^{2}(\mathfrak{L})$ and $\left\{ \mathfrak{e}%
_{x}\right\} _{x\in \mathfrak{L}}$ is the canonical orthonormal basis $%
\mathfrak{e}_{x}(y)\equiv \delta _{x,y}$ of $\ell ^{2}(\mathfrak{L})$. In
Equation (\ref{eq discrete lapla A}), $\alpha y+(1-\alpha )x$ and $y-x$ are
seen as vectors in ${\mathbb{R}}^{d}$.

Observe that there is $l_{0}\in \mathbb{R}^{+}$ such that%
\begin{equation*}
\Delta _{\mathrm{d}}^{(\mathbf{A})}-\Delta _{\mathrm{d}}\in \mathcal{B}(\ell
^{2}([-l_{0},l_{0}]^{d}\cap \mathfrak{L}))\subset \mathcal{B}(\ell ^{2}(%
\mathfrak{L}))
\end{equation*}%
for all times $t\in \mathbb{R}$, because $\mathbf{A}$ is by definition
compactly supported. Note also that, for simplicity, the time dependence is
often omitted in the notation
\begin{equation*}
\Delta _{\mathrm{d}}^{(\mathbf{A})}\equiv \Delta _{\mathrm{d}}^{(\mathbf{A}%
(t,\cdot ))}\ ,\qquad t\in \mathbb{R}\ ,
\end{equation*}%
but one has to keep in mind that the dynamics is \emph{non--autonomous}.

Indeed, the Schr\"{o}dinger equation on the one--particle Hilbert space $%
\ell ^{2}(\mathfrak{L})$ with time--dependent Hamiltonian $(\Delta _{\mathrm{%
d}}^{(\mathbf{A})}+\lambda V_{\omega })$ and initial value $\psi \in \ell
^{2}(\mathfrak{L})$ at $t=t_{0}$ has a unique solution $\mathrm{U}%
_{t,t_{0}}^{(\omega ,\lambda ,\mathbf{A})}\psi $ for any $t\geq t_{0}$.
Here, for any $\omega \in \Omega $, $\lambda \in \mathbb{R}_{0}^{+}$ and $%
\mathbf{A}\in \mathbf{C}_{0}^{\infty }$,
\begin{equation*}
\{\mathrm{U}_{t,s}^{(\omega ,\lambda ,\mathbf{A})}\}_{t\geq s}\subset
\mathcal{B}(\ell ^{2}(\mathfrak{L}))
\end{equation*}%
is the two--parameter group of unitary operators on $\ell ^{2}(\mathfrak{L})$
generated by the (anti--self--adjoint) operator $-i(\Delta _{\mathrm{d}}^{(%
\mathbf{A})}+\lambda V_{\omega })$:
\begin{equation}
\forall s,t\in {\mathbb{R}},\ t\geq s:\quad \partial _{t}\mathrm{U}%
_{t,s}^{(\omega ,\lambda ,\mathbf{A})}=-i(\Delta _{\mathrm{d}}^{(\mathbf{A}%
(t,\cdot ))}+\lambda V_{\omega })\mathrm{U}_{t,s}^{(\omega ,\lambda ,\mathbf{%
A})}\ ,\quad \mathrm{U}_{s,s}^{(\omega ,\lambda ,\mathbf{A})}:=\mathbf{1}\ .
\label{time evolution one-particle}
\end{equation}%
Since the map
\begin{equation}
t\mapsto (\Delta _{\mathrm{d}}^{(\mathbf{A}(t,\cdot ))}+\lambda V_{\omega
})\in \mathcal{B}(\ell ^{2}(\mathfrak{L}))
\label{Anderson Model with Electric Field}
\end{equation}%
from $\mathbb{R}$ to the space $\mathcal{B}(\ell ^{2}(\mathfrak{L}))$ of
bounded operators acting on $\ell ^{2}(\mathfrak{L})$ is continuously
differentiable for every $\mathbf{A}\in \mathbf{C}_{0}^{\infty }$, $\{%
\mathrm{U}_{t,s}^{(\omega )}\}_{t\geq s}$ is a norm--continuous
two--parameter group of unitary operators. For more details, see Section \ref%
{Section existence dynamics}.

Therefore, for all $\omega \in \Omega $, $\lambda \in \mathbb{R}_{0}^{+}$
and $\mathbf{A}\in \mathbf{C}_{0}^{\infty }$, the condition%
\begin{equation}
\tau _{t,s}^{(\omega ,\lambda ,\mathbf{A})}(a(\psi ))=a((\mathrm{U}%
_{t,s}^{(\omega ,\lambda ,\mathbf{A})})^{\ast }\psi )\ ,\text{\qquad }t\geq
s,\ \psi \in \ell ^{2}(\mathfrak{L})\ ,  \label{Cauchy problem 0}
\end{equation}%
uniquely defines a family $\{\tau _{t,s}^{(\omega ,\lambda ,\mathbf{A}%
)}\}_{t\geq s}$ of Bogoliubov automorphisms of the $C^{\ast }$--algebra $%
\mathcal{U}$, see \cite[Theorem 5.2.5]{BratteliRobinson}. It is a strongly
continuous two--parameter family which obeys the non--autonomous\ evolution
equation%
\begin{equation}
\forall s,t\in {\mathbb{R}},\ t\geq s:\quad \partial _{t}\tau
_{t,s}^{(\omega ,\lambda ,\mathbf{A})}=\tau _{t,s}^{(\omega ,\lambda ,%
\mathbf{A})}\circ \delta _{t}^{(\omega ,\lambda ,\mathbf{A})},\quad \tau
_{s,s}^{(\omega ,\lambda ,\mathbf{A})}:=\mathbf{1}\ ,
\label{Cauchy problem 1}
\end{equation}%
with $\mathbf{1}$ being the identity of $\mathcal{U}$. Here, at any \emph{%
fixed} time $t\in {\mathbb{R}}$, $\delta _{t}^{(\omega ,\lambda ,\mathbf{A}%
)} $ is the infinitesimal generator of the (Bogoliubov) group $\{\tau
_{s}^{(\omega ,\lambda ,\mathbf{A})}\}_{s\in {\mathbb{R}}}\equiv \{\tau
_{s}^{(\omega ,\lambda ,\mathbf{A}(t,\cdot ))}\}_{s\in {\mathbb{R}}}$ of
automorphisms defined by replacing $\Delta _{\mathrm{d}}$ with $\Delta _{%
\mathrm{d}}^{(\mathbf{A})}$ in (\ref{rescaled}), see (\ref{explicit delta}).
For more details on the properties of $\{\tau _{t,s}^{(\omega ,\lambda ,%
\mathbf{A})}\}_{t\geq s}$, see also Sections \ref{Section existence dynamics}%
--\ref{Section interaction picture}.

Observe that one can equivalently use either (\ref{Cauchy problem 0}) or (%
\ref{Cauchy problem 1}) to define the dynamics, see also Proposition \ref%
{bound incr 1 Lemma copy(8)}. However, only the second formulation (\ref%
{Cauchy problem 1}) is appropriate to study transport properties of systems
of interacting fermions on the lattice in its algebraic formulation.

\begin{bemerkung}[Heisenberg picture]
\label{Heisenberg Picture remark}\mbox{
}\newline
The initial value problem (\ref{Cauchy problem 1}) can easily be understood
in the Heisenberg picture. The time--evolution of any observable $B_{s}\in
\mathcal{B}(\ell ^{2}(\mathfrak{L}))$ at initial time $t=s\in {\mathbb{R}}$
equals $B_{t}=(\mathrm{U}_{t,s}^{(\omega ,\lambda ,\mathbf{A})})^{\ast }B_{s}%
\mathrm{U}_{t,s}^{(\omega ,\lambda ,\mathbf{A})}$ for $t\geq s$, which yields%
\begin{equation*}
\forall t\geq s:\quad \partial _{t}B_{t}=(\mathrm{U}_{t,s}^{(\omega ,\lambda
,\mathbf{A})})^{\ast }i[\Delta _{\mathrm{d}}^{(\mathbf{A})}+\lambda
V_{\omega },B_{s}]\mathrm{U}_{t,s}^{(\omega ,\lambda ,\mathbf{A})}\ .
\end{equation*}%
The action of the symmetric derivation $\delta _{t}^{(\omega ,\lambda ,%
\mathbf{A})}$ in (\ref{Cauchy problem 1})\ is related to the above
commutator whereas the map $B\mapsto (\mathrm{U}_{t,s}^{(\omega ,\lambda ,%
\mathbf{A})})^{\ast }B\mathrm{U}_{t,s}^{(\omega ,\lambda ,\mathbf{A})}$
leads to the family $\{\tau _{t,s}^{(\omega ,\lambda ,\mathbf{A})}\}_{t\geq
s}$ in the second quantization. See also Theorem \ref{Theo int pict}.
\end{bemerkung}

\subsection{Time--Dependent State of the System\label{Section initia states}}

States on the $C^{\ast }$--algebra $\mathcal{U}$ are, by definition,
continuous linear functionals $\rho \in \mathcal{U}^{\ast }$ which are
normalized and positive, i.e., $\rho (\mathbf{1})=1$ and $\rho (A^{\ast
}A)\geq 0$ for all $A\in \mathcal{U}$.

It is well--known that, at finite volume, the thermodynamic equilibrium of
the system is described by the corresponding Gibbs state, which is the
unique state minimizing the free--energy. It is stationary and satisfies the
so--called KMS condition. The latter also makes sense in infinite volume and
is thus used to define the thermodynamic equilibrium of the infinite system.
See, e.g., Section \ref{Section finite volume system}, in particular Theorem %
\ref{conv Gibbs}.

Therefore, we assume that, for any function $\omega \in \Omega $ and
strength $\lambda \in \mathbb{R}_{0}^{+}$ of the static potential, the state
of the system before the electric field is switched on is the unique $(\tau
^{(\omega ,\lambda )},\beta )$--KMS state $\varrho ^{(\beta ,\omega ,\lambda
)}$, see \cite[Example 5.3.2.]{BratteliRobinson} or \cite[Theorem 5.9]%
{AttalJoyePillet2006a}. Here, $\beta \in \mathbb{R}^{+}$ (i.e., $\beta >0$)
is the inverse temperature of the fermion system at equilibrium.

Since $\mathbf{A}(t,x)=0$ for all $t\leq t_{0}$, the time evolution of the
state of the system thus equals%
\begin{equation}
\rho _{t}^{(\beta ,\omega ,\lambda ,\mathbf{A})}:=\left\{
\begin{array}{lll}
\varrho ^{(\beta ,\omega ,\lambda )} & , & \qquad t\leq t_{0}\ , \\
\varrho ^{(\beta ,\omega ,\lambda )}\circ \tau _{t,t_{0}}^{(\omega ,\lambda ,%
\mathbf{A})} & , & \qquad t\geq t_{0}\ .%
\end{array}%
\right.  \label{time dependent state}
\end{equation}%
Remark that the definition does not depend on the particular choice of
initial time $t_{0}$ because of the stationarity of the KMS state $\varrho
^{(\beta ,\omega ,\lambda )}$ w.r.t. the unperturbed dynamics (cf. (\ref%
{stationary})). The state $\rho _{t}^{(\beta ,\omega ,\lambda ,\mathbf{A})}$
is, by construction, a quasi--free state.

\section{Heat Production\label{Section Heat Production}}

\subsection{Heat Production as Quantum Relative Entropy\label{Section Heat
Production as Entropy}}

Joule's law describes the rate at which resistance converts electric energy
into \emph{heat}. That quantity of heat is not characterized here by a \emph{%
local }increase of temperature, but it is proportional to an \emph{entropy}
production. The proportionality coefficient is of course the temperature $%
\beta ^{-1}\in \mathbb{R}^{+}$, which is is seen as a \emph{global}
parameter of the \emph{infinite} system. The heat production is thus, by
definition, a relative quantity w.r.t. the reference state of the system,
that is, the thermal (or equilibrium) state $\varrho ^{(\beta ,\omega
,\lambda )}$ for $\beta \in \mathbb{R}^{+}$, $\omega \in \Omega $ and $%
\lambda \in \mathbb{R}_{0}^{+}$. Its mathematical formulation requires
Araki's notion of \emph{relative entropy} \cite{Araki1,Araki2,OhyaPetz}.

The latter takes a simple form for finite dimensional $C^{\ast }$--algebras
like the local fermion algebras $\{\mathcal{U}_{\Lambda }\}_{\Lambda \in
\mathcal{P}_{f}(\mathfrak{L})}$: Let $\Lambda \in \mathcal{P}_{f}(\mathfrak{L%
})$ and denote by $\mathrm{tr}$ the normalized trace on $\mathcal{U}%
_{\Lambda }$, also named the tracial state of $\mathcal{U}_{\Lambda }$. By
\cite[Lemma 3.1 (i)]{Araki-Moriya}, for any state $\rho \in \mathcal{U}%
_{\Lambda }^{\ast }$, there is a unique adjusted density matrix $\mathrm{d}%
_{\rho }\in \mathcal{U}$, that is, $\mathrm{d}_{\rho }\geq 0$, $\mathrm{tr}%
\left( \mathrm{d}_{\rho }\right) =1$ and $\rho (A)=\mathrm{tr}\left( \mathrm{%
d}_{\rho }A\right) $ for all $A\in \mathcal{U}_{\Lambda }$. We define by $%
\mathrm{supp}\left( \rho \right) $ the smallest projection $\mathrm{P}\in
\mathcal{U}_{\Lambda }$ such that $\rho (\mathrm{P})=1$. Then, the relative
entropy of a state $\rho _{1}\in \mathcal{U}_{\Lambda }^{\ast }$ w.r.t. $%
\rho _{2}\in \mathcal{U}_{\Lambda }^{\ast }$ is defined by (\ref{relative
entropy general}) for $\mathcal{X}=\mathcal{U}_{\Lambda }$ and, by finite
dimensionality, it equals%
\begin{equation}
S_{\mathcal{U}_{\Lambda }}\left( \rho _{1}|\rho _{2}\right) =\left\{
\begin{array}{lll}
\rho _{1}\left( \ln \mathrm{d}_{\rho _{1}}-\ln \mathrm{d}_{\rho _{2}}\right)
\in \mathbb{R}_{0}^{+} & , & \qquad \text{if }\mathrm{supp}\left( \rho
_{2}\right) \geq \mathrm{supp}\left( \rho _{1}\right) \ , \\
+\infty & , & \qquad \text{otherwise}\ ,%
\end{array}%
\right.  \label{relative entropy finite dimensional}
\end{equation}%
under the convention $x\ln x|_{x=0}:=0$, see Lemma \ref{local
AC-conductivity lemma copy(2)}. It is always a non--negative quantity. See
for instance \cite[Eq. (1.3) and Proposition 1.1]{OhyaPetz}.

For more general $C^{\ast }$--algebras like the CAR $C^{\ast }$--algebra $%
\mathcal{U}$ of the infinite system, Araki's definition of relative entropy
\cite{Araki1,Araki2,OhyaPetz} invokes the modular theory. This definition is
rather abstract, albeit standard, and for the reader's convenience we thus
postpone it until Section \ref{Section Quantum Relative Entropy}. Indeed,
using the boxes%
\begin{equation}
\Lambda _{L}:=\{(x_{1},\ldots ,x_{d})\in \mathfrak{L}\,:\,|x_{1}|,\ldots
,|x_{d}|\leq L\}\in \mathcal{P}_{f}(\mathfrak{L})  \label{eq:def lambda n}
\end{equation}%
for any $L\in \mathbb{R}^{+}$, we observe that $\{\mathcal{U}_{\Lambda
_{L}}\}_{L\in \mathbb{R}^{+}}$ is an increasing net of $C^{\ast }$%
--subalgebras of the $C^{\ast }$--algebra $\mathcal{U}$. Moreover, the $\ast
$--algebra%
\begin{equation}
\mathcal{U}_{0}:=\underset{L\in \mathbb{R}^{+}}{\bigcup }\mathcal{U}%
_{\Lambda _{L}}\subset \mathcal{U}  \label{simple}
\end{equation}%
of local elements is, by construction, dense in $\mathcal{U}$. (Indeed, $%
\mathcal{U}$ is by definition the completion of the normed $\ast $--algebra $%
\mathcal{U}_{0}$.) We thus define the relative entropy of any state $\rho
_{1}\in \mathcal{U}^{\ast }$ w.r.t. $\rho _{2}\in \mathcal{U}^{\ast }$ by
\begin{equation}
\mathrm{S}\left( \rho _{1}|\rho _{2}\right) :=\underset{L\rightarrow \infty }%
{\lim }S_{\mathcal{U}_{\Lambda _{L}}}\left( \rho _{1,\Lambda _{L}}|\rho
_{2,\Lambda _{L}}\right) =\underset{L\in \mathbb{R}^{+}}{\sup }S_{\mathcal{U}%
_{\Lambda _{L}}}\left( \rho _{1,\Lambda _{L}}|\rho _{2,\Lambda _{L}}\right)
\in \left[ 0,\infty \right]  \label{relative entropy general0}
\end{equation}%
with $\rho _{1,\Lambda _{L}}$ and $\rho _{2,\Lambda _{L}}$ being the
restrictions to $\mathcal{U}_{\Lambda _{L}}$ of the states $\rho _{1}$ and $%
\rho _{2}$, respectively. By \cite[Proposition 5.23 (vi)]{OhyaPetz}, this
limit exists and equals Araki's relative entropy, i.e., $\mathrm{S}\left(
\rho _{1}|\rho _{2}\right) =S_{\mathcal{U}}\left( \rho _{1}|\rho _{2}\right)
$ with $S_{\mathcal{U}}$ defined by (\ref{relative entropy general}) for $%
\mathcal{X}=\mathcal{U}$. In particular, it is a non--negative (possibly
infinite) quantity. Since $\mathrm{S}=S_{\mathcal{U}}$, note that the second
equality in (\ref{relative entropy general0}) follows from \cite[Proposition
5.23 (iv)]{OhyaPetz}, which in turn results from the Uhlmann monotonicity
theorem for Schwarz mappings \cite[Proposition 5.3]{OhyaPetz}.

Therefore, the \emph{heat production} is defined from (\ref{time dependent
state}) and (\ref{relative entropy general0}) as follows:

\begin{definition}[Heat production]
\label{Heat production definition}\mbox{ }\newline
For any $\beta \in \mathbb{R}^{+}$, $\omega \in \Omega $, $\lambda \in
\mathbb{R}_{0}^{+}$ and $\mathbf{A}\in \mathbf{C}_{0}^{\infty }$, $\mathbf{Q}%
^{(\omega ,\mathbf{A})}\equiv \mathbf{Q}^{(\beta ,\omega ,\lambda ,\mathbf{A}%
)}$ is defined as a map from $\mathbb{R}$ to $\overline{\mathbb{R}}$ by%
\begin{equation*}
\mathbf{Q}^{(\omega ,\mathbf{A})}\left( t\right) :=\beta ^{-1}\mathrm{S}%
(\rho _{t}^{(\beta ,\omega ,\lambda ,\mathbf{A})}|\varrho ^{(\beta ,\omega
,\lambda )})\in \left[ 0,\infty \right] \ .
\end{equation*}
\end{definition}

\noindent The heat production $\mathbf{Q}^{(\omega ,\mathbf{A})}\left(
t\right) $ may a priori be infinite for some time $t\in \mathbb{R}$. We
prove in the next section that $\mathbf{Q}^{(\omega ,\mathbf{A})}$ is finite
for all times. In particular, the states $\varrho ^{(\beta ,\omega ,\lambda
)}$ and $\rho _{t}^{(\beta ,\omega ,\lambda ,\mathbf{A})}$ are globally
similar.

\subsection{Heat Production and 1st Law of Thermodynamics\label{Section
Linear Response}}

\noindent \textit{In a thermodynamic process of a closed system, the
increment in the internal energy is equal to the difference between the
increment of heat accumulated by the system and the increment of work done
by it}.\smallskip

\hfill \lbrack Clausius, English translation, 1850]\bigskip

\noindent This is the celebrated \emph{1st law of thermodynamics}, see \cite%
{thermo0}. For an historical and mathematical account on thermodynamics,
see, e.g., \cite{Thermo1}. See also \cite{thermo2} for an interesting
derivation of this law from quantum statistical mechanics.

In the system considered here, the increment of \emph{total} energy follows
from the interaction between electromagnetic fields and charged fermions.
Part of this increment results from the change of internal state of
fermions. It is interpreted below as an increment of \emph{internal} energy
of the system. The other part is an electromagnetic energy that is generally
non--vanishing even if the internal state of fermions would stay at
equilibrium. By this reason, this part is seen below as an increase of
electromagnetic\emph{\ potential} energy of charged particles within the
electromagnetic field. As the system under consideration does not interact
with surroundings and thus can neither perform work nor exchange heat, all
the increase of internal energy is expected to be converted into heat, by
the 1st law of thermodynamics. Therefore, the heat production $\mathbf{Q}%
^{(\omega ,\mathbf{A})}$ should be related to the increment of the internal
energy of the system. This is far from being explicit in Definition \ref%
{Heat production definition}. We show that it is indeed the case for the
fermion system considered here.

To this end, we first need to give precise definitions of the increments of
\emph{total}, \emph{internal} and (electromagnetic)\emph{\ potential}
energies. In quantum mechanics, these energies should be associated with
total, internal and potential energy observables, that is in our case,
self--adjoint elements of $\mathcal{U}$. They are defined as follows: For
any $L\in \mathbb{R}^{+}$, $[L]\in \mathbb{N}$ being its integer part, the
\emph{internal} energy observable in the box $\Lambda _{L}$ (\ref{eq:def
lambda n}) of side length $2[L]+1$ is defined by%
\begin{equation}
H_{L}^{(\omega ,\lambda )}:=\sum\limits_{x,y\in \Lambda _{L}}\langle
\mathfrak{e}_{x},(\Delta _{\mathrm{d}}+\lambda V_{\omega })\mathfrak{e}%
_{y}\rangle a_{x}^{\ast }a_{y}\in \mathcal{U}\ .  \label{def H loc}
\end{equation}%
It is the second quantization of the one--particle operator $\Delta _{%
\mathrm{d}}+\lambda V_{\omega }$ restricted to the subspace $\ell
^{2}(\Lambda _{L})\subset \ell ^{2}(\mathfrak{L})$. When the electromagnetic
field is switched on, i.e., for $t\geq t_{0}$, the (time--dependant) \emph{%
total} energy observable in the box $\Lambda _{L}$ is then equal to $%
H_{L}^{(\omega ,\lambda )}+W_{t}^{\mathbf{A}}$, where, for any $\mathbf{A}%
\in \mathbf{C}_{0}^{\infty }$ and $t\in \mathbb{R}$,
\begin{equation}
W_{t}^{\mathbf{A}}:=\sum\limits_{x,y\in \Lambda _{L}}\langle \mathfrak{e}%
_{x},(\Delta _{\mathrm{d}}^{(\mathbf{A})}-\Delta _{\mathrm{d}})\mathfrak{e}%
_{y}\rangle a_{x}^{\ast }a_{y}\in \mathcal{U}  \label{eq def W}
\end{equation}%
is the electromagnetic\emph{\ potential} energy observable.

As a consequence, for any $\beta \in \mathbb{R}^{+}$, $\omega \in \Omega $, $%
\lambda \in \mathbb{R}_{0}^{+}$, $\mathbf{A}\in \mathbf{C}_{0}^{\infty }$
and $t\in \mathbb{R}$, the \emph{total} energy increment engendered by the
interaction with the electromagnetic field equals
\begin{equation}
\lim_{L\rightarrow \infty }\left\{ \rho _{t}^{(\beta ,\omega ,\lambda ,%
\mathbf{A})}(H_{L}^{(\omega ,\lambda )}+W_{t}^{\mathbf{A}})-\varrho ^{(\beta
,\omega ,\lambda )}(H_{L}^{(\omega ,\lambda )})\right\} =\mathbf{S}^{(\omega
,\mathbf{A})}\left( t\right) +\mathbf{P}^{(\omega ,\mathbf{A})}\left(
t\right) \ .  \label{lim_en_incr full}
\end{equation}%
Here, $\mathbf{S}^{(\omega ,\mathbf{A})}\equiv \mathbf{S}^{(\beta ,\omega
,\lambda ,\mathbf{A})}$ is the \emph{internal} energy increment defined as a
map from $\mathbb{R}$ to $\overline{\mathbb{R}}$ by%
\begin{equation}
\mathbf{S}^{(\omega ,\mathbf{A})}\left( t\right) :=\lim_{L\rightarrow \infty
}\left\{ \rho _{t}^{(\beta ,\omega ,\lambda ,\mathbf{A})}(H_{L}^{(\omega
,\lambda )})-\varrho ^{(\beta ,\omega ,\lambda )}(H_{L}^{(\omega ,\lambda
)})\right\} \ ,  \label{entropic energy increment}
\end{equation}%
whereas the electromagnetic \emph{potential} energy (increment) $\mathbf{P}%
^{(\omega ,\mathbf{A})}\equiv \mathbf{P}^{(\beta ,\omega ,\lambda ,\mathbf{A}%
)}$ is defined as a map from $\mathbb{R}$ to $\mathbb{R}$ by%
\begin{equation}
\mathbf{P}^{(\omega ,\mathbf{A})}\left( t\right) :=\rho _{t}^{(\beta ,\omega
,\lambda ,\mathbf{A})}(W_{t}^{\mathbf{A}})=\rho _{t}^{(\beta ,\omega
,\lambda ,\mathbf{A})}(W_{t}^{\mathbf{A}})-\varrho ^{(\beta ,\omega ,\lambda
)}(W_{t_{0}}^{\mathbf{A}})  \label{electro free energy}
\end{equation}%
for any $\beta \in \mathbb{R}^{+}$, $\omega \in \Omega $, $\lambda \in
\mathbb{R}_{0}^{+}$ and $\mathbf{A}\in \mathbf{C}_{0}^{\infty }$. In
particular, $\mathbf{S}^{(\omega ,\mathbf{A})}$ is only non--vanishing if
the state of the fermion system changes, whereas $\mathbf{P}^{(\omega ,%
\mathbf{A})}$ vanishes in absence of external electromagnetic potential.

Remark that
\begin{equation}
\mathbf{P}^{(\omega ,\mathbf{A})}\left( t\right) =\left\{ \rho _{t}^{(\beta
,\omega ,\lambda ,\mathbf{A})}(W_{t}^{\mathbf{A}})-\varrho ^{(\beta ,\omega
,\lambda )}(W_{t}^{\mathbf{A}})\right\} +\varrho ^{(\beta ,\omega ,\lambda
)}(W_{t}^{\mathbf{A}})\ .  \label{electro free energybis}
\end{equation}%
The last part is the raw\emph{\ }electromagnetic energy given to the system
at equilibrium. It is the so--called \emph{diamagnetic} energy, which will
be studied in subsequent papers. The energy increment between brackets in
the right hand side (r.h.s.) of (\ref{electro free energybis}) will also be
analyzed in detail later and is part of a so--called \emph{paramagnetic}
energy increment. It is the amount of electromagnetic potential energy
absorbed or released by\ the fermion system\ to change its internal state.

It is not a priori obvious that the limits (\ref{lim_en_incr full}) and (\ref%
{entropic energy increment}) exist because, in general,%
\begin{equation*}
\rho _{t}^{(\beta ,\omega ,\lambda ,\mathbf{A})}(H_{L}^{(\omega ,\lambda )})=%
\mathcal{O}(L^{d})\ .
\end{equation*}%
We show below that these limits have nevertheless finite real--values.
Indeed, we infer from Theorem \ref{coro heat production2} that, for any $%
\beta \in \mathbb{R}^{+}$, $\omega \in \Omega $, $\lambda \in \mathbb{R}%
_{0}^{+}$ and $\mathbf{A}\in \mathbf{C}_{0}^{\infty }$, the energy sum (\ref%
{lim_en_incr full}) is the \emph{work} performed on the system by the
electromagnetic field at time $t\geq t_{0}$:%
\begin{equation}
\mathbf{S}^{(\omega ,\mathbf{A})}\left( t\right) +\mathbf{P}^{(\omega ,%
\mathbf{A})}\left( t\right) =\int_{t_{0}}^{t}\rho _{s}^{(\beta ,\omega
,\lambda ,\mathbf{A})}\left( \partial _{s}W_{s}^{\mathbf{A}}\right) \mathrm{d%
}s\ .  \label{work}
\end{equation}%
Here, $\rho _{t}^{(\beta ,\omega ,\lambda ,\mathbf{A})}(\partial _{t}W_{t}^{%
\mathbf{A}})$ is interpreted as the infinitesimal work of the
electromagnetic field at time $t\in \mathbb{R}$. See for instance
discussions in \cite[Section 5.4.4.]{BratteliRobinson}. Note that this
conservation law is not completely obvious in our case because the
considered system is \emph{infinitely extended}.

We derive now the 1st law of thermodynamics:

\begin{satz}[1st law of thermodynamics]
\label{main 1 copy(1)}\mbox{
}\newline
For any $\beta \in \mathbb{R}^{+}$, $\omega \in \Omega $, $\lambda \in
\mathbb{R}_{0}^{+}$, $\mathbf{A}\in \mathbf{C}_{0}^{\infty }$ and $t\in
\mathbb{R}$,%
\begin{equation*}
\mathbf{Q}^{(\omega ,\mathbf{A})}\left( t\right) =\mathbf{S}^{(\omega ,%
\mathbf{A})}\left( t\right) \in \mathbb{R}_{0}^{+}\ .
\end{equation*}%
In particular, the maps $\mathbf{Q}^{(\omega ,\mathbf{A})}$ and $\mathbf{S}%
^{(\omega ,\mathbf{A})}$ respectively defined by Definition \ref{Heat
production definition} and (\ref{entropic energy increment}) take always
positive and finite values for all times.
\end{satz}

\begin{proof}
All arguments are given in Section \ref{Section Existence}, see Theorem \ref%
{theo exist incr sympa}\ and Corollaries \ref{coro heat production1}--\ref%
{Theorem entropy production}. Note also that, by definition,
\begin{equation*}
\mathbf{P}^{(\omega ,\mathbf{A})}\left( t\right) =\mathbf{S}^{(\omega ,%
\mathbf{A})}\left( t\right) =\mathbf{Q}^{(\omega ,\mathbf{A})}\left(
t\right) =0
\end{equation*}%
whenever $t\leq t_{0}$.
\end{proof}

Observe that the state $\rho _{t}^{(\beta ,\omega ,\lambda ,\mathbf{A})}$ of
the fermion system still evolves for $t\geq t_{1}$ when the electromagnetic
field is turned off. Indeed, for any $\beta \in \mathbb{R}^{+}$, $\omega \in
\Omega $, $\lambda \in \mathbb{R}_{0}^{+}$, $\mathbf{A}\in \mathbf{C}%
_{0}^{\infty }$ and $t\geq t_{1}$,
\begin{equation*}
\rho _{t}^{(\beta ,\omega ,\lambda ,\mathbf{A})}=\rho _{t_{1}}^{(\beta
,\omega ,\lambda ,\mathbf{A})}\circ \tau _{t-t_{1}}^{(\omega ,\lambda ,%
\mathbf{A})}\ .
\end{equation*}%
Despite that, the total heat created by the electromagnetic field stays
\emph{constant} as soon as the electromagnetic field is turned off: By
Theorem \ref{main 1 copy(1)}, $\mathbf{S}^{(\omega ,\mathbf{A})}$\ is the
heat production due to the interaction with the electromagnetic field and
from (\ref{work}) we deduce that, for all $\beta \in \mathbb{R}^{+}$, $%
\omega \in \Omega $, $\lambda \in \mathbb{R}_{0}^{+}$, $\mathbf{A}\in
\mathbf{C}_{0}^{\infty }$ and $t\geq t_{1}$,
\begin{equation*}
\mathbf{Q}^{(\omega ,\mathbf{A})}\left( t\right) =\mathbf{S}^{(\omega ,%
\mathbf{A})}\left( t\right) =\int_{t_{0}}^{t_{1}}\rho _{s}^{(\beta ,\omega
,\lambda ,\mathbf{A})}\left( \partial _{s}W_{s}^{\mathbf{A}}\right) \mathrm{d%
}s=\mathbf{S}^{(\omega ,\mathbf{A})}\left( t_{1}\right) =\mathbf{Q}^{(\omega
,\mathbf{A})}\left( t_{1}\right) \ .
\end{equation*}%
If
\begin{equation*}
\mathbf{Q}^{(\omega ,\mathbf{A})}\left( t\right) =\int_{t_{0}}^{t_{1}}\rho
_{s}^{(\beta ,\omega ,\lambda ,\mathbf{A})}\left( \partial _{s}W_{s}^{%
\mathbf{A}}\right) \mathrm{d}s>0
\end{equation*}%
for any $t\geq t_{1}$, a strictly positive amount of electromagnetic work is
absorbed by the infinite volume fermion system. We will show in a subsequent
paper that this situation (almost surely) appears for $\lambda >0$, as
expected from\ Joule's law.

For specific static potentials $V_{\omega }$ like constant ones, the heat
conduction in the infinite system still implies a dissipation of energy, or
thermalization, in the sense that, for any \emph{fixed} $L\in \mathbb{R}^{+}$%
,
\begin{equation}
\lim_{t\rightarrow \infty }\left\{ \rho _{t}^{(\beta ,\omega ,\lambda ,%
\mathbf{A})}(H_{L}^{(\omega ,\lambda )})-\varrho ^{(\beta ,\omega ,\lambda
)}(H_{L}^{(\omega ,\lambda )})\right\} =0\ .  \label{heat conduction}
\end{equation}%
The latter can be verified by explicit computations. Beside the special case
of constant potentials $V_{\omega }$, the situation is more complicated.
Indeed, the self--adjoint operator $\Delta _{\mathrm{d}}+\lambda V_{\omega }$
acting on $\ell ^{2}(\mathfrak{L})$ can have eigenvalues. In particular, the
energy $\mathbf{Q}^{(\omega ,\mathbf{A})}\left( t_{1}\right) $ for $t\geq
t_{1}$ could be stored in bound states, in contrast with the perfect
conducting case (\ref{heat conduction}). As a consequence, we can only hope
for an asymptotic version of the above result:
\begin{equation*}
\underset{\lambda \rightarrow 0}{\lim \sup }\lim_{t\rightarrow \infty
}\left\{ \rho _{t}^{(\beta ,\omega ,\lambda ,\mathbf{A})}(H_{L}^{(\omega
,\lambda )})-\varrho ^{(\beta ,\omega ,\lambda )}(H_{L}^{(\omega ,\lambda
)})\right\} =0
\end{equation*}%
for any $\beta \in \mathbb{R}^{+}$, $\omega \in \Omega $, $\mathbf{A}\in
\mathbf{C}_{0}^{\infty }$ and each $L\in \mathbb{R}^{+}$.

\begin{bemerkung}[Internal energies]
\label{Remark incr one-part-new copy(1)}\mbox{
}\newline
The internal energy as defined in \cite[Eq. (15)]{thermo2} rather
corresponds in our case to the total energy increment. Then, (\ref{work})
is, in Salem--Fr\"{o}hlich's interpretation, the expression of the 1st law
of thermodynamics. Indeed, we have a closed system which cannot exchange
heat energy with its surrounding like in \cite[Eq. (16)]{thermo2}. In their
view point, $\mathbf{P}^{(\omega ,\mathbf{A})}$ should be seen as a
Helmholtz free--energy, i.e., the available energy which can perform work.
In fact, the authors in \cite[Eq. (16)]{thermo2} focus on the heat exchanged
with the surrounding, whereas we do not consider it and concentrate our
study on the heat production within the fermion system.
\end{bemerkung}

\subsection{Heat Production at Small Electromagnetic Fields\label{Section
Heat Production as Power series}}

The physical situation we will use to investigate Joule and Ohm's laws is as
follows: We start with a macroscopic bulk containing conducting fermions.
This is idealized by taking an infinite system of non--interacting fermions
as explained above. Then, the heat production or the conductivity is
measured in a region which is very small w.r.t. the size of the bulk, but
very large w.r.t. the lattice spacing of the crystal.

We implement this hierarchy of space scales by rescaling vector potentials.
That means, for any $l\in \mathbb{R}^{+}$ and $\mathbf{A}\in \mathbf{C}%
_{0}^{\infty }$, we consider the space--rescaled vector potential
\begin{equation}
\mathbf{A}_{l}(t,x):=\mathbf{A}(t,l^{-1}x)\ ,\quad t\in \mathbb{R},\ x\in
\mathbb{R}^{d}\ .  \label{rescaled vector potential}
\end{equation}%
Then, to ensure that a macroscopic number of lattice sites is involved, we
eventually perform the limit $l\rightarrow \infty $. Indeed, the scaling
factor $l^{-1}$ used in (\ref{rescaled vector potential}) means, at fixed $l$%
, that the space scale of the electric field (\ref{V bar 0}) is
infinitesimal w.r.t. the macroscopic bulk (which is the whole space),
whereas the lattice spacing gets infinitesimal w.r.t. the space scale of the
vector potential when $l\rightarrow \infty $.

Furthermore, Ohm's law is a linear\emph{\ }response to electric fields.
Therefore, we also rescale the strength of the electromagnetic potential $%
\mathbf{A}_{l}$ by a real parameter $\eta \in \mathbb{R}$ and will
eventually take the limit $\eta \rightarrow 0$ in a subsequent paper.

In the limit $(\eta ,l^{-1})\rightarrow (0,0)$ it turns out that the heat
production $\mathbf{Q}^{(\omega ,\eta \mathbf{A}_{l})}$ or, equivalently,
the internal energy increment $\mathbf{S}^{(\omega ,\eta \mathbf{A}_{l})}$,
respectively defined by Definition \ref{Heat production definition} and (\ref%
{entropic energy increment}), are of order $\mathcal{O}\left( \eta
^{2}l^{d}\right) $. This can be understood in a physical sense by the fact that the
energy contained in the electromagnetic field, that is, its $L^{2}$--norm,
is also of order $\mathcal{O}\left( \eta ^{2}l^{d}\right) $, by classical
electrodynamics. Then, in order to get Joule and Ohm's laws, we need to give
an explicit expression for the term of order $\mathcal{O}(\eta ^{2}l^{d})$
of $\mathbf{Q}^{(\omega ,\eta \mathbf{A}_{l})}$, uniformly w.r.t. some
parameters. This is performed in Section \ref{section Energy Increments as
Power Series} by using two important tools, also used several times in
subsequent papers:

\begin{itemize}
\item A Dyson--Phillips expansion in terms of multi--commutators of the
strongly continuous two--parameter family $\{\tau _{t,s}^{(\omega ,\lambda ,%
\mathbf{A})}\}_{t\geq s}$ defined by (\ref{Cauchy problem 0}). See Section %
\ref{Section existence dynamics}.

\item Tree--decay bounds on multi--commutators as explained in Section \ref%
{section Tree--decay Bounds}.
\end{itemize}

\noindent Recall that multi--commutators are defined by induction as follows:%
\begin{equation}
{[}B_{1},B_{2}{]}^{(2)}:=[B_{1},B_{2}]:=B_{1}B_{2}-B_{2}B_{1}\ ,\qquad
B_{1},B_{2}\in \mathcal{U}\ ,  \label{multi1-0}
\end{equation}%
and, for all integers $k>2$,
\begin{equation}
{[}B_{1},B_{2},\ldots ,B_{k+1}{]}^{(k+1)}:={[}B_{1},{[}B_{2},\ldots ,B_{k+1}{%
]}^{(k)}{]}\ ,\quad B_{1},\ldots ,B_{k+1}\in \mathcal{U}\ .  \label{multi2-0}
\end{equation}

In fact, provided $\eta \in \mathbb{R}$ is sufficiently small, we get in
Section \ref{section Energy Increments as Power Series} a representation of $%
\mathbf{S}^{(\omega ,\eta \mathbf{A}_{l})}$ as a power series in $\eta $
such that all $k$--order terms in $\eta $ are of order $\mathcal{O}(l^{d})$,
as $l\rightarrow \infty $, i.e., they behave as the volume of the support of
the electromagnetic field.

\begin{satz}[Heat production at small fields]
\label{Thm Heat production as power series}\mbox{
}\newline
Let $\mathbf{A}\in \mathbf{C}_{0}^{\infty }$. Then the heat production has
the following properties: \newline
\emph{(i)} Multi--commutator series. There exists $\eta _{0}\equiv \eta _{0,%
\mathbf{A}}\in \mathbb{R}^{+}$ such that, for all $|\eta |\in \lbrack 0,\eta
_{0}]$, $l,\beta \in \mathbb{R}^{+}$, $\omega \in \Omega $, $\lambda \in
\mathbb{R}_{0}^{+}$ and $t\geq t_{0}$,%
\begin{eqnarray}
\mathbf{Q}^{(\omega ,\eta \mathbf{A}_{l})}\left( t\right)
&=&\sum\limits_{k\in {\mathbb{N}}}\sum\limits_{x,z\in \mathfrak{L},|z|\leq
1}i^{k}\langle \mathfrak{e}_{x},\left( \Delta _{\mathrm{d}}+\lambda
V_{\omega }\right) \mathfrak{e}_{x+z}\rangle \int_{t_{0}}^{t}\mathrm{d}%
s_{1}\cdots \int_{t_{0}}^{s_{k-1}}\mathrm{d}s_{k}  \notag \\
&&\varrho ^{(\beta ,\omega ,\lambda )}\left( [W_{s_{k}-t_{0},s_{k}}^{\eta
\mathbf{A}_{l}},\ldots ,W_{s_{1}-t_{0},s_{1}}^{\eta \mathbf{A}_{l}},\tau
_{t-t_{0}}^{(\omega ,\lambda )}(a_{x}^{\ast }a_{x+z})]^{(k+1)}\right)
\label{series naive}
\end{eqnarray}%
with $W_{t,s}^{\eta \mathbf{A}_{l}}:=\tau _{t}^{(\omega ,\lambda
)}(W_{s}^{\eta \mathbf{A}_{l}})\in \mathcal{U}$ for any $t,s\in \mathbb{R}$.
The above sum is absolutely convergent. \newline
\emph{(ii)} Uniform analyticity at $\eta =0$. The function $\eta \mapsto
\mathbf{Q}^{(\omega ,\eta \mathbf{A}_{l})}$ is real analytic on $\mathbb{R}$
and there exist $\eta _{1}\equiv \eta _{1,\mathbf{A}}\in \mathbb{R}^{+}$ and
$D\equiv D_{\mathbf{A}}\in \mathbb{R}^{+}$ such that, for all $l,\beta \in
\mathbb{R}^{+}$, $\omega \in \Omega $, $\lambda \in \mathbb{R}_{0}^{+}$, $%
t\geq t_{0}$ and $m\in \mathbb{N}$,
\begin{equation}
\left\vert \partial _{\eta }^{m}\mathbf{Q}^{(\omega ,\eta \mathbf{A}%
_{l})}\left( t\right) |_{\eta =0}\right\vert \leq Dl^{d}\left( \eta
_{1}^{-m}m!\right) \ .  \label{analytic bounds}
\end{equation}%
In particular, the Taylor series in $\eta $ of $l^{-d}\mathbf{Q}^{(\omega
,\eta \mathbf{A}_{l})}$ is absolutely convergent in a neighborhood of $\eta
=0$, uniformly in the parameters $l,\beta \in \mathbb{R}^{+}$, $\omega \in
\Omega $, $\lambda \in \mathbb{R}_{0}^{+}$ and $t\geq t_{0}$.
\end{satz}

\begin{proof}
To prove (i), combine Theorem \ref{main 1 copy(1)} with Equation (\ref%
{energy increment bis}). See also Lemma \ref{bound incr 1 Lemma}. The second
assertion (ii) is a direct consequence of Corollary \ref{coro heat
production1 copy(1)} and Lemma \ref{bound incr 1 Lemma copy(1)} together
with Theorem \ref{main 1 copy(1)}. Note that Lemma \ref{bound incr 1 Lemma
copy(1)} shows slightly stronger bounds than (\ref{analytic bounds}).
\end{proof}

Note that $\mathbf{Q}^{(\omega ,0)}\left( t\right) =0$ and thus, (\ref{work}%
) directly gives the estimate
\begin{equation*}
\mathbf{Q}^{(\omega ,\eta \mathbf{A}_{l})}\left( t\right) -\mathbf{Q}%
^{(\omega ,0)}\left( t\right) =\mathcal{O}(\left\vert \eta \right\vert l^{d})
\end{equation*}%
for the rest of order one of the Taylor expansion of $\mathbf{Q}^{(\omega
,\eta \mathbf{A}_{l})}$. This is a special case of Theorem \ref{Thm Heat
production as power series} (ii) which implies, for all $M\in \mathbb{N}$
and $\eta \in \lbrack 0,\eta _{1}]$, that%
\begin{equation}
\mathbf{Q}^{(\omega ,\eta \mathbf{A}_{l})}\left( t\right) -\overset{M}{%
\sum\limits_{m=1}}\frac{\eta ^{m}}{m!}\left( \partial _{\eta }^{m}\mathbf{Q}%
^{(\omega ,\eta \mathbf{A}_{l})}\left( t\right) |_{\eta =0}\right) =\mathcal{%
O}(\left\vert \eta \right\vert ^{M+1}l^{d})\ .  \label{truc}
\end{equation}

By explicit computations, the Taylor coefficients of order zero and one of
the function $\eta \mapsto \mathbf{Q}^{(\omega ,\eta \mathbf{A}_{l})}\left(
t\right) $ always vanish. Hence, using Theorem \ref{Thm Heat production as
power series} (ii), one shows that%
\begin{equation}
l^{-d}\mathbf{Q}^{(\omega ,\eta \mathbf{A}_{l})}(t)=\mathcal{O}(\eta ^{2})+%
\mathcal{O}(\left\vert \eta \right\vert ^{3})\ .  \label{this}
\end{equation}%
The term $\mathcal{O}\left( \eta ^{2}\right) $ can be made explicit whereas
the correction term of order $\mathcal{O}(\eta ^{3})$ is uniformly bounded
in $l,\beta \in \mathbb{R}^{+}$, $\omega \in \Omega $, $\lambda \in \mathbb{R%
}_{0}^{+}$ and $t\geq t_{0}$. The detailed analysis of the leading term $%
\mathcal{O}(\eta ^{2})$ is postponed to a subsequent paper.

As a consequence, for any $\beta \in \mathbb{R}^{+}$, $\omega \in \Omega $, $%
\lambda \in \mathbb{R}_{0}^{+}$, $\mathbf{A}\in \mathbf{C}_{0}^{\infty }$
and $t\in \mathbb{R}$, one can analyze the density $\mathbf{q}\equiv \mathbf{%
q}^{(\beta ,\omega ,\lambda ,\mathbf{A})}$ of heat production by the limits%
\begin{multline*}
\mathbf{q}\left( t\right) :=\underset{(\eta ,l^{-1})\rightarrow (0,0)}{\lim }%
\left\{ \left( \eta ^{2}l^{d}\right) ^{-1}\mathbf{Q}^{(\omega ,\eta \mathbf{A%
}_{l})}\left( t\right) \right\} \\
=\underset{(\eta ,l^{-1})\rightarrow (0,0)}{\lim }\left\{ \left( \eta
^{2}l^{d}\right) ^{-1}\mathbf{S}^{(\omega ,\eta \mathbf{A}_{l})}\left(
t\right) \right\} :=\mathbf{s}\left( t\right) \ ,
\end{multline*}%
see Theorem \ref{main 1 copy(1)}. This study will lead to Joule's law, which
describes the rate at which resistance in the fermion system converts
electric energy into heat energy. The details of such a study, like for
instance the existence of the above limits, are the subject of a companion
paper.

By (\ref{truc}), the density of heat production should be a real analytic
function at $\eta =0$. Hence, Theorem \ref{Thm Heat production as power
series} makes also possible the study of non--quadratic (resp. non--linear)
corrections to Joule's law (resp. Ohm's law).

\section{Tree--Decay Bounds\label{section Tree--decay Bounds}}

Remark that
\begin{equation*}
W_{t,s}^{\eta \mathbf{A}_{l}}:=\tau _{t}^{(\omega ,\lambda )}(W_{s}^{\eta
\mathbf{A}_{l}})=\mathcal{O}\left( \left\vert \eta \right\vert l^{d}\right)
\end{equation*}%
for any $t,s\in \mathbb{R}$ and $\mathbf{A}\in \mathbf{C}_{0}^{\infty }$,
see also (\ref{rescaled vector potential}). Thus, using Equation (\ref%
{series naive}), naive bounds on its r.h.s. predict that, for some constant $%
D>1$,
\begin{equation*}
\mathbf{Q}^{(\omega ,\eta \mathbf{A}_{l})}\left( t\right) =\mathcal{O}%
(D^{\left\vert \eta \right\vert l^{d}})\ .
\end{equation*}%
To obtain the much more accurate estimate
\begin{equation}
\mathbf{Q}^{(\omega ,\eta \mathbf{A}_{l})}\left( t\right) =\mathcal{O}\left(
\eta ^{2}l^{d}\right)  \label{eq la}
\end{equation}%
and to prove Theorem \ref{Thm Heat production as power series}, we need good
bounds on the multi--commu%
\-%
tators in the series (\ref{series naive}). This is achieved by using the
so--called \emph{tree--decay bounds} on the expectation of such multi--commu%
\-%
tators. Indeed, tree--decay bounds we derive here are a useful tool to
control multi--commu%
\-%
tators of products of annihilation and creation operators. This technique
will also be used many times in subsequent papers in order to derive Joule
and Ohm's laws.

Observe that (\ref{eq la}) implies thermodynamic behavior of the heat
production w.r.t. $l\in \mathbb{R}^{+}$, i.e., $\mathbf{Q}^{(\omega ,\eta
\mathbf{A}_{l})}$ is proportional to the volume $l^{d}$. This kind of issue
is well--known in statistical physics of interacting systems where cluster
or graph expansions are used to obtain such a behavior for quantities like
the free--energy or the ground--state energy at large volumes. In the
langage of construtive physics, the main result of the present section, that
is, Corollary \ref{tree bound main copy(1)}, yields the convergence of a
tree--expansion for the heat production.

The proof of Corollary \ref{tree bound main copy(1)} uses Theorem~\ref{tree
bound main} as an important ingredient. The latter is a tree--expansion for
multi--commu%
\-%
tators of monomials in annihilation and creation operators. Such kind of
combinatorial result was already used before, for instance in \cite%
{FroehlichMerkliUeltschi}. In fact, Theorem~\ref{tree bound main} is very
similar to arguments used in \cite[Section 4]{FroehlichMerkliUeltschi}.

Before going into details, let us first illustrate what will be proven in
Theorem~\ref{tree bound main}. The aim is to simplify $N$--fold multi--commu%
\-%
tators of monomials in annihilation and creation operators, as for example
\begin{equation}
\lbrack a^{\ast }(\psi _{1})a(\psi _{2})a^{\ast }(\psi _{3})a^{\ast }(\psi
_{4}),a^{\ast }(\psi _{5})a(\psi _{6}),\ldots ]^{(N)}  \label{example}
\end{equation}%
with $\psi _{1},\psi _{2},\ldots \in \ell ^{2}(\mathfrak{L})$. See (\ref%
{multi1-0})--(\ref{multi2-0}) for the precise definition of
multi--commutators. At a first glance one expects sums over monomials
containing all occurring annihilation and creation operators. Because of the
structure of the multi--commutator, there are certain terms that can be
summed up, getting then monomials containing all occurring annihilation and
creation operators except two, times the anti--commutator of those two, see (%
\ref{eq:FroehlichMerkliUeltschi}). This is useful because the
anti--commutator is a multiple of the identity, c.f.~\eqref{CAR}. This
procedure can be iterated $N-1$ times in order to reduce the number of
annihilation and creation operators in the remaining monomials. As one might
expect, only pairs of creation and annihilation operators that come from
\emph{different} entries of the multi--commutator can be removed. This is
why we consider in the following a family of trees, similar to \cite%
{FroehlichMerkliUeltschi}. The $N-1$ edges (bonds) of those trees
(containing $N$ vertices) represent the contractions of annihilation and
creation operators into anti--commutators. The vertices of such trees stand
for the $N$ entries of the $N$--fold multi--commutator.

Now, we need to introduce some notation to express the monomials in
annihilation and creation operators in a convenient way, before formulating
Theorem~\ref{tree bound main}. Each of the entries of the $N$--fold
multi--commutator is a product of annihilation and creation operators, which
we characterize by certain finite index sets $\bar{\Lambda}_{1},\Lambda
_{1},\ldots ,\bar{\Lambda}_{N},\Lambda _{N}\subset \mathbb{N}$, where the
set $\bar{\Lambda}_{i}$ refers to creation operators in entry $i$ and $%
\Lambda _{i}$ to annihilation operators in the same entry. For example, we
choose for (\ref{example}) the sets%
\begin{equation}
\bar{\Lambda}_{1}=\{1,3,4\},\quad \Lambda _{1}=\{2\},\quad \bar{\Lambda}%
_{2}=\{5\},\quad \Lambda _{2}=\{6\},\ \ldots  \label{examplebis}
\end{equation}%
The kind of products we are interested in allows us to restrict our
considerations to index sets $\bar{\Lambda}_{1},\Lambda _{1},\ldots ,\bar{%
\Lambda}_{N},\Lambda _{N}\subset \mathbb{N}$ that are non--empty, mutually
disjoint and such that
\begin{equation*}
\left\vert \bar{\Lambda}_{j}\right\vert +\left\vert \Lambda _{j}\right\vert
:=2n_{j}\in 2{\mathbb{N}}\ ,
\end{equation*}%
for all $j\in \{1,\ldots ,N\}$. Hence, each entry in the multi--commutator
contains an even number of annihilation and creation operators. To shorten
the notation we set
\begin{equation*}
\Omega _{j}:=(\{+\}\times \bar{\Lambda}_{j})\cup (\{-\}\times \Lambda _{j})\
,
\end{equation*}%
for all $j\in \{1,\ldots ,N\}$. To determine the position of annihilation
and creation operators in the monomial of the $j$th entry we choose a
numbering of $\Omega _{j}$, that is, a bijective map%
\begin{equation}
\pi _{j}:\{1,\ldots ,2n_{j}\}\rightarrow \Omega _{j}\ .  \label{property pij}
\end{equation}%
In the example (\ref{example})--(\ref{examplebis}),
\begin{equation*}
\Omega _{1}=\{(+,1),(+,3),(+,4),(-,2)\}
\end{equation*}%
and its numbering is defined by
\begin{equation*}
\pi _{1}\left( 1\right) =(+,1),\ \pi _{1}\left( 2\right) =(-,2),\ \pi
_{1}\left( 3\right) =(+,3),\ \pi _{1}\left( 4\right) =(+,4)\ .
\end{equation*}%
Furthermore, for all $x\in \bigcup_{j=1}^{N}\bar{\Lambda}_{j}\cup \Lambda
_{j}$, let $\psi _{x}\in \ell ^{2}(\mathfrak{L})$ be the corresponding wave
function and denote (only in this section) the annihilation and creation operators respectively by
\begin{equation*}
a(-,x):=a(\psi _{x})\qquad \text{and}\qquad a(+,x):=a^{\ast }(\psi _{x})\ .
\end{equation*}%
Using this notation,
we then define the monomials
\begin{equation}
\mathfrak{p}_{j}:=\prod_{k=1}^{2n_{j}}a(\pi _{j}(k))  \label{pj def}
\end{equation}%
in $a(\pm ,x)$ for all $j\in \{1,\ldots ,N\}$. Recall that $\mathfrak{p}_{j}$
is the $j$th entry in the $N$--fold multi--commutator.

To formulate Theorem~\ref{tree bound main}, we need two more things. Recall
that a tree is a connected graph that has no loops. Here, we have a finite
number of labeled vertices, denoted by $1,\ldots ,N$, and (non--oriented)
bonds between these vertices. For example, the bond connecting vertices $i$
and $j$ is denoted by $\{i,j\}=\{j,i\}$. A tree is characterized by the set
of its $N-1$ bonds. The family of trees we use is defined as follows: Let $%
\mathcal{T}_{2}$ be the set of all trees with exactly two vertices. This set
contains a unique tree $T=\{\{1,2\}\}$ which, in turn, contains the unique
bond $\{1,2\}$, i.e., $\mathcal{T}_{2}:=\left\{ \left\{ \{1,2\}\right\}
\right\} $. Then, for each integer $N\geq 3$, we recursively define the set $%
\mathcal{T}_{N}$ of trees with $N$ vertices by
\begin{equation}
\mathcal{T}_{N}:=\Big\{\{\{k,N\}\}\cup T\text{ }:\text{ }k=1,\ldots
,N-1,\quad T\in \mathcal{T}_{N-1}\Big\}\ .  \label{def.tree}
\end{equation}%
In other words, $\mathcal{T}_{N}$ is the set of all trees with vertex set $%
\mathcal{V}_{N}:=\{1,\ldots N\}$ for which $N\in \mathcal{V}_{N}$ is a leaf,
and if the leaf $N$ is removed, the vertex $N-1$ is a leaf in the remaining
tree and so on.

Now, for every tree $T\in \mathcal{T}_{N}$, we define maps $\mathbf{x},%
\mathbf{y}:T\rightarrow \bigcup_{j=1}^{N}\Omega _{j}$ that choose, for each
bond $\{i,j\}\in T$, a point in the set $\Omega _{i}$ and one point in the
set $\Omega _{j}$, respectively. More precisely, we assume for $i<j$ that $%
\mathbf{x}(\{i,j\})\in \Omega _{i}$ and $\mathbf{y}(\{i,j\})\in \Omega _{j}$%
. The induced orientation of the bond is completely arbitrary, because of
the symmetry of anti--commutators. The set of all those maps is given by%
\begin{eqnarray*}
\mathcal{K}_{T} &:=&\left\{ \left( \mathbf{x},\mathbf{y}\right) \text{ }|\text{ }\mathbf{x},\mathbf{y}:T\rightarrow \cup _{j=1}^{N}\Omega _{j}\text{ }\right. \\
&&\text{ }\left. \text{with }\mathbf{x}(b)\in \Omega _{i},\ \mathbf{y}(b)\in
\Omega _{j}\text{ for }b=\left\{ i,j\right\} \in T,\ i<j\right\} \ .
\end{eqnarray*}%

We are finally ready to express a $N$--fold multi--commutators of products
of annihilation and creation operators as a sum over trees $T\in \mathcal{T}%
_{N}$ of monomials in annihilation and creation operators:

\begin{satz}[Multi--commutators as sum over trees]
\label{tree bound main}\mbox{}\newline
Let $N\geq 2$. Then, for all $T\in \mathcal{T}_{N}$ and $(\mathbf{x},\mathbf{%
y})\in \mathcal{K}_{T}$, there are constants
\begin{equation*}
\mathrm{m}_{T}(\mathbf{x},\mathbf{y})\in \{-1,0,1\}
\end{equation*}%
and injective maps
\begin{equation*}
\pi _{T}\left( \mathbf{x},\mathbf{y}\right) :\left\{ 1,2,\ldots ,2\overline{%
\mathrm{N}}\right\} \rightarrow \bigcup_{j=1}^{N}\Omega _{j}\backslash
\left( \mathbf{x}(T)\cup \mathbf{y}(T)\right)
\end{equation*}%
where $\overline{\mathrm{N}}:=\sum_{j=1}^{N}n_{j}-(N-1)\geq 1$, such that
\begin{eqnarray}
\left[ \mathfrak{p}_{N},\ldots ,\mathfrak{p}_{1}\right] ^{(N)} &=&\sum_{T\in
\mathcal{T}_{N}}\sum_{(\mathbf{x},\mathbf{y})\in \mathcal{K}_{T}}\mathrm{m}%
_{T}\left( \mathbf{x},\mathbf{y}\right) \mathfrak{p}_{T}\left( \mathbf{x},%
\mathbf{y}\right) \prod_{b\in T}\left\{ a\left( \mathbf{x}(b)\right)
,a\left( \mathbf{y}(b)\right) \right\} \ ,  \notag \\
&&  \label{196 bis}
\end{eqnarray}%
with $\{B_{1},B_{2}\}:=B_{1}B_{2}+B_{2}B_{1}$ being the anti--commutator of $%
B_{1},B_{2}\in \mathcal{U}$ and%
\begin{equation*}
\mathfrak{p}_{T}(\mathbf{x},\mathbf{y}):=\prod_{k=1}^{2\overline{\mathrm{N}}%
}a(\pi _{T}(\mathbf{x},\mathbf{y})(k))\ .
\end{equation*}
\end{satz}

\begin{proof}
We first observe that, for any integers $n_{1},n_{2}\in \mathbb{N}$ and all
elements $B_{1},\ldots ,B_{2n_{2}}\in \mathcal{U}$ and $\tilde{B}_{1},\ldots
,\tilde{B}_{2n_{1}}\in \mathcal{U}$,
\begin{eqnarray}
&&\left[ B_{1}\ldots B_{2n_{2}},\tilde{B}_{1}\ldots \tilde{B}_{2n_{1}}\right]
\label{eq:FroehlichMerkliUeltschi} \\
&=&\underset{{1\leq k_{1}\leq }2n_{1}}{\underset{{1\leq k_{2}\leq }2n_{2}}{%
\sum }}(-1)^{k_{1}+1}B_{1}\ldots B_{k_{2}-1}\tilde{B}_{1}\ldots \tilde{B}%
_{k_{1}-1}  \notag \\
&&\hspace{1cm}\hspace{1cm}\quad \times \{B_{k_{2}},\tilde{B}_{k_{1}}\}\tilde{%
B}_{k_{1}+1}\ldots \tilde{B}_{2n_{1}}B_{k_{2}+1}\ldots B_{2n_{2}}\ ,  \notag
\end{eqnarray}%
see \cite[Eq. (4.18)]{FroehlichMerkliUeltschi}. Note also that, for ${k_{2}=1%
}$, one obtains
\begin{equation*}
(-1)^{k_{1}+1}B_{1}\ldots B_{1-1}\tilde{B}_{1}\ldots \tilde{B}%
_{k_{1}-1}\{B_{1},\tilde{B}_{k_{1}}\}\tilde{B}_{k_{1}+1}\ldots \tilde{B}%
_{2n_{1}}B_{2}\ldots B_{2n_{2}}\ .
\end{equation*}%
This has to be understood of course as
\begin{equation*}
(-1)^{k_{1}+1}\tilde{B}_{1}\ldots \tilde{B}_{k_{1}-1}\{B_{1},\tilde{B}%
_{k_{1}}\}\tilde{B}_{k_{1}+1}\ldots \tilde{B}_{2n_{1}}B_{2}\ldots
B_{2n_{2}}\ .
\end{equation*}%
Similar remarks can be done for the cases ${k_{1}=1,}2n_{1}$ and ${k_{2}=}%
2n_{2}$. We now prove the assertion by induction.

For $N=2$, the set $\mathcal{T}_{2}:=\left\{ \{\{1,2\}\}\right\} $ consists
of only one tree $T=\{\{1,2\}\}$. Using (\ref{pj def}) and (\ref%
{eq:FroehlichMerkliUeltschi}) we get%
\begin{align}
\left[ \mathfrak{p}_{2},\mathfrak{p}_{1}\right] & =\underset{{1\leq
k_{1}\leq }2n_{1}}{\underset{{1\leq k_{2}\leq }2n_{2}}{\sum }}%
(-1)^{k_{1}+1}a(\pi _{2}(1))\ldots a(\pi _{2}(k_{2}-1))a(\pi _{1}(1))\ldots
a(\pi _{1}(k_{1}-1))  \notag \\
& \hspace{1cm}\hspace{1cm}\quad \times \{a(\pi _{2}(k_{2})),a(\pi
_{1}(k_{1}))\}a(\pi _{1}(k_{1}+1))\ldots a(\pi _{1}(2n_{1}))  \notag \\
& \hspace{1cm}\hspace{1cm}\hspace{1cm}\hspace{1cm}\quad \times a(\pi
_{2}(k_{2}+1))\ldots a(\pi _{2}(2n_{2}))\ .  \label{197 bis}
\end{align}%
Note that $\{a(\pi _{2}(k_{2})),a(\pi _{1}(k_{1}))\}$ is always a multiple
of the identity in $\mathcal{U}$, see (\ref{CAR}) and (\ref{creation
operators}). Therefore, the assertion for $N=2$ directly follows from the
previous equality by observing that the sum over $k_{1}$ and $k_{2}$ in (\ref%
{197 bis}) corresponds to the sum over $(\mathbf{x},\mathbf{y})\in \mathcal{K%
}_{\{\{1,2\}\}}$ in (\ref{196 bis}) by choosing%
\begin{eqnarray}
\mathfrak{p}_{\{\{1,2\}\}}(\mathbf{x},\mathbf{y}) &:=&a(\pi _{2}(1))\ldots
a(\pi _{2}(k_{2}-1))a(\pi _{1}(1))\ldots a(\pi _{1}(k_{1}-1))  \label{pt} \\
&&\quad \times \ a(\pi _{1}(k_{1}+1))\ldots a(\pi _{1}(2n_{1}))a(\pi
_{2}(k_{2}+1))\ldots a(\pi _{2}(2n_{2}))\   \notag
\end{eqnarray}%
for%
\begin{eqnarray*}
\mathbf{x}(\{1,2\}) &=&\pi _{1}(k_{1})\in \Omega _{1}\ ,\qquad \ {k_{1}\in }%
\left\{ 1,\ldots ,2n_{1}\right\} \ , \\
\mathbf{y}(\{1,2\}) &=&\pi _{2}(k_{2})\in \Omega _{2}\ ,\qquad \ {k_{2}\in }%
\left\{ 1,\ldots ,2n_{2}\right\} \ .
\end{eqnarray*}%
Indeed, for $\left( \mathbf{x},\mathbf{y}\right) \in \mathcal{K}%
_{\{\{1,2\}\}}$ as above, the constant $\mathrm{m}_{\{\{1,2\}\}}(\mathbf{x},%
\mathbf{y})$ equals $(-1)^{k_{1}+1}\in \{-1,1\}\subset \{-1,0,1\}$, whereas
the associated map
\begin{equation*}
\pi _{\{\{1,2\}\}}\left( \mathbf{x},\mathbf{y}\right) :\left\{ 1,2,\ldots ,2%
\overline{\mathrm{N}}\right\} \rightarrow \Omega _{1}\cup \Omega
_{2}\backslash \left( \mathbf{x}(\{\{1,2\}\})\cup \mathbf{y}%
(\{\{1,2\}\})\right)
\end{equation*}%
with
\begin{equation*}
\overline{\mathrm{N}}:=\left( n_{1}+n_{2}\right) -1\geq 1
\end{equation*}%
depends on the order of the factors in the r.h.s. of (\ref{pt}):%
\begin{equation*}
\pi _{\{\{1,2\}\}}\left( \mathbf{x},\mathbf{y}\right) \left( k\right)
:=\left\{
\begin{array}{lll}
\pi _{2}(k) & , & {k\in }\left\{ 1,2,\ldots ,k_{2}-1\right\} \ . \\
\pi _{1}(k-k_{2}+1) & , & {k\in }\left\{ k_{2},\ldots ,k_{2}+k_{1}-2\right\}
\ . \\
\pi _{1}(k-k_{2}+2) & , & {k\in }\left\{ k_{2}+k_{1}-1,\ldots
,2n_{1}-2+k_{2}\right\} \ . \\
\pi _{2}(k-2n_{1}+2) & , & {k\in }\left\{ 2n_{1}-2+k_{2}+1,\ldots ,2%
\overline{\mathrm{N}}\right\} \ .%
\end{array}%
\right.
\end{equation*}

We assume now that the assertion holds for some fixed integer $N\geq 2$.
Recall that $N$--fold multi--commutators are defined by (\ref{multi1-0})--(%
\ref{multi2-0}). In particular,
\begin{equation*}
{[\mathfrak{p}_{N+1},\ldots ,\mathfrak{p}_{1}]}^{(N+1)}={[\mathfrak{p}_{N+1}}%
,{[\mathfrak{p}_{N}},\ldots ,\mathfrak{p}_{1}{]}^{(N)}{]}
\end{equation*}%
where, by assumption,
\begin{equation*}
{[\mathfrak{p}_{N}},\ldots ,\mathfrak{p}_{1}{]}^{(N)}=\sum\limits_{T\in
\mathcal{T}_{N}}\sum\limits_{(\mathbf{x},\mathbf{y})\in \mathcal{K}_{T}}%
\mathrm{m}_{T}\left( \mathbf{x},\mathbf{y}\right) \mathfrak{p}_{T}\left(
\mathbf{x},\mathbf{y}\right) \prod\limits_{b\in T}\left\{ a\left( \mathbf{x}%
(b)\right) ,a\left( \mathbf{y}(b)\right) \right\} \ ,
\end{equation*}%
as stated in the theorem. Therefore,%
\begin{eqnarray}
{[\mathfrak{p}_{N+1},\ldots ,\mathfrak{p}_{1}]}^{(N+1)} &=&\sum\limits_{T\in
\mathcal{T}_{N}}\sum\limits_{(\mathbf{x},\mathbf{y})\in \mathcal{K}_{T}}%
\mathrm{m}_{T}\left( \mathbf{x},\mathbf{y}\right) {[\mathfrak{p}_{N+1}},%
\mathfrak{p}_{T}\left( \mathbf{x},\mathbf{y}\right) ]  \notag \\
&&\hspace{1cm}\hspace{1cm}\quad \times \prod\limits_{b\in T}\left\{ a\left(
\mathbf{x}(b)\right) ,a\left( \mathbf{y}(b)\right) \right\} \ ,
\label{inductionN1}
\end{eqnarray}%
whereas, using again (\ref{eq:FroehlichMerkliUeltschi}),%
\begin{align}
{[\mathfrak{p}_{N+1}},\mathfrak{p}_{T}\left( \mathbf{x},\mathbf{y}\right) ]&
=\underset{{1\leq k_{1}\leq 2\overline{\mathrm{N}}}}{\underset{{1\leq
k_{2}\leq }2n_{N+1}}{\sum }}(-1)^{k_{1}+1}a(\pi _{N+1}(1))\cdots a(\pi
_{N+1}(k_{2}-1))  \notag \\
& \hspace{1cm}\hspace{1cm}\quad \times a(\pi _{T}(1))\cdots a(\pi
_{T}(k_{1}-1))  \notag \\
& \hspace{1cm}\hspace{1cm}\quad \times a(\pi _{T}(k_{1}+1))\cdots a(\pi _{T}(%
{2\overline{\mathrm{N}}}))  \notag \\
& \hspace{1cm}\hspace{1cm}\quad \times a(\pi _{N+1}(k_{2}+1))\cdots a(\pi
_{N+1}(2n_{N+1}))\   \notag \\
& \hspace{1cm}\hspace{1cm}\quad \times \{a(\pi _{N+1}(k_{2})),a(\pi
_{T}(k_{1}))\}\ .  \label{inductionN2}
\end{align}%
Note that, for simplicity, we sometimes use (as above) the notation $\pi
_{T}\equiv \pi _{T}\left( \mathbf{x},\mathbf{y}\right) $. To get now the
assertion for $\left( N+1\right) $--fold multi--commutators, for any $(%
\mathbf{x},\mathbf{y})\in \mathcal{K}_{T}$, we define:%
\begin{eqnarray*}
X &:=&\pi _{T}(k_{1})\in \bigcup\limits_{j=1}^{N}\Omega _{j}\backslash
\left( \mathbf{x}(T)\cup \mathbf{y}(T)\right) \ ,\qquad \ {k_{1}\in }\left\{
1,\ldots ,{2\overline{\mathrm{N}}}\right\} \ , \\
Y &:=&\pi _{N+1}(k_{2})\in \Omega _{N+1}\ ,\qquad \qquad \ \ \ \ \ \ \ \ \ \
\ \ \ \ \ \ {k_{2}\in }\left\{ 1,\ldots ,2n_{N+1}\right\} \ ,
\end{eqnarray*}%
as well as%
\begin{equation*}
\mathrm{\tilde{m}}_{T}\left( X,Y\right) :=(-1)^{k_{1}+1}
\end{equation*}%
and%
\begin{eqnarray*}
\mathfrak{\tilde{p}}_{T}(\mathbf{x},\mathbf{y},X,Y) &:=&a(\pi
_{N+1}(1))\cdots a(\pi _{N+1}(k_{2}-1)) \\
&&\quad \times \ a(\pi _{T}(1))\cdots a(\pi _{T}(k_{1}-1))a(\pi
_{T}(k_{1}+1))\cdots a(\pi _{T}({2\overline{\mathrm{N}}})) \\
&&\quad \times \ a(\pi _{N+1}(k_{2}+1))\cdots a(\pi _{N+1}(2n_{N+1}))\ .
\end{eqnarray*}%
Then, by (\ref{inductionN1})--(\ref{inductionN2}), one has%
\begin{multline*}
{[\mathfrak{p}_{N+1},\ldots ,\mathfrak{p}_{1}]}^{(N+1)}=\sum\limits_{T\in
\mathcal{T}_{N}}\sum\limits_{(\mathbf{x},\mathbf{y})\in \mathcal{K}_{T}}%
\underset{X\in \left( \Omega _{1}\cup \cdots \cup \Omega _{N}\right)
\backslash \left( \mathbf{x}(T)\cup \mathbf{y}(T)\right) }{\sum }\ \underset{%
Y\in \Omega _{N+1}}{\sum } \\
\mathrm{m}_{T}\left( \mathbf{x},\mathbf{y}\right) \mathrm{\tilde{m}}%
_{T}\left( X,Y\right) \mathfrak{\tilde{p}}_{T}(\mathbf{x},\mathbf{y}%
,X,Y)\{a(X),a(Y)\}\prod\limits_{b\in T}\left\{ a\left( \mathbf{x}(b)\right)
,a\left( \mathbf{y}(b)\right) \right\} \ .
\end{multline*}%
This last equation can clearly be rewritten as
\begin{align}
& {[\mathfrak{p}_{N+1},\ldots ,\mathfrak{p}_{1}]}^{(N+1)}
\label{inductionN3} \\
& =\sum\limits_{T\in \mathcal{T}_{N}}\sum\limits_{(\mathbf{x},\mathbf{y})\in
\mathcal{K}_{T}}\underset{k\in \left\{ 1,{\ldots ,N}\right\} }{\sum }\
\underset{X_{{\{k,N+1\}}}\in \Omega _{k}}{\sum }\ \underset{Y_{{\{k,N+1\}}%
}\in \Omega _{N+1}}{\sum }  \notag \\
& \quad \quad \mathbf{1}\left[ X_{{\{k,N+1\}}}\notin \left( \mathbf{x}%
(T)\cup \mathbf{y}(T)\right) \right] \mathrm{m}_{T}\left( \mathbf{x},\mathbf{%
y}\right) \mathrm{\tilde{m}}_{T}\left( X_{{\{k,N+1\}}},Y_{{\{k,N+1\}}}\right)
\notag \\
& \quad \times \ \mathfrak{\tilde{p}}_{T}(\mathbf{x},\mathbf{y},X_{{\{k,N+1\}%
}},Y_{{\{k,N+1\}}})  \notag \\
& \quad \times \ \{a(X_{{\{k,N+1\}}}),a(Y_{{\{k,N+1\}}})\}\prod\limits_{b\in
T}\left\{ a\left( \mathbf{x}(b)\right) ,a\left( \mathbf{y}(b)\right)
\right\} \ .  \notag
\end{align}%
Since $\Omega _{j}$, $j\in \{1,\ldots ,N\}$, are, by definition, mutually
disjoint sets, the latter yields the assertion for the $\left( N+1\right) $%
--fold multi--commutator. Indeed, one only needs to define, for any tree $%
T\in \mathcal{T}_{N+1}$ with $N+1$ vertices and fixed $(\mathbf{x},\mathbf{y}%
)\in \mathcal{K}_{T}$, an appropriate constant $\mathrm{m}_{T}(\mathbf{x},%
\mathbf{y})\in \{-1,0,1\}$ and map $\pi _{T}\left( \mathbf{x},\mathbf{y}%
\right) $. This can directly be deduced from (\ref{def.tree}) and (\ref%
{inductionN3}) and we omit the details.
\end{proof}

Because of (\ref{inductionN3}) note that, for any $N\geq 2$, all $T\in
\mathcal{T}_{N}$ and $(\mathbf{x},\mathbf{y})\in \mathcal{K}_{T}$, the
constants $\mathrm{m}_{T}(\mathbf{x},\mathbf{y})$ of Theorem \ref{tree bound
main} satisfy $\mathrm{m}_{T}\left( \mathbf{x},\mathbf{y}\right) =0$
whenever
\begin{equation*}
\left\vert \mathbf{x}(T)\right\vert +\left\vert \mathbf{y}(T)\right\vert
<2(N-1)\ .
\end{equation*}%
Similar to $\{\pi _{j}\}_{j\in \{1,\ldots ,N\}}$ (see (\ref{property pij})),
the maps $\pi _{T}(\mathbf{x},\mathbf{y})$ are (injective) numberings:%
\begin{equation*}
\begin{array}{c}
\big \{%
x\;:\;\pi _{T}(\mathbf{x},\mathbf{y})(k)=(+,x),\ \ \text{for }k\in
\{1,\ldots ,2\overline{\mathrm{N}}\}%
\big \}%
=\bigcup_{j=1}^{N}\bar{\Lambda}_{j}\backslash \bar{\Lambda}_{\mathbf{x},%
\mathbf{y}}\ , \\
\big \{%
x\;:\;\pi _{T}(\mathbf{x},\mathbf{y})(k)=(-,x),\ \ \text{for }k\in
\{1,\ldots ,2\overline{\mathrm{N}}\}%
\big \}%
=\bigcup_{j=1}^{N}\Lambda _{j}\backslash \Lambda _{\mathbf{x},\mathbf{y}}\ ,%
\end{array}%
\end{equation*}%
where, for any $T\in \mathcal{T}_{N}$ and $(\mathbf{x},\mathbf{y})\in
\mathcal{K}_{T}$,%
\begin{equation*}
\begin{array}{c}
\Lambda _{\mathbf{x},\mathbf{y}}:=%
\big \{%
z\,\in \mathfrak{L}\;:\;(-,z)\in \{\mathbf{x}(b),\mathbf{y}(b)\}\;\ \text{%
for some}\;b\in T%
\big \}%
\ , \\
\bar{\Lambda}_{\mathbf{x},\mathbf{y}}:=%
\big \{%
z\,\in \mathfrak{L}\;:\;(+,z)\in \{\mathbf{x}(b),\mathbf{y}(b)\}\;\ \text{%
for some}\;b\in T%
\big \}%
\ .%
\end{array}%
\end{equation*}

We conclude this section by the notion of tree--decay bounds: Let $\rho \in
\mathcal{U}^{\ast }$\ be any state and $\tau \equiv \{\tau _{t}\}_{t\in {%
\mathbb{R}}}$ be any one--parameter group of automorphisms on the $C^{\ast }$%
--algebra $\mathcal{U}$. We say that $(\rho ,\tau )$ satisfies \emph{%
tree--decay bounds} with parameters $\epsilon \in \mathbb{R}^{+}$ and $%
t_{0},t\in \mathbb{R}$, $t_{0}<t$, if there is a finite constant $D\in
\mathbb{R}^{+}$ such that, for any integer $N\geq 2$, $s_{1},\ldots
,s_{N}\in \lbrack t_{0},t]$, $x_{1},\ldots ,x_{N}\in \mathfrak{L}$ and all $%
z_{1},\ldots ,z_{N}\in \mathfrak{L}$ satisfying $|z_{i}|=1$ for $i\in
\{1,\ldots ,N\}$,%
\begin{equation}
\left\vert \rho \left( \left[ \tau _{s_{1}}(a_{x_{1}}^{\ast
}a_{x_{1}+z_{1}}),\ldots ,\tau _{s_{N}}(a_{x_{N}}^{\ast }a_{x_{N}+z_{N}})%
\right] ^{(N)}\right) \right\vert \leq D^{N-1}\mathbf{v}_{N}^{(\epsilon
)}\left( x_{1},\ldots ,x_{N}\right) \ ,  \label{tree bound 2}
\end{equation}%
where%
\begin{equation*}
\mathbf{v}_{N}^{(\epsilon )}\left( x_{1},\ldots ,x_{N}\right)
=\sum\limits_{T\in \mathcal{T}_{N}}\prod\limits_{\{k,l\}\in T}\frac{1}{%
1+|x_{k}-x_{l}|^{d+\epsilon }}\ ,\qquad x_{1},\ldots ,x_{N}\in \mathfrak{L}\
.
\end{equation*}%
(Recall that $\mathfrak{L}:=\mathbb{Z}^{d}$ with $d\in \mathbb{N}$.)

Such a property is used in Section \ref{section Energy Increments as Power
Series} and will be exploited many times in the subsequent papers for $\tau
=\tau ^{(\omega ,\lambda )}$ and $\rho =\varrho ^{(\beta ,\omega ,\lambda )}$
with $\beta \in \mathbb{R}^{+}$, $\omega \in \Omega $ and $\lambda \in
\mathbb{R}_{0}^{+}$. In fact, using Theorem \ref{tree bound main} we show
below that the one--parameter Bogoliubov group $\tau ^{(\omega ,\lambda )}$
of automorphisms defined by (\ref{rescaledbis}) and any state $\rho $
satisfy tree--decay bounds. Indeed, observe first the following elementary
lemma:

\begin{lemma}[Correlation decays]
\label{bound anticomm}\mbox{
}\newline
For any $T,\epsilon \in \mathbb{R}^{+}$, there is a finite constant $D\in
\mathbb{R}^{+}$ such that%
\begin{equation*}
\left\vert \left\langle \mathfrak{e}_{x},\mathrm{e}^{it(\Delta _{\mathrm{d}%
}+\lambda V_{\omega })}\mathfrak{e}_{y}\right\rangle \right\vert \leq \frac{D%
}{1+|x-y|^{d+\epsilon }}
\end{equation*}%
for all $\omega \in \Omega $, $\lambda \in \mathbb{R}_{0}^{+}$, $t\in
\lbrack -T,T\mathbb{]}$ and $x,y\in \mathfrak{L}$. Recall that $\left\{
\mathfrak{e}_{x}\right\} _{x\in \mathfrak{L}}$ is the canonical orthonormal
basis of $\ell ^{2}(\mathfrak{L})$ defined by $\mathfrak{e}_{x}(y)\equiv
\delta _{x,y}$ for all $x,y\in \mathfrak{L}$.
\end{lemma}

\begin{proof}
Let $\omega \in \Omega $, $\lambda \in \mathbb{R}_{0}^{+}$, $t\in \mathbb{R}$
and $x,y\in \mathfrak{L}$. Using the Trotter--Kato formula and the canonical
orthonormal basis $\left\{ \mathfrak{e}_{x}\right\} _{x\in \mathfrak{L}}$ of
$\ell ^{2}(\mathfrak{L})$ we first observe that
\begin{eqnarray}
\left\langle \mathfrak{e}_{x},\mathrm{e}^{it(\Delta _{\mathrm{d}}+\lambda
V_{\omega })}\mathfrak{e}_{y}\right\rangle &=&\lim_{m\rightarrow \infty
}\left\langle \mathfrak{e}_{x},\left[ \mathrm{e}^{\frac{it}{m}\Delta _{%
\mathrm{d}}}\mathrm{e}^{\frac{it}{m}\lambda V_{\omega }}\right] ^{m}%
\mathfrak{e}_{y}\right\rangle  \label{eq 1 ptt lemme} \\
&=&\lim_{m\rightarrow \infty }\lim_{L\rightarrow \infty
}\sum\limits_{x_{1},\ldots ,x_{m-1}\in \Lambda _{L}}\left\langle \mathfrak{e}%
_{x},\mathrm{e}^{\frac{it}{m}\Delta _{\mathrm{d}}}\mathfrak{e}%
_{x_{1}}\right\rangle \cdots  \notag \\
&&\left\langle \mathfrak{e}_{x_{m-1}},\mathrm{e}^{%
\frac{it}{m}\Delta _{\mathrm{d}}}\mathfrak{e}_{y}\right\rangle \times \mathrm{e}^{\frac{it}{m}\lambda V_{\omega }\left( x_{1}\right)
}\times \cdots \times \mathrm{e}^{\frac{it}{m}\lambda V_{\omega }\left(
y\right) }\ ,  \notag
\end{eqnarray}%
where $\Lambda _{L}$ is the finite box (\ref{eq:def lambda n}) of side
length $2[L]+1$ for $L\in \mathbb{R}^{+}$. Writing now the exponential $%
\mathrm{e}^{\frac{it}{m}\Delta _{\mathrm{d}}}$\ as a power series and using
the definition (\ref{discrete laplacian}) of the discrete Laplacian $\Delta
_{\mathrm{d}}$ we arrive at the upper bound
\begin{equation}
\left\vert \left\langle \mathfrak{e}_{x},\mathrm{e}^{\frac{it}{m}\Delta _{%
\mathrm{d}}}\mathfrak{e}_{y}\right\rangle \right\vert \leq \mathrm{e}^{\frac{%
4d\left\vert t\right\vert }{m}}\left\langle \mathfrak{e}_{x},\mathrm{e}^{-%
\frac{\left\vert t\right\vert }{m}\Delta _{\mathrm{d}}}\mathfrak{e}%
_{y}\right\rangle \ ,\qquad x,y\in \mathfrak{L}\ ,\ t\in \mathbb{R},\ m\in
\mathbb{N}\ .  \label{eq 2 ptt lemme}
\end{equation}%
Therefore, we infer from (\ref{eq 1 ptt lemme})--(\ref{eq 2 ptt lemme}) that%
\begin{equation}
\left\vert \left\langle \mathfrak{e}_{x},\mathrm{e}^{it(\Delta _{\mathrm{d}%
}+\lambda V_{\omega })}\mathfrak{e}_{y}\right\rangle \right\vert \leq
\mathrm{e}^{4d\left\vert t\right\vert }\left\langle \mathfrak{e}_{x},\mathrm{%
e}^{-\left\vert t\right\vert \Delta _{\mathrm{d}}}\mathfrak{e}%
_{y}\right\rangle \ .  \label{ineaulity iodiot}
\end{equation}%
Note that $\Delta _{\mathrm{d}}$ is explicitly given in Fourier space by the
dispersion relation%
\begin{equation*}
E(p):=2\left[ d-\left( \cos (p_{1})+\cdots +\cos (p_{d})\right) \right] \
,\qquad p\in \left[ -\pi ,\pi \right] ^{d}\ .
\end{equation*}%
Thus, explicit computations show that, for all $s\in \mathbb{R}$,%
\begin{equation*}
\left\langle \mathfrak{e}_{x},\mathrm{e}^{s\Delta _{\mathrm{d}}}\mathfrak{e}%
_{y}\right\rangle =\frac{1}{(2\pi )^{d}}\int\nolimits_{[-\pi ,\pi ]^{d}}%
\mathrm{e}^{sE(p)-ip\cdot (x-y)}\mathrm{d}^{d}p\ ,
\end{equation*}%
which, combined with (\ref{ineaulity iodiot}), implies the assertion.
\end{proof}

By (\ref{CAR}) and (\ref{rescaledbis}),
\begin{equation}
\left\Vert \{\tau _{s_{1}}^{(\omega ,\lambda )}(a_{x}^{\ast }),\tau
_{s_{2}}^{(\omega ,\lambda )}(a_{y})\}\right\Vert =\left\vert \left\langle
\mathfrak{e}_{x},\mathrm{e}^{i\left( s_{2}-s_{1}\right) (\Delta _{\mathrm{d}%
}+\lambda V_{\omega })}\mathfrak{e}_{y}\right\rangle \right\vert
\label{bound anticommbis}
\end{equation}%
for every $s_{1},s_{2}\in \mathbb{R}$, $x,y\in \mathfrak{L}$, $\omega \in
\Omega $ and $\lambda \in \mathbb{R}_{0}^{+}$. Hence, for any $\epsilon \in
\mathbb{R}^{+}$ and $t_{0},t\in \mathbb{R}$, $t_{0}<t$, we infer from Lemma %
\ref{bound anticomm} the existence of a finite constant $D\in \mathbb{R}^{+}$
(only depending on $\epsilon ,t_{0},t$) such that
\begin{equation}
\left\Vert \left\{ \tau _{s_{1}}^{(\omega ,\lambda )}\left( a_{x}^{\ast
}\right) ,\tau _{s_{2}}^{(\omega ,\lambda )}\left( a_{y}\right) \right\}
\right\Vert \leq \frac{D}{1+\left\vert x-y\right\vert ^{d+\epsilon }}
\label{eq:anticomm bound 1}
\end{equation}%
for all $s_{1},s_{2}\in \lbrack t_{0},t]$, $x,y\in \mathfrak{L}$, $\omega
\in \Omega $ and $\lambda \in \mathbb{R}_{0}^{+}$. Using this and Theorem %
\ref{tree bound main} we obtain (\ref{tree bound 2}) with a uniform constant
$D<\infty $ not depending on $\omega \in \Omega $ and $\lambda \in \mathbb{R}%
_{0}^{+}$:

\begin{koro}[Uniform tree--decay bounds]
\label{tree bound main copy(1)}\mbox{
}\newline
Let $\rho $ be any arbitrary state on $\mathcal{U}$ and $\tau =\tau
^{(\omega ,\lambda )}$ be the one--parameter Bogoliubov group of
automorphisms defined by (\ref{rescaledbis}) for $\omega \in \Omega $ and $%
\lambda \in \mathbb{R}_{0}^{+}$. Then, for every $\epsilon \in \mathbb{R}%
^{+} $ and $t_{0},t\in \mathbb{R}$, $t_{0}<t$, there is $D=D_{\epsilon
,t_{0},t}\in \mathbb{R}^{+}$ such that the tree--decay bound (\ref{tree
bound 2}) holds for all $\omega \in \Omega $ and $\lambda \in \mathbb{R}%
_{0}^{+}$.
\end{koro}

\begin{proof}
Choose in Theorem \ref{tree bound main} sets $\bar{\Lambda}_{j},\Lambda _{j}$
containing exactly one element and note that, in this case, $|\mathcal{K}%
_{T}|=2^{2|T|}=2^{2(N-1)}$. Observe also that $\left\Vert \mathfrak{p}%
_{T}\left( \mathbf{x},\mathbf{y}\right) \right\Vert \leq 1$ as the
corresponding vectors $\psi _{x}$ have norm $1$. The assertion then follows
from (\ref{eq:anticomm bound 1}) and Theorem \ref{tree bound main}.
\end{proof}

\section{Proofs of Main Results\label{sect technical proofs}}

\subsection{Preliminary}

For the reader's convenience we start by reminding a few important
definitions and some standard mathematical results used in our proofs.

Recall that $\mathfrak{L}:=\mathbb{Z}^{d}$ with $d\in \mathbb{N}$, and $%
\mathcal{P}_{f}(\mathfrak{L})\subset 2^{\mathfrak{L}}$ is the set of all
finite subsets of $\mathfrak{L}$. For any $\Lambda \in \mathcal{P}_{f}(%
\mathfrak{L})$, $\mathcal{U}_{\Lambda }$ is the CAR $C^{\ast }$--algebra
generated by the identity $\mathbf{1}$ and the annihilation operators $%
\{a_{x}\}_{x\in \Lambda }$. It is isomorphic to the finite dimensional $%
C^{\ast }$--algebra $\mathcal{B}(\bigwedge \mathcal{H}_{\Lambda })$ of all
linear operators on the fermion Fock space $\bigwedge \mathcal{H}_{\Lambda }$%
, where $\mathcal{H}_{\Lambda }:=\oplus _{x\in \Lambda }\mathcal{H}_{x}$ is
the Cartesian product of copies $\mathcal{H}_{x}$, $x\in \Lambda $, of\ the
one--dimensional Hilbert space $\mathcal{H}\equiv \mathbb{C}$. (I.e., the
one--particle Hilbert space $\mathcal{H}_{\Lambda }$ is isomorphic to $%
\mathbb{C}^{\Lambda }$.) The CAR $C^{\ast }$--algebra $\mathcal{U}$ is the
(separable) $C^{\ast }$--algebra defined by the inductive limit of $\{%
\mathcal{U}_{\Lambda }\}_{\Lambda \in \mathcal{P}_{f}(\mathfrak{L})}$. Note
here that $\mathcal{U}_{\Lambda ^{\prime }}\subset \mathcal{U}_{\Lambda }$
whenever $\Lambda ^{\prime }\subset \Lambda $. For any one--particle wave
function $\psi \in \ell ^{2}(\mathfrak{L})$ we define annihilation and
creation operators $a(\psi ),a^{\ast }(\psi )\in \mathcal{U}$\ of a
(spinless) fermion, see (\ref{creation operators}).

For $\omega \in \Omega $ and $\lambda \in \mathbb{R}_{0}^{+}$, the
unperturbed dynamics of the fermion system studied here is given by the
one--parameter group $\tau ^{(\omega ,\lambda )}:=\{\tau _{t}^{(\omega
,\lambda )}\}_{t\in {\mathbb{R}}}$ of Bogoliubov automorphisms on the
algebra $\mathcal{U}$ uniquely defined by the condition (\ref{rescaledbis}),
that is,
\begin{equation}
\tau _{t}^{(\omega ,\lambda )}(a(\psi ))=a(\mathrm{e}^{it(\Delta _{\mathrm{d}%
}+\lambda V_{\omega })}\psi )\ ,\text{\qquad }t\in \mathbb{R},\ \psi \in
\ell ^{2}(\mathfrak{L})\ ,  \label{eq:Bogoliubov}
\end{equation}%
see \cite[Theorem 5.2.5]{BratteliRobinson}. As $\tau _{t}^{(\omega ,\lambda
)}$ is an automorphism of $\mathcal{U}$, by definition, we have in
particular that
\begin{equation}
\tau _{t}^{(\omega ,\lambda )}(B_{1}B_{2})=\tau _{t}^{(\omega ,\lambda
)}(B_{1})\tau _{t}^{(\omega ,\lambda )}(B_{2})\ ,\qquad B_{1},B_{2}\in
\mathcal{U}\ ,\ t\in \mathbb{R}\ .  \label{inequality idiote}
\end{equation}%
Physically, (\ref{eq:Bogoliubov}) means that the fermionic particles do not
experience any mutual force: They interact with each other via the Pauli
exclusion principle only, i.e., they form an ideal lattice fermion system.
From (\ref{norm cont of a}) and the norm--continuity of the unitary group $\{%
\mathrm{e}^{it(\Delta _{\mathrm{d}}+\lambda V_{\omega })}\}_{t\in {\mathbb{R}%
}}$ it follows that the (Bogoliubov) group $\tau ^{(\omega ,\lambda )}$ of
automorphisms is strongly continuous. $(\mathcal{U},\tau ^{(\omega ,\lambda
)})$ is thus a $C^{\ast }$--dynamical system.

For each $\omega \in \Omega $ and $\lambda \in \mathbb{R}_{0}^{+}$, the
generator of the strongly continuous group $\tau ^{(\omega ,\lambda )}$ is
denoted by $\delta ^{(\omega ,\lambda )}$. It is a symmetric unbounded
derivation. This means that the domain $\mathrm{Dom}(\delta ^{(\omega
,\lambda )})$ of $\delta ^{(\omega ,\lambda )}$ is a dense $\ast $%
--subalgebra of ${\mathcal{U}}$ and, for all $B_{1},B_{2}\in \mathrm{Dom}%
(\delta ^{(\omega ,\lambda )})$,
\begin{equation*}
\delta ^{(\omega ,\lambda )}(B_{1})^{\ast }=\delta ^{(\omega ,\lambda
)}(B_{1}^{\ast }),\quad \delta ^{(\omega ,\lambda )}(B_{1}B_{2})=\delta
^{(\omega ,\lambda )}(B_{1})B_{2}+B_{1}\delta ^{(\omega ,\lambda )}(B_{2})\ .
\end{equation*}

Recall that states on the $C^{\ast }$--algebra $\mathcal{U}$ are linear
functionals $\rho \in \mathcal{U}^{\ast }$ which are normalized and
positive, i.e., $\rho (\mathbf{1})=1$ and $\rho (A^{\ast }A)\geq 0$ for all $%
A\in \mathcal{U}$. Thermal equilibrium states of the fermion system under
consideration can be defined, at inverse temperature $\beta \in \mathbb{R}%
^{+}$ and for any $\omega \in \Omega $ and $\lambda \in \mathbb{R}_{0}^{+}$,
through the bounded positive operator%
\begin{equation*}
\mathbf{d}_{\mathrm{fermi}}^{(\beta ,\omega ,\lambda )}:=\frac{1}{1+\mathrm{e%
}^{\beta \left( \Delta _{\mathrm{d}}+\lambda V_{\omega }\right) }}\in
\mathcal{B}(\ell ^{2}(\mathfrak{L}))\ .
\end{equation*}%
Indeed, the so--called \emph{symbol} $\mathbf{d}_{\mathrm{fermi}}^{(\beta
,\omega ,\lambda )}$ uniquely defines a (faithful) \emph{quasi--free} state $%
\varrho ^{(\beta ,\omega ,\lambda )}$ on the CAR algebra $\mathcal{U}$ by
the conditions $\varrho ^{(\beta ,\omega ,\lambda )}(\mathbf{1})=1$ and%
\begin{equation*}
\varrho ^{(\beta ,\omega ,\lambda )}\left( a^{\ast }(f_{1})\ldots a^{\ast
}(f_{m})a(g_{n})\ldots a(g_{1})\right) =\delta _{m,n}\,\mathrm{det}\left(
[\langle g_{k},\mathbf{d}_{\mathrm{fermi}}^{(\beta ,\omega ,\lambda
)}f_{j}\rangle ]_{j,k}\right)
\end{equation*}%
for all $\left\{ f_{j}\right\} _{j=1}^{m},\left\{ g_{j}\right\}
_{j=1}^{n}\subset \ell ^{2}(\mathfrak{L})$ and $m,n\in \mathbb{N}$. $\langle
\cdot ,\cdot \rangle $ is here the scalar product in $\ell ^{2}(\mathfrak{L}%
) $.

The state $\varrho ^{(\beta ,\omega ,\lambda )}\in \mathcal{U}^{\ast }$ is
the unique $(\tau ^{(\omega ,\lambda )},\beta )$--KMS\ state of$\ $the $%
C^{\ast }$--dynamical system $(\mathcal{U},\tau ^{(\omega ,\lambda )})$.
This means that, for every $B_{1},B_{2}\in \mathcal{U}$, the map
\begin{equation*}
t\mapsto F_{B_{1},B_{2}}\left( t\right) :=\varrho ^{(\beta ,\omega ,\lambda
)}(B_{1}\tau _{t}^{(\omega ,\lambda )}(B_{2}))
\end{equation*}%
from ${\mathbb{R}}$ to $\mathbb{C}$ extends uniquely to a continuous map on $%
{\mathbb{R}}+i[0,\beta ]\subset {\mathbb{C}}$ which is holomorphic on ${%
\mathbb{R}}+i(0,\beta )$, such that
\begin{equation*}
F_{B_{1},B_{2}}\left( t+i\beta \right) =\varrho ^{(\beta ,\omega ,\lambda
)}(\tau _{t}^{(\omega ,\lambda )}(B_{2})B_{1})
\end{equation*}%
for all $t\in {\mathbb{R}}$. The latter is named \emph{KMS condition} or
\emph{modular condition} (when $\beta =1$) in the context of von Neumann
algebras.

The KMS condition is usually taken as the mathematical characterization of
thermal equilibriums of $C^{\ast }$--dynamical systems. This definition of
thermal equilibrium states for infinite systems is rather abstract. However,
it can be physically motivated from a maximum entropy principle by observing
that $\varrho ^{(\beta ,\omega ,\lambda )}$ is the unique weak$^{\ast }$%
--limit of Gibbs states $\varrho ^{(\beta ,\omega ,\lambda ,L)}$ (\ref{Gibbs
state H_n})--(\ref{Gibbs state H_nbis}), as $L\rightarrow \infty $. See
Theorem \ref{conv Gibbs}. Moreover, KMS states are stationary and thus, $%
\varrho ^{(\beta ,\omega ,\lambda )}$ is invariant under the dynamics
defined by the (Bogoliubov) group $\tau ^{(\omega ,\lambda )}$ of
automorphisms:
\begin{equation}
\varrho ^{(\beta ,\omega ,\lambda )}\circ \tau _{t}^{(\omega ,\lambda
)}=\varrho ^{(\beta ,\omega ,\lambda )}\ ,\qquad \beta \in \mathbb{R}^{+},\
\omega \in \Omega ,\ \lambda \in \mathbb{R}_{0}^{+},\ t\in \mathbb{R}\ .
\label{stationary}
\end{equation}

\subsection{Series Representation of Dynamics\label{Section existence
dynamics}}

\noindent The assertions of this subsection are similar to \cite[Proposition
5.4.26.]{BratteliRobinson}. Note however that the generator $\delta
^{(\omega ,\lambda )}$ of the (unperturbed) one--parameter group $\tau
^{(\omega ,\lambda )}$ is an \emph{unbounded} symmetric derivation, in
contrast to \cite[Proposition 5.4.26.]{BratteliRobinson}. Here, $\omega \in
\Omega $, $\lambda \in \mathbb{R}_{0}^{+}$ and $\mathbf{A}\in \mathbf{C}%
_{0}^{\infty }$ are arbitrarily fixed. See Sections \ref{Section dynamics}--%
\ref{Section Electromagnetic Fields}.

We start our proofs by giving an explicit expression of the automorphism $%
\tau _{t,s}^{(\omega ,\lambda ,\mathbf{A})}$ of $\mathcal{U}$ in terms of a
series involving multi--commutators. Meanwhile, we give an alternative
characterization of the two--parameter family $\{\tau _{t,s}^{(\omega
,\lambda ,\mathbf{A})}\}_{t\geq s}$ as a solution of an abstract Cauchy
initial value problem. This last observation is very useful in order to
generalize the present results to \emph{interacting} fermion systems.

First, recall that there is a unique (norm--continuous) two--parameter group
$\{\mathrm{U}_{t,s}^{(\omega ,\lambda ,\mathbf{A})}\}_{t\geq s}$ which is norm continuous and
solution of the non--autonomous Cauchy initial value problem (\ref{time
evolution one-particle}), that is,
\begin{equation*}
\forall s,t\in {\mathbb{R}},\ t\geq s:\quad \partial _{t}\mathrm{U}%
_{t,s}^{(\omega ,\lambda ,\mathbf{A})}=-i(\Delta _{\mathrm{d}}^{(\mathbf{A}%
(t,\cdot ))}+\lambda V_{\omega })\mathrm{U}_{t,s}^{(\omega ,\lambda ,\mathbf{%
A})}\ ,\quad \mathrm{U}_{s,s}^{(\omega ,\lambda ,\mathbf{A})}:=\mathbf{1}\ .
\end{equation*}%
(The restriction $t\geq s$ is not essential here and $\mathrm{U}%
_{t,s}^{(\omega ,\lambda ,\mathbf{A})}$ could also be defined for all $%
s,t\in {\mathbb{R}}$.) Indeed, $\Delta _{\mathrm{d}}\in \mathcal{B}(\ell
^{2}(\mathfrak{L}))$ and the map
\begin{equation}
t\mapsto \mathbf{w}_{t}^{\mathbf{A}}:=(\Delta _{\mathrm{d}}^{(\mathbf{A}%
(t,\cdot ))}-\Delta _{\mathrm{d}})\in \mathcal{B}(\ell ^{2}(\mathfrak{L}))
\label{def w}
\end{equation}%
from $\mathbb{R}$ to the set $\mathcal{B}(\ell ^{2}(\mathfrak{L}))$ of
bounded operators acting on $\ell ^{2}(\mathfrak{L})$ is continuously
differentiable for every $\mathbf{A}\in \mathbf{C}_{0}^{\infty }$. Hence, $\{%
\mathrm{U}_{t,s}^{(\omega )}\}_{t\geq s}$ can explicitly be written as the
Dyson--Phillips series%
\begin{eqnarray}
&&\mathrm{U}_{t,s}^{(\omega ,\lambda ,\mathbf{A})}-\mathrm{U}_{t-s}^{(\omega
,\lambda )}  \label{rewritebisbis} \\
&=&\sum\limits_{k\in {\mathbb{N}}}(-i)^{k}\int_{s}^{t}\mathrm{d}s_{1}\cdots
\int_{s}^{s_{k-1}}\mathrm{d}s_{k}\mathrm{U}_{t-s_{1}}^{(\omega ,\lambda )}%
\mathbf{w}_{s_{1}}^{\mathbf{A}}\mathrm{U}_{s_{1}-s_{2}}^{(\omega ,\lambda
)}\cdots \mathrm{U}_{s_{k-1}-s_{k}}^{(\omega ,\lambda )}\mathbf{w}_{s_{k}}^{%
\mathbf{A}}\mathrm{U}_{s_{k}-s}^{(\omega ,\lambda )}  \notag
\end{eqnarray}%
for any $t\geq s$, $\omega \in \Omega $, $\lambda \in \mathbb{R}_{0}^{+}$
and $\mathbf{A}\in \mathbf{C}_{0}^{\infty }$. Since all operators are
bounded, it is easy to check that $\{\mathrm{U}_{t,s}^{(\omega )}\}_{t\geq
s} $ is a family of unitary operators.

We are now in position to represent the Bogoliubov automorphisms $\tau
_{t,s}^{(\omega ,\lambda ,\mathbf{A})}$ defined by (\ref{Cauchy problem 0})
as a Dyson--Phillips series involving the unperturbed dynamics defined by
the one--parameter group $\tau ^{(\omega ,\lambda )}:=\{\tau _{t}^{(\omega
,\lambda )}\}_{t\in {\mathbb{R}}}$, see (\ref{rescaled}) and (\ref%
{rescaledbis}). To this end, for every $\mathbf{A}\in \mathbf{C}_{0}^{\infty
}$, we denote the second quantization of $\mathbf{w}_{t}^{\mathbf{A}}$ by
\begin{eqnarray}
W_{t}^{\mathbf{A}} &=&\sum\limits_{x,y\in \mathfrak{L}}\left[ \exp \left(
-i\int\nolimits_{0}^{1}\left[ \mathbf{A}(t,\alpha y+(1-\alpha )x)\right]
(y-x)\mathrm{d}\alpha \right) -1\right]  \notag \\
&&\qquad \ \ \ \times \langle \mathfrak{e}_{x},\Delta _{\mathrm{d}}\mathfrak{%
e}_{y}\rangle a_{x}^{\ast }a_{y}\ ,  \label{def:W}
\end{eqnarray}%
see (\ref{eq discrete lapla A}), (\ref{eq def W}) and (\ref{def w}). Note
that there is a finite subset $\Lambda \in \mathcal{P}_{f}(\mathfrak{L})$
such that $W_{t}^{\mathbf{A}}\in \mathcal{U}_{\Lambda }$ for all $t\in
\mathbb{R}$ because $\mathbf{A}\in \mathbf{C}_{0}^{\infty }$. We also define
the continuously differentiable map
\begin{equation}
t\mapsto L_{t}^{\mathbf{A}}:=i[W_{t}^{\mathbf{A}},\ \cdot \ ]\in \mathcal{B}%
\left( \mathcal{U}\right)  \label{L new}
\end{equation}%
from $\mathbb{R}$ to the set $\mathcal{B}\left( \mathcal{U}\right) $ of
bounded operators acting on $\mathcal{U}$.

\begin{satz}[Dynamics as a Dyson--Phillips series]
\label{bound incr 1 Lemma copy(4)}\mbox{
}\newline
For any $\omega \in \Omega $, $\lambda \in \mathbb{R}_{0}^{+}$, $\mathbf{A}%
\in \mathbf{C}_{0}^{\infty }$ and $t,s\in {\mathbb{R}}$,\ $t\geq s$,
\begin{equation*}
\tau _{t,s}^{(\omega ,\lambda ,\mathbf{A})}=\tau _{t-s}^{(\omega ,\lambda
)}+\sum\limits_{k\in {\mathbb{N}}}\int_{s}^{t}\mathrm{d}s_{1}\cdots
\int_{s}^{s_{k-1}}\mathrm{d}s_{k}\tau _{s_{k}-s}^{(\omega ,\lambda
)}L_{s_{k}}^{\mathbf{A}}\tau _{s_{k-1}-s_{k}}^{(\omega ,\lambda )}\cdots
\tau _{s_{1}-s_{2}}^{(\omega ,\lambda )}L_{s_{1}}^{\mathbf{A}}\tau
_{t-s_{1}}^{(\omega ,\lambda )}\ .
\end{equation*}
\end{satz}

\begin{proof}
Let $\omega \in \Omega $, $\lambda \in \mathbb{R}_{0}^{+}$, $\mathbf{A}\in
\mathbf{C}_{0}^{\infty }$ and define%
\begin{equation}
\check{\tau}_{t,s}^{(\omega ,\lambda ,\mathbf{A})}:=\tau _{t-s}^{(\omega
,\lambda )}+\sum\limits_{k\in {\mathbb{N}}}\int_{s}^{t}\mathrm{d}s_{1}\cdots
\int_{s}^{s_{k-1}}\mathrm{d}s_{k}\tau _{s_{k}-s}^{(\omega ,\lambda
)}L_{s_{k}}^{\mathbf{A}}\tau _{s_{k-1}-s_{k}}^{(\omega ,\lambda )}\cdots
\tau _{s_{1}-s_{2}}^{(\omega ,\lambda )}L_{s_{1}}^{\mathbf{A}}\tau
_{t-s_{1}}^{(\omega ,\lambda )}  \label{rewritebis}
\end{equation}%
for any $t\geq s$. This series is absolutely convergent. Indeed, $\tau
^{(\omega ,\lambda )}:=\{\tau _{t}^{(\omega ,\lambda )}\}_{t\in {\mathbb{R}}%
} $ is a norm--continuous one--parameter group of contractions, i.e.,
\begin{equation*}
\Vert \tau _{t}^{(\omega ,\lambda )}\Vert _{\mathrm{op}}\leq 1\ ,\qquad t\in
\mathbb{R}\ ,
\end{equation*}%
whereas, for any $\mathbf{A}\in \mathbf{C}_{0}^{\infty }$, the map (\ref{L
new}) is continuously differentiable and there is a constant $D\in \mathbb{R}%
^{+}$ such that%
\begin{equation}
\underset{t\in \mathbb{R}}{\sup }\Vert L_{t}^{\mathbf{A}}\Vert _{\mathrm{op}%
}<D\ ,  \label{rewrite new}
\end{equation}%
because $W_{t}^{\mathbf{A}}=0$ for any $t\notin \lbrack t_{0},t_{1}]$, i.e.,
there is no electromagnetic field for times $t\notin \lbrack t_{0},t_{1}]$.
Here, the notation $\Vert \cdot \Vert _{\mathrm{op}}$ stands for the
operator norm. By (\ref{rewritebis})--(\ref{rewrite new}), it follows that
\begin{equation*}
\Vert \check{\tau}_{t,s}^{(\omega ,\lambda ,\mathbf{A})}\Vert _{\mathrm{op}%
}\leq \mathrm{e}^{D(t-s)}\ ,\qquad t,s\in {\mathbb{R}},\ t\geq s\ .
\end{equation*}%
Now, straightforward computations using (\ref{def w}) and (\ref{L new}) show
that the following \textquotedblleft pull through\textquotedblright\ formula
holds:
\begin{equation}
L_{t}^{\mathbf{A}}(a(\psi ))=a(i\mathbf{w}_{t}^{\mathbf{A}}\psi )\ ,\text{%
\qquad }t\in {\mathbb{R}},\ \psi \in \ell ^{2}(\mathfrak{L})\ . \label{extra}
\end{equation}%
We therefore infer from (\ref{rescaledbis}), (\ref{rewritebisbis}) and (\ref%
{rewritebis}) that
\begin{equation}
\check{\tau}_{t,s}^{(\omega ,\lambda ,\mathbf{A})}\left( a\left( \psi
\right) \right) =a((\mathrm{U}_{t,s}^{(\omega ,\lambda ,\mathbf{A})})^{\ast
}(\psi ))\ ,\text{\qquad }t\geq s,\ \psi \in \ell ^{2}(\mathfrak{L})\ ,
\label{automorphism taucheck}
\end{equation}%
for all $\omega \in \Omega $, $\lambda \in \mathbb{R}_{0}^{+}$ and $\mathbf{A%
}\in \mathbf{C}_{0}^{\infty }$. Direct computations show, for all $t\geq s$,
that $\check{\tau}_{t,s}^{(\omega ,\lambda ,\mathbf{A})}$ is an automorphism
of $\mathcal{U}$: Use the fact that, for all $t\in {\mathbb{R}}$, $\tau
_{t}^{(\omega ,\lambda )}$ is an automorphisms of $\mathcal{U}$ and $L_{t}^{%
\mathbf{A}}$ is a bounded symmetric derivation on $\mathcal{U}$, i.e., $%
L_{t}^{\mathbf{A}}(B_{1}^{\ast })=L_{t}^{\mathbf{A}}(B_{1})^{\ast }$ and
\begin{equation*}
L_{t}^{\mathbf{A}}(B_{1}B_{2})=L_{t}^{\mathbf{A}}(B_{1})B_{2}+B_{1}L_{t}^{%
\mathbf{A}}(B_{2})\in \mathcal{U}\ ,\qquad B_{1},B_{2}\in \mathcal{U}\ .
\end{equation*}%
By \cite[Theorem 5.2.5]{BratteliRobinson}, the condition (\ref{automorphism
taucheck}) uniquely defines automorphisms of $\mathcal{U}$. \ As a
consequence, one gets $\check{\tau}_{t,s}^{(\omega ,\lambda ,\mathbf{A}%
)}=\tau _{t,s}^{(\omega ,\lambda ,\mathbf{A})}$, see (\ref{Cauchy problem 0}%
).
\end{proof}

A straightforward consequence of Theorem \ref{bound incr 1 Lemma copy(4)} is
that, for any $\omega \in \Omega $ and $\lambda \in \mathbb{R}_{0}^{+}$, the
family $\{\tau _{t,s}^{(\omega ,\lambda ,\mathbf{A})}\}_{t\geq s}$ satisfies
(\ref{Cauchy problem 1}) with%
\begin{equation}
\delta _{t}^{(\omega ,\lambda ,\mathbf{A})}:=\delta ^{(\omega ,\lambda
)}+i[W_{t}^{\mathbf{A}},\ \cdot \ ]\ ,\text{\qquad }t\in \mathbb{R}\ .
\label{explicit delta}
\end{equation}%
Here, the symmetric derivation $\delta ^{(\omega ,\lambda )}$ is the
(unbounded) generator of the one--parameter group $\tau ^{(\omega ,\lambda
)} $. Indeed, one obtains:

\begin{koro}[Abstract Cauchy initial value problem for $\protect\tau %
_{t,s}^{(\protect\omega ,\protect\lambda ,\mathbf{A})}$]
\label{bound incr 1 Lemma copy(6)}\mbox{
}\newline
For any $\omega \in \Omega $, $\lambda \in \mathbb{R}_{0}^{+}$ and $\mathbf{A%
}\in \mathbf{C}_{0}^{\infty }$, $\{\tau _{t,s}^{(\omega ,\lambda ,\mathbf{A}%
)}\}_{t\geq s}$ satisfies (\ref{Cauchy problem 1}), that is,
\begin{equation*}
\forall t,s\in {\mathbb{R}},\ t\geq s:\quad \partial _{t}\tau
_{t,s}^{(\omega ,\lambda ,\mathbf{A})}=\tau _{t,s}^{(\omega ,\lambda ,%
\mathbf{A})}\circ \delta _{t}^{(\omega ,\lambda ,\mathbf{A})},\quad \tau
_{s,s}^{(\omega ,\lambda ,\mathbf{A})}:=\mathbf{1}\ ,
\end{equation*}%
on the dense subspace $\mathrm{Dom}(\delta ^{(\omega ,\lambda )})\subset
\mathcal{U}$.
\end{koro}

\begin{proof}
By Theorem \ref{bound incr 1 Lemma copy(4)}, the family $\{\tau
_{t,s}^{(\omega ,\lambda ,\mathbf{A})}\}_{t\geq s}$ obeys the integral
equation
\begin{equation*}
\tau _{t,s}^{(\omega ,\lambda ,\mathbf{A})}\left( B\right) =\tau
_{t-s}^{(\omega ,\lambda )}\left( B\right) +\int_{s}^{t}\tau
_{s_{1},s}^{(\omega ,\lambda )}L_{s_{1}}^{\mathbf{A}}\tau
_{t-s_{1}}^{(\omega ,\lambda )}\left( B\right) \mathrm{d}s_{1}\ ,\qquad B\in
\mathcal{U}\ ,
\end{equation*}%
which directly yields the assertion because $\mathbf{A}\in \mathbf{C}%
_{0}^{\infty }$.
\end{proof}

Recall the notation%
\begin{equation}
W_{t,s}^{\mathbf{A}}\equiv W_{t,s}^{(\omega ,\lambda ,\mathbf{A})}:=\tau
_{t}^{(\omega ,\lambda )}(W_{s}^{\mathbf{A}})\in \mathcal{U},\qquad \omega
\in \Omega ,\ \lambda \in \mathbb{R}_{0}^{+},\ \mathbf{A}\in \mathbf{C}%
_{0}^{\infty },\ t,s\in \mathbb{R}\ ,\   \label{def LA}
\end{equation}%
and the inductive definition (\ref{multi1-0})--(\ref{multi2-0}) of
multi--commutators:
\begin{equation}
{[}B_{1},B_{2}{]}^{(2)}:=[B_{1},B_{2}]:=B_{1}B_{2}-B_{2}B_{1}\ ,\qquad
B_{1},B_{2}\in \mathcal{U}\ ,  \label{multi1}
\end{equation}%
and, for all integers $k>2$,
\begin{equation}
{[}B_{1},B_{2},\ldots ,B_{k+1}{]}^{(k+1)}:={[}B_{1},{[}B_{2},\ldots ,B_{k+1}{%
]}^{(k)}{]}\ ,\qquad B_{1},\ldots ,B_{k+1}\in \mathcal{U}\ .  \label{multi2}
\end{equation}%
Then, using (\ref{inequality idiote}) we rewrite the Dyson--Phillips series
of Theorem \ref{bound incr 1 Lemma copy(4)} as%
\begin{eqnarray}
&&\tau _{t,s}^{(\omega ,\lambda ,\mathbf{A})}\left( B\right) -\tau
_{t-s}^{(\omega ,\lambda )}\left( B\right)  \label{Dyson tau 1} \\
&=&\sum\limits_{k\in {\mathbb{N}}}i^{k}\int_{s}^{t}\mathrm{d}s_{1}\cdots
\int_{s}^{s_{k-1}}\mathrm{d}s_{k}[W_{s_{k}-s,s_{k}}^{\mathbf{A}},\ldots
,W_{s_{1}-s,s_{1}}^{\mathbf{A}},\tau _{t-s}^{(\omega ,\lambda )}(B)]^{(k+1)}
\notag
\end{eqnarray}%
for any $B\in \mathcal{U}$, $\omega \in \Omega $, $\lambda \in \mathbb{R}%
_{0}^{+}$, $\mathbf{A}\in \mathbf{C}_{0}^{\infty }$ and $t\geq s$.

\subsection{Interaction Picture of Dynamics\label{Section interaction
picture}}

In contrast to the two--parameter family $\{\tau _{t,s}^{(\omega ,\lambda ,%
\mathbf{A})}\}_{t\geq s}$,
\begin{equation*}
\{\tau _{t_{0}}^{(\omega ,\lambda )}\circ \tau _{t,t_{0}}^{(\omega ,\lambda ,%
\mathbf{A})}\circ \tau _{-t}^{(\omega ,\lambda )}\}_{t\geq t_{0}}
\end{equation*}%
is a family of \emph{inner} automorphisms of the CAR\ algebra $\mathcal{U}$,
i.e., it can be implemented by conjugation with unitary elements $\mathfrak{V%
}_{t,t_{0}}$\ of $\mathcal{U}$, similar to Remark \ref{Heisenberg Picture
remark}:
\begin{equation*}
\tau _{t_{0}}^{(\omega ,\lambda )}\circ \tau _{t,t_{0}}^{(\omega ,\lambda ,%
\mathbf{A})}\circ \tau _{-t}^{(\omega ,\lambda )}\left( B\right) =\mathfrak{V%
}_{t,t_{0}}B\mathfrak{V}_{t,t_{0}}^{\ast }\ ,\qquad B\in \mathcal{U}\ .
\end{equation*}%
On the other hand, by using two times the stationarity of the KMS state $%
\varrho ^{(\beta ,\omega ,\lambda )}$ w.r.t. the unperturbed dynamics (cf. (%
\ref{stationary})) as well as (\ref{inequality idiote}), we observe that the
time evolution (\ref{time dependent state}) of the state of the fermion
system equals
\begin{eqnarray}
\rho _{t}^{(\beta ,\omega ,\lambda ,\mathbf{A})}\left( B\right) &=&\varrho
^{(\beta ,\omega ,\lambda )}\circ \tau _{t_{0}}^{(\omega ,\lambda )}\circ
\tau _{t,t_{0}}^{(\omega ,\lambda ,\mathbf{A})}\left( B\right)
\label{explanation} \\
&=&\varrho ^{(\beta ,\omega ,\lambda )}\left( \mathfrak{V}_{t,t_{0}}\tau
_{t}^{(\omega ,\lambda ,\mathbf{A})}\left( B\right) \mathfrak{V}%
_{t,t_{0}}^{\ast }\right) =\varrho ^{(\beta ,\omega ,\lambda )}\left(
\mathfrak{U}_{t}^{\ast }B\mathfrak{U}_{t}\right)  \notag
\end{eqnarray}%
for any $t\geq t_{0}$, where%
\begin{equation}
\mathfrak{U}_{t}:=\tau _{-t}^{(\omega ,\lambda )}\left( \mathfrak{V}%
_{t,t_{0}}^{\ast }\right) \ ,\qquad t\geq t_{0}\ .
\label{definition tho chap inv3}
\end{equation}%
This family of unitary elements of $\mathcal{U}$ turns out to be within the
domain $\mathrm{Dom}(\delta ^{(\omega ,\lambda )})$ of the (unbounded)
generator $\delta ^{(\omega ,\lambda )}$ of the one--parameter group $\tau
^{(\omega ,\lambda )}$ of automorphisms. These properties are quite useful
to show in Section \ref{Section Existence} both the existence of the energy
increment (\ref{entropic energy increment}) as well as Theorem \ref{main 1
copy(1)}.

The above heuristics is proven in the following theorem:

\begin{satz}[Interaction picture of dynamics]
\label{Theo int pict}\mbox{
}\newline
For any $\omega \in \Omega $, $\lambda \in \mathbb{R}_{0}^{+}$ and $\mathbf{A%
}\in \mathbf{C}_{0}^{\infty }$, there is a family
\begin{equation*}
\{\mathfrak{U}_{t}\equiv \mathfrak{U}_{t}^{(\omega ,\lambda ,\mathbf{A}%
)}\}_{t\geq t_{0}}\subset \mathrm{Dom}(\delta ^{(\omega ,\lambda )})
\end{equation*}%
of unitary elements of $\mathcal{U}$ such that, for all $\beta \in \mathbb{R}%
^{+}$, $t\geq t_{0}$ and $B\in $ $\mathcal{U}$,%
\begin{equation*}
\rho _{t}^{(\beta ,\omega ,\lambda ,\mathbf{A})}\left( B\right) =\varrho
^{(\beta ,\omega ,\lambda )}\left( \mathfrak{U}_{t}^{\ast }B\mathfrak{U}%
_{t}\right) \ .
\end{equation*}
\end{satz}

\begin{proof}
The arguments to prove this theorem are relatively standard for autonomous
perturbations of KMS states, see \cite[Sections 5.4.1.]{BratteliRobinson}.
We adapt them to the non--autonomous case as suggested in \cite[Sections
5.4.4., Proposition 5.4.26.]{BratteliRobinson}. However, in contrast to \cite%
[Sections 5.4.1., 5.4.4.]{BratteliRobinson}, the situation we treat here
requires more care because the symmetric derivation $\delta ^{(\omega
,\lambda )}$ is \emph{unbounded}.

For any $\omega \in \Omega $, $\lambda \in \mathbb{R}_{0}^{+}$ and $\mathbf{A%
}\in \mathbf{C}_{0}^{\infty }$, we define the family $\{\mathfrak{U}%
_{t,s}\}_{t,s\in \mathbb{R}}\subset \mathcal{U}$ by the series
\begin{equation}
\mathfrak{V}_{t,s}\equiv \mathfrak{V}_{t,s}^{(\omega ,\lambda ,\mathbf{A})}:=%
\mathbf{1+}\sum\limits_{k\in {\mathbb{N}}}i^{k}\int_{s}^{t}\mathrm{d}%
s_{1}\cdots \int_{s}^{s_{k-1}}\mathrm{d}s_{k}W_{s_{k},s_{k}}^{\mathbf{A}%
}\cdots W_{s_{1},s_{1}}^{\mathbf{A}}\ ,  \label{Dyson interact}
\end{equation}%
where we recall that $W_{t,s}^{\mathbf{A}}\equiv W_{t,s}^{(\omega ,\lambda ,%
\mathbf{A})}\in \mathcal{U}$ is defined by (\ref{def LA}) for any $t,s\in
\mathbb{R}$. The series is well--defined in the Banach space%
\begin{equation}
\mathcal{Y}:=(\mathrm{Dom}(\delta ^{(\omega ,\lambda )}),\left\Vert \cdot
\right\Vert _{\delta ^{(\omega ,\lambda )}})\ ,  \label{graph space}
\end{equation}%
where $\left\Vert \cdot \right\Vert _{\delta ^{(\omega ,\lambda )}}$ stands
for the graph norm of the closed operator $\delta ^{(\omega ,\lambda )}$. In
particular,
\begin{equation}
\{\mathfrak{V}_{t,s}\}_{t,s\in \mathbb{R}}\subset \mathrm{Dom}(\delta
^{(\omega ,\lambda )})\ .  \label{toto0}
\end{equation}

Indeed, the strongly continuous group $\tau ^{(\omega ,\lambda )}$ on $%
\mathcal{U}$ defines, by restriction, a strongly continuous group on $%
\mathcal{Y}$. For more details, see, e.g., \cite[Section II.5.a, 5.2
Proposition]{EngelNagel}. Observe also from the strong continuity and group
property in $\mathcal{Y}$ of the restriction of $\tau ^{(\omega ,\lambda )}$
to the space $\mathrm{Dom}(\delta ^{(\omega ,\lambda )})$ that%
\begin{equation}
\left\Vert \tau _{t}^{(\omega ,\lambda )}|_{\mathrm{Dom}(\delta ^{(\omega
,\lambda )})}\right\Vert _{\mathcal{B}(\mathcal{Y})}\leq D_{1}e^{D_{^{2}}|t|}
\label{bound semigroup}
\end{equation}%
for some finite constants $D_{1},D_{2}\in \mathbb{R}^{+}$ and all $t\in
\mathbb{R}$. Here, $\mathcal{B}(\mathcal{Y})$ is the Banach space of bounded
operators acting on $\mathcal{Y}$. Moreover, for any $\mathbf{A}\in \mathbf{C%
}_{0}^{\infty }$, $s\mapsto W_{s}^{\mathbf{A}}$ is a smooth, compactly
supported map from $\mathbb{R}$ to $\mathcal{Y}$. Since $\delta ^{(\omega
,\lambda )}$ is a symmetric derivation, it follows that the series (\ref%
{Dyson interact}) absolutely converges in the Banach space $\mathcal{Y}$ and%
\begin{equation*}
\mathfrak{V}_{t,s}=\mathbf{1+}\sum\limits_{k\in {\mathbb{N}}%
}i^{k}\int_{s}^{t}\mathrm{d}s_{k}\cdots \int_{s_{2}}^{t}\mathrm{d}%
s_{1}W_{s_{k},s_{k}}^{\mathbf{A}}\cdots W_{s_{1},s_{1}}^{\mathbf{A}}\ ,
\end{equation*}%
where the r.h.s. of this equation also absolutely converges in $\mathcal{Y}$%
. Therefore, for any $t,s\in \mathbb{R}$, the operator $\mathfrak{V}_{t,s}$
obeys the integral equation
\begin{equation}
\mathfrak{V}_{t,s}=\mathbf{1}+i\int_{s}^{t}\mathfrak{V}%
_{s_{1},s}W_{s_{1},s_{1}}^{\mathbf{A}}\mathrm{d}s_{1}=\mathbf{1}%
+i\int_{s}^{t}W_{s_{1},s_{1}}^{\mathbf{A}}\mathfrak{V}_{t,s_{1}}\mathrm{d}%
s_{1}  \label{Dyson interact +1}
\end{equation}%
in $\mathcal{Y}$. The families $\{\mathfrak{U}_{t,s}\}_{t,s\in \mathbb{R}}$
and $\{W_{t,t}^{\mathbf{A}}\}_{t\in \mathbb{R}}$ are both continuous in $%
\mathcal{Y}$ and $\delta ^{(\omega ,\lambda )}$ is a symmetric derivation.
As a consequence, (\ref{Dyson interact +1}) implies that, for any $t,s\in
\mathbb{R}$,%
\begin{equation}
\partial _{t}\mathfrak{V}_{t,s}=i\mathfrak{V}_{t,s}W_{t,t}^{\mathbf{A}%
}\qquad \text{and}\qquad \partial _{s}\mathfrak{V}_{t,s}=-iW_{s,s}^{\mathbf{A%
}}\mathfrak{V}_{t,s}  \label{Dyson interact diff eq}
\end{equation}%
both in the Banach space $\mathcal{Y}$, and thus in $\mathcal{U}$. Since $%
W_{t,t}^{\mathbf{A}}=(W_{t,t}^{\mathbf{A}})^{\ast }$, by using the
norm--continuity of the map $B\mapsto B^{\ast }$ on $\mathcal{U}$, we
compute from (\ref{Dyson interact diff eq}) that%
\begin{eqnarray*}
\mathbf{1}-\mathfrak{V}_{t,s}^{\ast }\mathfrak{V}_{t,s}
&=&\int_{s}^{t}\partial _{s_{1}}\left\{ \mathfrak{V}_{t,s_{1}}^{\ast }%
\mathfrak{V}_{t,s_{1}}\right\} \mathrm{d}s_{1}=0\ . \\
\mathbf{1}-\mathfrak{V}_{t,s}\mathfrak{V}_{t,s}^{\ast }
&=&\int_{t}^{s}\partial _{s_{1}}\left\{ \mathfrak{V}_{s_{1},s}\mathfrak{V}%
_{s_{1},s}^{\ast }\right\} \mathrm{d}s_{1}=0\ .
\end{eqnarray*}%
In other words, $\{\mathfrak{V}_{t,s}\}_{t,s\in \mathbb{R}}$ is a family of
unitary elements of $\mathrm{Dom}(\delta ^{(\omega ,\lambda )})\subset
\mathcal{U}$, by (\ref{toto0}).

Now, we define the family $\{\mathfrak{W}_{s,t}^{(\omega ,\lambda ,\mathbf{A}%
)}\}_{s,t\in \mathbb{R}}$ of bounded operators acting on the Banach space $%
\mathcal{U}$ by
\begin{equation}
\mathfrak{W}_{s,t}^{(\omega ,\lambda ,\mathbf{A})}\left( B\right) :=\tau
_{-s}^{(\omega ,\lambda )}\left( \mathfrak{V}_{t,s}\tau _{t}^{(\omega
,\lambda )}(B)\mathfrak{V}_{t,s}^{\ast }\right) \ ,\qquad B\in \mathcal{U}\ .
\label{definition tho chap inv}
\end{equation}%
Clearly, for any $B\in \mathcal{U}$, the map%
\begin{equation*}
(s,t)\mapsto \mathfrak{W}_{s,t}^{(\omega ,\lambda ,\mathbf{A})}\left(
B\right) \in \mathcal{U}
\end{equation*}%
from $\mathbb{R}^{2}$ to $\mathcal{U}$ is continuous. Moreover, by
construction, $\mathfrak{W}_{t,t}^{(\omega ,\lambda ,\mathbf{A})}=\mathbf{1}$
and for all $B\in \mathrm{Dom}(\delta ^{(\omega ,\lambda )})$ and $s,t\in
\mathbb{R}$,
\begin{equation*}
\mathfrak{W}_{s,t}^{(\omega ,\lambda ,\mathbf{A})}\left( B\right) \in
\mathrm{Dom}(\delta ^{(\omega ,\lambda )})=\mathrm{Dom}(\delta _{s}^{(\omega
,\lambda ,\mathbf{A})})\ ,
\end{equation*}%
because $\tau _{t}^{(\omega ,\lambda )}$ preserves the (dense) subspace $%
\mathrm{Dom}(\delta ^{(\omega ,\lambda )})\subset \mathcal{U}$. Therefore,
we infer from (\ref{inequality idiote}), (\ref{explicit delta}) and (\ref%
{Dyson interact diff eq}) that%
\begin{equation}
\forall s,t\in {\mathbb{R}}:\quad \partial _{s}\mathfrak{W}_{s,t}^{(\omega
,\lambda ,\mathbf{A})}=-\delta _{s}^{(\omega ,\lambda ,\mathbf{A})}\circ
\mathfrak{W}_{s,t}^{(\omega ,\lambda ,\mathbf{A})},\quad \mathfrak{W}%
_{t,t}^{(\omega ,\lambda ,\mathbf{A})}=\mathbf{1}\ ,  \label{cauchy toto1}
\end{equation}%
whereas
\begin{equation}
\forall s,t\in {\mathbb{R}}:\quad \partial _{t}\mathfrak{W}_{s,t}^{(\omega
,\lambda ,\mathbf{A})}=\mathfrak{W}_{s,t}^{(\omega ,\lambda ,\mathbf{A}%
)}\circ \delta _{t}^{(\omega ,\lambda ,\mathbf{A})},\quad \mathfrak{W}%
_{s,s}^{(\omega ,\lambda ,\mathbf{A})}=\mathbf{1}\ ,  \label{cauchy toto2}
\end{equation}%
both in the strong sense in $\mathrm{Dom}(\delta ^{(\omega ,\lambda
)})\subset \mathcal{U}$. In particular, by Corollary \ref{bound incr 1 Lemma
copy(6)}, the families $\{\tau _{t,s}^{(\omega ,\lambda ,\mathbf{A}%
)}\}_{t\geq s}$ and $\{\mathfrak{W}_{s,t}^{(\omega ,\lambda ,\mathbf{A}%
)}\}_{s,t\in \mathbb{R}}$ satisfy the equality
\begin{equation}
\tau _{t,s}^{(\omega ,\lambda ,\mathbf{A})}\left( B\right) -\mathfrak{W}%
_{s,t}^{(\omega ,\lambda ,\mathbf{A})}\left( B\right) =\int_{s}^{t}\partial
_{s_{1}}\{\tau _{s_{1},s}^{(\omega ,\lambda ,\mathbf{A})}\mathfrak{W}%
_{s_{1},t}^{(\omega ,\lambda ,\mathbf{A})}\left( B\right) \}\mathrm{d}s_{1}=0
\label{uniqueness eq}
\end{equation}%
for any $B\in \mathrm{Dom}(\delta ^{(\omega ,\lambda )})$ and $t\geq s$.
Remark that we use the strong continuity of the family $\{\tau
_{t,s}^{(\omega ,\lambda ,\mathbf{A})}\}_{t\geq s}$ w.r.t. $t\in \mathbb{R}$
to show from Corollary \ref{bound incr 1 Lemma copy(6)} and (\ref{cauchy
toto1}) that%
\begin{equation*}
\partial _{s_{1}}\{\tau _{s_{1},s}^{(\omega ,\lambda ,\mathbf{A})}\mathfrak{W%
}_{s_{1},t}^{(\omega ,\lambda ,\mathbf{A})}\left( B\right) \}=0\ ,
\end{equation*}%
for any $B\in \mathrm{Dom}(\delta ^{(\omega ,\lambda )})$ and $t\geq s$. The
domain $\mathrm{Dom}(\delta ^{(\omega ,\lambda )})$ is dense in $\mathcal{U}$
and both operators $\tau _{t,s}^{(\omega ,\lambda ,\mathbf{A})}$ and $%
\mathfrak{W}_{s,t}^{(\omega ,\lambda ,\mathbf{A})}$ are bounded. As a
consequence, (\ref{uniqueness eq}) yields
\begin{equation}
\tau _{t,s}^{(\omega ,\lambda ,\mathbf{A})}=\mathfrak{W}_{s,t}^{(\omega
,\lambda ,\mathbf{A})}  \label{definition tho chap inv2}
\end{equation}%
for any $\omega \in \Omega $, $\lambda \in \mathbb{R}_{0}^{+}$, $\mathbf{A}%
\in \mathbf{C}_{0}^{\infty }$ and $t\geq s$.

Use now Equation (\ref{definition tho chap inv3}) to define the family $\{%
\mathfrak{U}_{t}\}_{t\geq t_{0}}$. Since, for any $t\in \mathbb{R}$, $\tau
_{-t}^{(\omega ,\lambda )}$ is an automorphism of $\mathcal{U}$ which
preserves the domain $\mathrm{Dom}(\delta ^{(\omega ,\lambda )})$, we deduce
from (\ref{toto0}) and the unitarity of $\mathfrak{V}_{t,s}$ that
\begin{equation*}
\{\mathfrak{U}_{t}\}_{t\geq t_{0}}\subset \mathrm{Dom}(\delta ^{(\omega
,\lambda )})
\end{equation*}%
is a family of unitary elements of $\mathcal{U}$. Note indeed that $\mathrm{%
Dom}(\delta ^{(\omega ,\lambda )})$ is a $\ast $--algebra, since $\delta
^{(\omega ,\lambda )}$ is a symmetric derivation. Moreover, from (\ref{time
dependent state}), (\ref{definition tho chap inv3}), (\ref{definition tho
chap inv}) and (\ref{definition tho chap inv2}) combined with the
stationarity of the KMS state $\varrho ^{(\beta ,\omega ,\lambda )}$ w.r.t.
the unperturbed dynamics (cf. (\ref{stationary})) we arrive at the
assertion, as explained in Equation (\ref{explanation}).
\end{proof}

The proof of Theorem \ref{Theo int pict} gives supplementary information on
the dynamics. This is not used in the present paper, but it can be employed
to uniquely define dynamics for systems of interacting fermions on the
lattice, as discussed at the end of Section \ref{section Dynamics}.

First, by (\ref{definition tho chap inv}), $\{\mathfrak{W}_{s,t}^{(\omega
,\lambda ,\mathbf{A})}\}_{s,t\in \mathbb{R}}$ is a family of bounded
operators acting on the Banach space $\mathcal{U}$ that of course extends $%
\{\tau _{t,s}^{(\omega ,\lambda ,\mathbf{A})}\}_{t\geq s}$ to all $s,t\in
\mathbb{R}$, see (\ref{definition tho chap inv2}). Moreover, it is the
unique \emph{fundamental solution} of a non--autonomous evolution equation.
By fundamental solution, we mean here that the family $\{\mathfrak{W}%
_{s,t}^{(\omega ,\lambda ,\mathbf{A})}\}_{s,t\in \mathbb{R}}$ of bounded
operators acting on $\mathcal{U}$ is strongly continuous, conserves the
domain
\begin{equation*}
\mathrm{Dom}(\delta _{t}^{(\omega ,\lambda ,\mathbf{A})})=\mathrm{Dom}%
(\delta ^{(\omega ,\lambda )})\ ,
\end{equation*}%
satisfies%
\begin{eqnarray*}
\mathfrak{W}_{\cdot ,t}^{(\omega ,\lambda ,\mathbf{A})}(B) &\in &C^{1}(%
\mathbb{R};(\mathrm{Dom}(\delta ^{(\omega ,\lambda )}),\left\Vert \cdot
\right\Vert ))\ , \\
\mathfrak{W}_{s,\cdot }^{(\omega ,\lambda ,\mathbf{A})}(B) &\in &C^{1}(%
\mathbb{R};(\mathrm{Dom}(\delta ^{(\omega ,\lambda )}),\left\Vert \cdot
\right\Vert ))\ ,
\end{eqnarray*}%
for all $B\in \mathrm{Dom}(\delta ^{(\omega ,\lambda )})$, and solves the
abstract Cauchy initial value\ problem (\ref{cauchy toto1}) on $\mathrm{Dom}%
(\delta ^{(\omega ,\lambda )})$:

\begin{proposition}[Evolution equations for $\mathfrak{W}_{s,t}^{(\protect%
\omega ,\protect\lambda ,\mathbf{A})}$]
\label{bound incr 1 Lemma copy(8)}\mbox{
}\newline
For $\omega \in \Omega $, $\lambda \in \mathbb{R}_{0}^{+}$ and $\mathbf{A}%
\in \mathbf{C}_{0}^{\infty }$, $\{\mathfrak{W}_{s,t}^{(\omega ,\lambda ,%
\mathbf{A})}\}_{s,t\in \mathbb{R}}$ has the following properties: \newline
\emph{(i)}\ It satisfies the Chapman--Kolmogorov property
\begin{equation*}
\forall t,r,s\in \mathbb{R}:\quad \mathfrak{W}_{s,t}^{(\omega ,\lambda ,%
\mathbf{A})}=\mathfrak{W}_{s,r}^{(\omega ,\lambda ,\mathbf{A})}\mathfrak{W}%
_{r,t}^{(\omega ,\lambda ,\mathbf{A})}\ .
\end{equation*}%
\emph{(ii)}\ It is the unique fundamental solution of the Cauchy initial
value\ problem
\begin{equation*}
\forall s,t\in {\mathbb{R}}:\quad \partial _{s}\mathfrak{W}_{s,t}^{(\omega
,\lambda ,\mathbf{A})}=-\delta _{s}^{(\omega ,\lambda ,\mathbf{A})}\circ
\mathfrak{W}_{s,t}^{(\omega ,\lambda ,\mathbf{A})},\quad \mathfrak{W}%
_{t,t}^{(\omega ,\lambda ,\mathbf{A})}=\mathbf{1}\ .
\end{equation*}%
\emph{(iii)}\ It solves on $\mathrm{Dom}(\delta ^{(\omega ,\lambda )})$ the
abstract Cauchy initial value\ problem
\begin{equation*}
\forall s,t\in {\mathbb{R}}:\quad \partial _{t}\mathfrak{W}_{s,t}^{(\omega
,\lambda ,\mathbf{A})}=\mathfrak{W}_{s,t}^{(\omega ,\lambda ,\mathbf{A}%
)}\circ \delta _{t}^{(\omega ,\lambda ,\mathbf{A})},\quad \mathfrak{W}%
_{s,s}^{(\omega ,\lambda ,\mathbf{A})}=\mathbf{1}\ .
\end{equation*}
\end{proposition}

\begin{proof}
Use (\ref{definition tho chap inv})--(\ref{cauchy toto2}) and an argument
similar to (\ref{uniqueness eq}). We omit the details.
\end{proof}

\subsection{Internal Energy Increment and Heat Production\label{Section
Existence}}

Recall that the internal energy increment is defined by (\ref{entropic
energy increment}), that is,
\begin{equation}
\mathbf{S}^{(\omega ,\mathbf{A})}\left( t\right) \equiv \mathbf{S}^{(\beta
,\omega ,\lambda ,\mathbf{A})}\left( t\right) :=\lim_{L\rightarrow \infty
}\left\{ \rho _{t}^{(\beta ,\omega ,\lambda ,\mathbf{A})}(H_{L}^{(\omega
,\lambda )})-\varrho ^{(\beta ,\omega ,\lambda )}(H_{L}^{(\omega ,\lambda
)})\right\}  \label{entropic energy incrementbis}
\end{equation}%
for any $\beta \in \mathbb{R}^{+}$, $\omega \in \Omega $, $\lambda \in
\mathbb{R}_{0}^{+}$, $\mathbf{A}\in \mathbf{C}_{0}^{\infty }$ and $t\in
\mathbb{R}$. To show that it is well--defined and has finite value for all
times, we use the interaction picture of the dynamics described in Theorem %
\ref{Theo int pict}:

\begin{satz}[Existence of the internal energy increment]
\label{theo exist incr sympa}\mbox{
}\newline
For any $\beta \in \mathbb{R}^{+}$, $\omega \in \Omega $, $\lambda \in
\mathbb{R}_{0}^{+}$, $\mathbf{A}\in \mathbf{C}_{0}^{\infty }$ and $t\geq
t_{0}$,
\begin{equation*}
\mathbf{S}^{(\omega ,\mathbf{A})}\left( t\right) =-i\varrho ^{(\beta ,\omega
,\lambda )}\left( \mathfrak{U}_{t}^{\ast }\delta ^{(\omega ,\lambda )}\left(
\mathfrak{U}_{t}\right) \right) \in \mathbb{R}
\end{equation*}%
with $\{\mathfrak{U}_{t}\}_{t\geq t_{0}}\subset \mathrm{Dom}(\delta
^{(\omega ,\lambda )})$ being defined in Theorem \ref{Theo int pict}.
\end{satz}

\begin{proof}
$\mathfrak{U}_{t}\in \mathrm{Dom}(\delta ^{(\omega ,\lambda )})$ and, by explicit computations using Equations (\ref{definition tho chap inv3})--(\ref{Dyson interact}) together with the "pull through" formula (\ref{extra}),
\begin{equation*}
\delta ^{(\omega ,\lambda )}\left( \mathfrak{U}_{t}\right)
=\lim_{L\rightarrow \infty }\left\{ i[H_{L}^{(\omega ,\lambda )},\mathfrak{U}%
_{t}]\right\} \in \mathcal{U}\ ,
\end{equation*}%
whereas one obviously has%
\begin{equation*}
\varrho ^{(\beta ,\omega ,\lambda )}\left( \mathfrak{U}_{t}^{\ast
}[H_{L}^{(\omega ,\lambda )},\mathfrak{U}_{t}]\right) =\rho _{t}^{(\beta
,\omega ,\lambda ,\mathbf{A})}(H_{L}^{(\omega ,\lambda )})-\varrho ^{(\beta
,\omega ,\lambda )}(H_{L}^{(\omega ,\lambda )})\ ,
\end{equation*}%
by Theorem \ref{Theo int pict}. We obtain the assertion by combining (\ref%
{entropic energy incrementbis}) with these two equalities and the continuity
of states.
\end{proof}

Therefore, $\mathbf{S}^{(\omega ,\mathbf{A})}$ is a map from $\mathbb{R}$ to
$\mathbb{R}$. Now, by the Pusz--Woronowicz theorem (see, e.g., \cite[Theorem
5.3.22]{BratteliRobinson}), it is well--known that $(\tau ,\beta )$--KMS
states $\varrho $ are \emph{passive states}, that is,%
\begin{equation*}
-i\varrho (U^{\ast }\delta (U))\geq 0
\end{equation*}%
for all unitaries $U$ both in the domain of definition of the generator $\delta $ of the group  $\tau $ of automorphisms
and in the connected component of the identity of the group of all unitary elements of the CAR algebra with the norm topology.
The latter together with Equations (\ref{definition tho chap inv3})--(\ref{Dyson interact}), Theorems %
\ref{Theo int pict} and \ref{theo exist incr sympa} directly implies the
positivity of the internal energy increment $\mathbf{S}^{(\omega ,\mathbf{A}%
)}$:

\begin{koro}[Positivity of the internal energy increment]
\label{coro heat production1}\mbox{
}\newline
For any $\beta \in \mathbb{R}^{+}$, $\omega \in \Omega $, $\lambda \in
\mathbb{R}_{0}^{+}$, $\mathbf{A}\in \mathbf{C}_{0}^{\infty }$ and all $t\geq
t_{0}$, $\mathbf{S}^{(\omega ,\mathbf{A})}\left( t\right) \geq 0$.
\end{koro}

Moreover, for any $\beta \in \mathbb{R}^{+}$, $\omega \in \Omega $, $\lambda
\in \mathbb{R}_{0}^{+}$, $\mathbf{A}\in \mathbf{C}_{0}^{\infty }$ and $t\geq
t_{0}$, we also infer from \cite[Theorem 1.1]{JaksicPillet} and Theorem \ref%
{Theo int pict} that
\begin{equation*}
-i\varrho ^{(\beta ,\omega ,\lambda )}\left( \mathfrak{U}_{t}^{\ast }\delta
^{(\omega ,\lambda )}\left( \mathfrak{U}_{t}\right) \right) =\beta ^{-1}%
\mathrm{S}(\rho _{t}^{(\beta ,\omega ,\lambda ,\mathbf{A})}|\varrho ^{(\beta
,\omega ,\lambda )})
\end{equation*}%
with $\mathrm{S}$ being the relative entropy defined by (\ref{relative
entropy general0}). See also (\ref{relative entropy general}) and recall
that $\mathrm{S}=S_{\mathcal{U}}$. By Definition \ref{Heat production
definition}, we thus recover the heat production $\mathbf{Q}^{(\omega ,%
\mathbf{A})}$ from Theorem \ref{theo exist incr sympa}:

\begin{koro}[Heat production as internal energy increment]
\label{Theorem entropy production}\mbox{
}\newline
For any $\beta \in \mathbb{R}^{+}$, $\omega \in \Omega $, $\lambda \in
\mathbb{R}_{0}^{+}$ and $\mathbf{A}\in \mathbf{C}_{0}^{\infty }$, $\mathbf{S}%
^{(\omega ,\mathbf{A})}=\mathbf{Q}^{(\omega ,\mathbf{A})}$.
\end{koro}

Finally, Theorems \ref{Theo int pict} and \ref{theo exist incr sympa} also
yield a simple and convenient expression of the \emph{total} energy
increment (\ref{lim_en_incr full})--(\ref{electro free energy}) delivered to
the system by the electromagnetic field at time $t\in \mathbb{R}$:

\begin{satz}[Total energy increment and electromagnetic work]
\label{coro heat production2}\mbox{
}\newline
For any $\beta \in \mathbb{R}^{+}$, $\omega \in \Omega $, $\lambda \in
\mathbb{R}_{0}^{+}$, $\mathbf{A}\in \mathbf{C}_{0}^{\infty }$ and $t\geq
t_{0}$,
\begin{equation*}
\mathbf{S}^{(\omega ,\mathbf{A})}\left( t\right) +\mathbf{P}^{(\omega ,%
\mathbf{A})}\left( t\right) =\int_{t_{0}}^{t}\rho _{s}^{(\beta ,\omega
,\lambda ,\mathbf{A})}\left( \partial _{s}W_{s}^{\mathbf{A}}\right) \mathrm{d%
}s\ .
\end{equation*}
\end{satz}

\begin{proof}
The proof is an extension of the one of \cite[Lemma 5.4.27.]%
{BratteliRobinson} to the \emph{unbounded} symmetric derivation $\delta
^{(\omega ,\lambda )}$.

By (\ref{definition tho chap inv3}) and the stationarity of the KMS state $%
\varrho ^{(\beta ,\omega ,\lambda )}$ w.r.t. the unperturbed dynamics (cf. (%
\ref{stationary})), we first observe that, for any $\beta \in \mathbb{R}^{+}$%
, $\omega \in \Omega $, $\lambda \in \mathbb{R}_{0}^{+}$, $\mathbf{A}\in
\mathbf{C}_{0}^{\infty }$ and $t\geq t_{0}$,%
\begin{equation}
\varrho ^{(\beta ,\omega ,\lambda )}\left( \mathfrak{U}_{t}^{\ast }\delta
^{(\omega ,\lambda )}\left( \mathfrak{U}_{t}\right) \right) =\varrho
^{(\beta ,\omega ,\lambda )}\left( \mathfrak{V}_{t,t_{0}}\delta ^{(\omega
,\lambda )}\left( \mathfrak{V}_{t,t_{0}}^{\ast }\right) \right)
\label{thm exp value1}
\end{equation}%
with the unitary elements $\mathfrak{V}_{t,t_{0}}$ being defined by (\ref%
{Dyson interact}).

The maps%
\begin{equation*}
t\mapsto \mathfrak{V}_{t,t_{0}}\qquad \text{and}\qquad t\mapsto \delta
^{(\omega ,\lambda )}\left( \mathfrak{V}_{t,t_{0}}^{\ast }\right)
\end{equation*}%
from $\mathbb{R}$ to $\mathrm{Dom}(\delta ^{(\omega ,\lambda )})$ are
continuously differentiable in the Banach spaces $\mathcal{Y}$\ and $%
\mathcal{U}$, respectively. See (\ref{graph space}) and (\ref{Dyson interact
diff eq}). Therefore, the map%
\begin{equation*}
t\mapsto \partial _{t}\left\{ \varrho ^{(\beta ,\omega ,\lambda )}\left(
\mathfrak{V}_{t,t_{0}}\delta ^{(\omega ,\lambda )}\left( \mathfrak{V}%
_{t,t_{0}}^{\ast }\right) \right) \right\}
\end{equation*}%
from $\mathbb{R}$ to $\mathbb{R}$ is also continuously differentiable and,
from (\ref{Dyson interact diff eq}) and the fact that $\delta ^{(\omega
,\lambda )}$ is a symmetric derivation, we compute that, for all $t\in
\mathbb{R}$,
\begin{equation}
\partial _{t}\left\{ \varrho ^{(\beta ,\omega ,\lambda )}\left( \mathfrak{V}%
_{t,t_{0}}\delta ^{(\omega ,\lambda )}\left( \mathfrak{V}_{t,t_{0}}^{\ast
}\right) \right) \right\} =-i\varrho ^{(\beta ,\omega ,\lambda )}\left(
\mathfrak{V}_{t,t_{0}}\{\delta ^{(\omega ,\lambda )}(W_{t,t}^{\mathbf{A}})\}%
\mathfrak{V}_{t,t_{0}}^{\ast }\right) \ .  \label{Dyson interact diff eq4}
\end{equation}%
On the other hand, using again (\ref{Dyson interact diff eq}) we observe that%
\begin{equation*}
\partial _{t}\left\{ \mathfrak{V}_{t,t_{0}}W_{t,t}^{\mathbf{A}}\mathfrak{V}%
_{t,t_{0}}^{\ast }\right\} =\mathfrak{V}_{t,t_{0}}(\partial _{t}W_{t,t}^{%
\mathbf{A}})\mathfrak{V}_{t,t_{0}}^{\ast }
\end{equation*}%
for any $t\in \mathbb{R}$, which, combined with the identity%
\begin{equation*}
\delta ^{(\omega ,\lambda )}(W_{t,t}^{\mathbf{A}})=\partial _{t}W_{t,t}^{%
\mathbf{A}}-\tau _{t}^{(\omega ,\lambda )}(\partial _{t}W_{t}^{\mathbf{A}})\
,
\end{equation*}%
yields%
\begin{equation*}
\mathfrak{V}_{t,t_{0}}\{\delta ^{(\omega ,\lambda )}(W_{t,t}^{\mathbf{A}})\}%
\mathfrak{V}_{t,t_{0}}^{\ast }=\partial _{t}\left\{ \mathfrak{V}%
_{t,t_{0}}W_{t,t}^{\mathbf{A}}\mathfrak{V}_{t,t_{0}}^{\ast }\right\} -%
\mathfrak{V}_{t,t_{0}}\tau _{t}^{(\omega ,\lambda )}(\partial _{t}W_{t}^{%
\mathbf{A}})\mathfrak{V}_{t,t_{0}}^{\ast }\ .
\end{equation*}%
Using this equality together with (\ref{Dyson interact diff eq4}) we thus
find that, for any $t\in \mathbb{R}$,
\begin{eqnarray}
&& \partial _{t}\left\{ \varrho ^{(\beta ,\omega ,\lambda )}\left( \mathfrak{V}%
_{t,t_{0}}\delta ^{(\omega ,\lambda )}\left( \mathfrak{V}_{t,t_{0}}^{\ast
}\right) \right) \right\} \label{thm exp value2} \\=-i\varrho ^{(\beta ,\omega ,\lambda )}&&\left(
\partial _{t}\left\{ \mathfrak{V}_{t,t_{0}}W_{t,t}^{\mathbf{A}}\mathfrak{V}%
_{t,t_{0}}^{\ast }\right\} \right) +i\varrho ^{(\beta ,\omega ,\lambda )}\left( \mathfrak{V}_{t,t_{0}}\tau
_{t}^{(\omega ,\lambda )}(\partial _{t}W_{t}^{\mathbf{A}})\mathfrak{V}%
_{t,t_{0}}^{\ast }\right) \ .  \notag
\end{eqnarray}%
Now, for $t\in \mathbb{R}$, we use Equations (\ref{time dependent state}), (%
\ref{stationary}), (\ref{definition tho chap inv}), (\ref{definition tho
chap inv2}), (\ref{thm exp value1}) and (\ref{thm exp value2}) to arrive at
\begin{eqnarray*}
\partial _{t}\left\{ \varrho ^{(\beta ,\omega ,\lambda )}\left( \mathfrak{U}%
_{t}^{\ast }\delta ^{(\omega ,\lambda )}\left( \mathfrak{U}_{t}\right)
\right) \right\} &=&-i\varrho ^{(\beta ,\omega ,\lambda )}\left( \partial
_{t}\left\{ \mathfrak{V}_{t,t_{0}}W_{t,t}^{\mathbf{A}}\mathfrak{V}%
_{t,t_{0}}^{\ast }\right\} \right) \\
&&+i\rho _{t}^{(\beta ,\omega ,\lambda ,\mathbf{A})}\left( \partial
_{t}W_{t}^{\mathbf{A}}\right) \ .
\end{eqnarray*}%
We next integrate this last equality by using $\mathfrak{V}_{t_{0},t_{0}}=%
\mathfrak{U}_{t_{0}}=\mathbf{1}$ to get
\begin{eqnarray}
\varrho ^{(\beta ,\omega ,\lambda )}\left( \mathfrak{U}_{t}^{\ast }\delta
^{(\omega ,\lambda )}\left( \mathfrak{U}_{t}\right) \right)
&=&i\int_{t_{0}}^{t}\rho _{s}^{(\beta ,\omega ,\lambda ,\mathbf{A})}\left(
\partial _{s}W_{s}^{\mathbf{A}}\right) \mathrm{d}s  \label{eq sympa cool} \\
&&-i\rho _{t}^{(\beta ,\omega ,\lambda ,\mathbf{A})}\left( W_{t}^{\mathbf{A}%
}\right) +i\varrho ^{(\beta ,\omega ,\lambda )}\left( W_{t_{0}}^{\mathbf{A}%
}\right)  \notag
\end{eqnarray}%
for any $\beta \in \mathbb{R}^{+}$, $\omega \in \Omega $, $\lambda \in
\mathbb{R}_{0}^{+}$, $\mathbf{A}\in \mathbf{C}_{0}^{\infty }$ and $t\geq
t_{0}$. The assertion then follows from (\ref{eq sympa cool}) combined with (%
\ref{electro free energy}) and Theorem \ref{theo exist incr sympa}.
\end{proof}

Following the terminology of \cite[Section 5.4.4.]{BratteliRobinson} with
their definition of $L^{P}$, Theorem \ref{coro heat production2} means that
the total energy increment (\ref{lim_en_incr full}) is equal to the \emph{%
work} performed on the system by the electromagnetic field at time $t\geq
t_{0}$. Moreover, Theorem \ref{coro heat production2} leads to the real
analyticity of the internal energy increment w.r.t. to the field strength $%
\eta \in \mathbb{R}$:

\begin{koro}[Real analyticity of the internal energy increment]
\label{coro heat production1 copy(1)}\mbox{
}\newline
For any $\beta \in \mathbb{R}^{+}$, $\omega \in \Omega $, $\lambda \in
\mathbb{R}_{0}^{+}$, $\mathbf{A}\in \mathbf{C}_{0}^{\infty }$ and $t\geq
t_{0}$, $\mathbf{S}^{(\omega ,\eta \mathbf{A})}\left( t\right) $ is a real
analytic function of $\eta \in \mathbb{R}$.
\end{koro}

\begin{proof}
Use Theorem \ref{coro heat production2} and write the terms $\mathbf{P}%
^{(\omega ,\eta \mathbf{A})}\left( t\right) $ and
\begin{equation*}
\int_{t_{0}}^{t}\rho _{s}^{(\beta ,\omega ,\lambda ,\eta \mathbf{A})}\left(
\partial _{s}W_{s}^{\eta \mathbf{A}}\right) \mathrm{d}s
\end{equation*}%
as Dyson--Phillips series in terms of multi--commutators, see (\ref{time
dependent state}) and (\ref{Dyson tau 1}). Observe finally that both maps
\begin{equation*}
\eta \mapsto W_{s}^{\eta \mathbf{A}}\in \mathcal{U}\text{\qquad and\qquad }%
\eta \mapsto \partial _{s}W_{s}^{\eta \mathbf{A}}\in \mathcal{U}
\end{equation*}%
are real analytic with infinite analyticity radius.
\end{proof}

\subsection{Behavior of the Internal Energy Increment at Small Fields\label%
{section Energy Increments as Power Series}}

We study here the asymptotic behavior of $\mathbf{S}^{(\omega ,\eta \mathbf{A%
}_{l})}\equiv \mathbf{S}^{(\beta ,\omega ,\lambda ,\eta \mathbf{A}_{l})}$ at
small field strength $\eta \in \mathbb{R}$ and large space scale $l\in
\mathbb{R}^{+}$. In fact, in view of Corollary \ref{Theorem entropy
production} saying that $\mathbf{S}^{(\omega ,\eta \mathbf{A}_{l})}=\mathbf{Q%
}^{(\omega ,\eta \mathbf{A}_{l})}$, we prove here Theorem \ref{Thm Heat
production as power series}. Recall that $\mathbf{A}_{l}\in \mathbf{C}%
_{0}^{\infty }$ is defined by (\ref{rescaled vector potential}), that is,
\begin{equation}
\mathbf{A}_{l}(t,x):=\mathbf{A}(t,l^{-1}x)\ ,\quad t\in \mathbb{R},\ x\in
\mathbb{R}^{d}\ ,  \label{26bis}
\end{equation}%
for any $\mathbf{A}\in \mathbf{C}_{0}^{\infty }$ and $l\in \mathbb{R}^{+}$.

Using Equations (\ref{discrete laplacian}), (\ref{time dependent state}), (%
\ref{def H loc}), (\ref{stationary}) and (\ref{Dyson tau 1}) we first
observe that%
\begin{eqnarray}
&&\rho _{t}^{(\beta ,\omega ,\lambda ,\eta \mathbf{A}_{l})}(H_{L}^{(\omega
,\lambda )})-\rho _{t_{0}}^{(\beta ,\omega ,\lambda ,\eta \mathbf{A}%
_{l})}(H_{L}^{(\omega ,\lambda )})  \notag \\
&=&\sum\limits_{x\in \Lambda _{L}}\sum\limits_{z\in \mathfrak{L},|z|\leq
1}\langle \mathfrak{e}_{x},\left( \Delta _{\mathrm{d}}+\lambda V_{\omega
}\right) \mathfrak{e}_{x+z}\rangle \mathbf{1}[x+z\in \Lambda
_{L}]\sum\limits_{k\in {\mathbb{N}}}i^{k}  \label{incr lim 3} \\
&&\times \int_{t_{0}}^{t}\mathrm{d}s_{1}\int_{t_{0}}^{s_{1}}\mathrm{d}%
s_{2}\cdots \int_{t_{0}}^{s_{k-1}}\mathrm{d}s_{k}  \notag \\
&&\varrho ^{(\beta ,\omega ,\lambda )}\left( [W_{s_{k}-t_{0},s_{k}}^{\eta
\mathbf{A}_{l}},\ldots ,W_{s_{1}-t_{0},s_{1}}^{\eta \mathbf{A}_{l}},\tau
_{t-t_{0}}^{(\omega ,\lambda )}(a_{x}^{\ast }a_{x+z})]^{(k+1)}\right)  \notag
\end{eqnarray}%
for any $L,l,\beta \in \mathbb{R}^{+}$, $\omega \in \Omega $, $\lambda \in
\mathbb{R}_{0}^{+}$, $\eta \in \mathbb{R}$, $\mathbf{A}\in \mathbf{C}%
_{0}^{\infty }$ and $t\geq t_{0}$. Recall that the time--dependent
electromagnetic perturbation $W_{t,s}^{\mathbf{A}}$ is defined by (\ref{def
LA}). See also (\ref{multi1})--(\ref{multi2}) for the precise definition of
multi--commutators.

Therefore, in order to write $\mathbf{S}^{(\omega ,\eta \mathbf{A}_{l})}$ in
terms of multi--commutators, we prove the following lemma by using
tree--decay bounds:

\begin{lemma}[Bounds on multi--commutators]
\label{bound incr 1 Lemma}\mbox{
}\newline
For any $\mathbf{A}\in \mathbf{C}_{0}^{\infty }$, there is $\eta _{0}\in
\mathbb{R}^{+}$ such that, for any $l,\varepsilon \in \mathbb{R}^{+}$, there
is a ball%
\begin{equation}
B(0,R):=\{x\in \mathfrak{L}:|x|\leq R\}  \label{ball}
\end{equation}%
of radius $R\in \mathbb{R}^{+}$ centered at $0$ such that, for all $|\eta
|\in \lbrack 0,\eta _{0}]$, $\beta \in \mathbb{R}^{+}$, $\omega \in \Omega $%
, $\lambda \in \mathbb{R}_{0}^{+}$ and $t_{0}\leq s_{1},\ldots ,s_{k}\leq t$%
,
\begin{multline*}
\sum\limits_{x\in \Lambda _{L}\backslash B_{R}}\sum\limits_{z\in \mathfrak{L}%
,|z|\leq 1}\sum\limits_{k\in {\mathbb{N}}}\frac{\left( t-t_{0}\right) ^{k}}{%
k!} \\
\left\vert \varrho ^{(\beta ,\omega ,\lambda )}\left(
[W_{s_{k}-t_{0},s_{k}}^{\eta \mathbf{A}_{l}},\ldots
,W_{s_{1}-t_{0},s_{1}}^{\eta \mathbf{A}_{l}},\tau _{t-t_{0}}^{(\omega
,\lambda )}(a_{x}^{\ast }a_{x+z})]^{(k+1)}\right) \right\vert \leq
\varepsilon \ .
\end{multline*}
\end{lemma}

\begin{proof}
We first need to bound the $(k+1)$--fold multi--commutator
\begin{equation*}
\lbrack W_{s_{k}-t_{0},s_{k}}^{\mathbf{A}},\ldots ,W_{s_{1}-t_{0},s_{1}}^{%
\mathbf{A}},\tau _{t-t_{0}}^{(\omega ,\lambda )}(a_{x}^{\ast
}a_{x+z})]^{(k+1)}
\end{equation*}%
for any $k\in {\mathbb{N}}$, $x\in \Lambda _{L}$ and $z\in \mathfrak{L}$ so
that $|z|\leq 1$. This is done by using tree--decay bounds as explained in
Section \ref{section Tree--decay Bounds}. Indeed, by (\ref{26bis}), for any $%
l\in \mathbb{R}^{+}$ and $\mathbf{A}\in \mathbf{C}_{0}^{\infty }$, there
exists a finite subset $\widetilde{\Lambda }_{l}\in \mathcal{P}_{f}(%
\mathfrak{L})$ such that $\mathbf{A}_{l}(t,x)=0$ for all $x\in \mathfrak{L}%
\backslash \widetilde{\Lambda }_{l}$ and $t\in \mathbb{R}$. Then, we infer
from (\ref{def:W}) and (\ref{def LA}) that, for all $l\in \mathbb{R}^{+}$, $%
x,y\in \mathfrak{L}$, $\mathbf{A}\in \mathbf{C}_{0}^{\infty }$ and $t,\eta
\in \mathbb{R}$, there are constants $D_{x,y}^{\eta \mathbf{A}_{l}}(t)\in
\mathbb{C}$ such that%
\begin{equation}
W_{s_{1},s_{2}}^{\eta \mathbf{A}_{l}}=\sum\limits_{x\in \widetilde{\Lambda }%
_{l}}\sum\limits_{z\in \mathfrak{L},|z|\leq 1}D_{x,x+z}^{\eta \mathbf{A}%
_{l}}(s_{2})\tau _{s_{1}}^{(\omega ,\lambda )}\left( a_{x}^{\ast
}a_{x+z}\right)  \label{eq:W sum}
\end{equation}%
for any $\omega \in \Omega $, $\lambda \in \mathbb{R}_{0}^{+}$ and $%
s_{1},s_{2}\in \mathbb{R}$. Here, the constants $D_{x,y}^{\eta \mathbf{A}%
_{l}}(t)$ are always of order $\eta $:%
\begin{equation}
\underset{t\in \mathbb{R}\ ,\ x,y\in \mathfrak{L}}{\sup }\left\vert
D_{x,y}^{\eta \mathbf{A}_{l}}\left( t\right) \right\vert \leq K_{\eta }
\label{eq:W sumbisbis0}
\end{equation}%
with%
\begin{equation}
K_{\eta }:=\left\Vert \Delta _{\mathrm{d}}\right\Vert _{\mathrm{op}%
}\left\vert \exp \left\{ i\left\vert \eta \right\vert \underset{\left(
t,x\right) \in \mathbb{R}\times \mathbb{R}^{d}\ ,\ z\in \mathfrak{L},|z|\leq
1}{\max }\left\vert \left[ \mathbf{A}(t,x)\right] \left( z\right)
\right\vert \right\} -1\right\vert =\mathcal{O}\left( \left\vert \eta
\right\vert \right) \ .  \label{eq:W sumbisbis}
\end{equation}%
(Recall that $\left\Vert \cdot \right\Vert _{\mathrm{op}}$ is the operator
norm.) Therefore, using Corollary \ref{tree bound main copy(1)} we deduce
that, for every $\epsilon \in \mathbb{R}^{+}$, $\mathbf{A}\in \mathbf{C}%
_{0}^{\infty }$ and $t>t_{0}$, there is a constant $D\in \mathbb{R}^{+}$
such that, for any $k\in {\mathbb{N}}$, $L,l,\beta \in \mathbb{R}^{+}$, $%
\omega \in \Omega $, $\lambda \in \mathbb{R}_{0}^{+}$, $\eta \in \mathbb{R}$%
, $s_{1},\ldots ,s_{k}\in \lbrack t_{0},t]$ and $R>R_{l}$,
\begin{eqnarray}
\sum\limits_{x\in \Lambda _{L}\backslash B_{R}}\sum\limits_{z\in \mathfrak{%
L},|z|\leq 1} && \left\vert \varrho ^{(\beta ,\omega ,\lambda )}\left(
[W_{s_{k}-t_{0},s_{k}}^{\eta \mathbf{A}_{l}},\ldots
,W_{s_{1}-t_{0},s_{1}}^{\eta \mathbf{A}_{l}},\tau _{t-t_{0}}^{(\omega
,\lambda )}(a_{x}^{\ast }a_{x+z})]^{(k+1)}\right) \right\vert  \notag \\
&\leq &|\widetilde{\Lambda }_{l}|\ |\mathcal{T}_{k+1}|\ \left[
\sum\limits_{x\in \mathfrak{L},|x|\geq R-R_{l}}\frac{K_{\eta }D}{%
1+|x|^{d+\epsilon }}\right] \left[ \sum\limits_{x\in \mathfrak{L}}\frac{%
K_{\eta }D}{1+|x|^{d+\epsilon }}\right] ^{k-1}\ ,  \label{tototree1}
\end{eqnarray}%
with $B(0,R)$ being the ball (\ref{ball}) of radius $R\in \mathbb{R}^{+}$
centered at $0$ and where $|\widetilde{\Lambda }_{l}|$ is the volume of the
finite subset $\widetilde{\Lambda }_{l}\in $ $\mathcal{P}_{f}(\mathfrak{L})$
with radius%
\begin{equation}
R_{l}:=\max \left\{ \left\vert x\right\vert :x\in \widetilde{\Lambda }%
_{l}\right\} \in \mathbb{R}^{+}\ ,\qquad l\in \mathbb{R}^{+}\ .
\label{eq:W sumbis}
\end{equation}%
Note that there exists a finite constant $D\in \mathbb{R}^{+}$ such that $%
R_{l}\leq lD$ for all $l\in \mathbb{R}^{+}$.

From (\ref{def:W}) and (\ref{def LA}) it follows that $W_{t,s}^{\mathbf{A}%
}=0 $ for any $t\geq t_{1}$, where $t_{1}$ is the time when the
electromagnetic potential is switched off. Therefore, without loss of
generality (w.l.o.g.) we only consider times $t\in (t_{0},t_{1}]$ with $%
t_{1}>t_{0}$. Thus, take $\eta _{0}\in \mathbb{R}^{+}$ sufficiently small to
imply
\begin{equation*}
\sum\limits_{x\in \mathfrak{L}}\frac{K_{\eta }D}{1+|x|^{d+\epsilon }}\leq
\sum\limits_{x\in \mathfrak{L}}\frac{K_{\eta _{0}}D}{1+|x|^{d+\epsilon }}%
\leq \frac{1}{2\left( t_{1}-t_{0}\right) }
\end{equation*}%
for all $|\eta |\in \lbrack 0,\eta _{0}]$. Then, using $|\mathcal{T}%
_{k+1}|=k!$ and the upper bound (\ref{tototree1}) we arrive at%
\begin{eqnarray}
\sum\limits_{x\in \Lambda _{L}\backslash B_{R}}\sum\limits_{z\in \mathfrak{%
L},|z|\leq 1}&&\left\vert \varrho ^{(\beta ,\omega ,\lambda )}\left(
[W_{s_{k}-t_{0},s_{k}}^{\eta \mathbf{A}_{l}},\ldots
,W_{s_{1}-t_{0},s_{1}}^{\eta \mathbf{A}_{l}},\tau _{t-t_{0}}^{(\omega
,\lambda )}(a_{x}^{\ast }a_{x+z})]^{(k+1)}\right) \right\vert  \notag \\
&\leq &\frac{k!}{2^{k-1}\left( t_{1}-t_{0}\right) ^{k-1}}|\widetilde{\Lambda
}_{l}|\sum\limits_{x\in \mathfrak{L},|x|\geq R-R_{l}}\frac{K_{\eta }D}{%
1+|x|^{d+\epsilon }}  \label{tototree2}
\end{eqnarray}%
for all $|\eta |\in \lbrack 0,\eta _{0}]$ and any $L,l,\beta \in \mathbb{R}%
^{+}$, $\omega \in \Omega $, $\lambda \in \mathbb{R}_{0}^{+}$, $k\in \mathbb{%
N}$, $t\in (t_{0},t_{1}]$ and $s_{1},\ldots ,s_{k}\in \lbrack t_{0},t]$.
Therefore, we get the assertion from (\ref{tototree2}) by choosing $R\in
\mathbb{R}^{+}$ such that
\begin{equation*}
2\left( t_{1}-t_{0}\right) |\widetilde{\Lambda }_{l}|\sum\limits_{x\in
\mathfrak{L},|x|\geq R-R_{l}}\frac{K_{\eta _{0}}D}{1+|x|^{d+\epsilon }}\leq
\varepsilon
\end{equation*}%
for some fixed arbitrarily chosen parameter $\varepsilon \in \mathbb{R}^{+}$.
\end{proof}

For any $\mathbf{A}\in \mathbf{C}_{0}^{\infty }$, this lemma implies the
existence of a constant $\eta _{0}\in \mathbb{R}^{+}$ such that, for all $%
|\eta |\in \lbrack 0,\eta _{0}]$, $l,\beta \in \mathbb{R}^{+}$, $\omega \in
\Omega $, $\lambda \in \mathbb{R}_{0}^{+}$ and $t\geq t_{0}$, the limit (\ref%
{entropic energy incrementbis}) equals%
\begin{eqnarray}
\mathbf{S}^{(\omega ,\eta \mathbf{A}_{l})}\left( t\right)
&=&\sum\limits_{k\in {\mathbb{N}}}\sum\limits_{x,z\in \mathfrak{L},|z|\leq
1}\langle \mathfrak{e}_{x},\left( \Delta _{\mathrm{d}}+\lambda V_{\omega
}\right) \mathfrak{e}_{x+z}\rangle i^{k}\int_{t_{0}}^{t}\mathrm{d}%
s_{1}\cdots \int_{t_{0}}^{s_{k-1}}\mathrm{d}s_{k}  \notag \\
&&\varrho ^{(\beta ,\omega ,\lambda )}\left( [W_{s_{k}-t_{0},s_{k}}^{\eta
\mathbf{A}_{l}},\ldots ,W_{s_{1}-t_{0},s_{1}}^{\eta \mathbf{A}_{l}},\tau
_{t-t_{0}}^{(\omega ,\lambda )}(a_{x}^{\ast }a_{x+z})]^{(k+1)}\right) \ . \label{energy increment bis}
\end{eqnarray}%
This series is absolutely convergent, by Lemma \ref{bound incr 1 Lemma}.
This proves Theorem \ref{Thm Heat production as power series} (i) because of
Corollary \ref{Theorem entropy production}.

By Corollary \ref{coro heat production1 copy(1)}, recall that, for any $%
l,\beta \in \mathbb{R}^{+}$, $\omega \in \Omega $, $\lambda \in \mathbb{R}%
_{0}^{+}$, $\mathbf{A}\in \mathbf{C}_{0}^{\infty }$ and $t\geq t_{0}$, $%
\mathbf{S}^{(\omega ,\eta \mathbf{A}_{l})}\left( t\right) $ is a real
analytic function of $\eta \in \mathbb{R}$. Now, we use (\ref{energy
increment bis}) to bound the Taylor coefficients of the function $\eta
\mapsto \mathbf{S}^{(\omega ,\eta \mathbf{A}_{l})}\left( t\right) $ at $\eta
=0$, i.e., we prove Theorem \ref{Thm Heat production as power series} (ii):

\begin{lemma}[Analytic norm of the internal energy increment]
\label{bound incr 1 Lemma copy(1)}\mbox{
}\newline
For any $\mathbf{A}\in \mathbf{C}_{0}^{\infty }$, there exist $\eta
_{1},D,\varepsilon \in \mathbb{R}^{+}$ that depend on $\mathbf{A}$ such
that, for all $l,\beta \in \mathbb{R}^{+}$, $\omega \in \Omega $, $\lambda
\in \mathbb{R}_{0}^{+}$ and $t\geq t_{0}$,
\begin{equation*}
\underset{m=0}{\overset{\infty }{\sum }}\frac{\eta _{1}^{m}}{m!}\ \underset{%
\eta \in \lbrack -\varepsilon ,\varepsilon ]}{\sup }\left\vert \partial
_{\eta }^{m}\mathbf{S}^{(\omega ,\eta \mathbf{A}_{l})}\left( t\right)
\right\vert \leq Dl^{d}\ .
\end{equation*}
\end{lemma}

\begin{proof}
Similar to the derivation of (\ref{tototree2}),\ for any $\mathbf{A}\in
\mathbf{C}_{0}^{\infty }$, there are constants $\eta _{1},D,\varepsilon \in
\mathbb{R}^{+}$ such that, for any $L,l,\beta \in \mathbb{R}^{+}$, $\omega
\in \Omega $, $\lambda \in \mathbb{R}_{0}^{+}$, $k\in \mathbb{N}$, $t\in
(t_{0},t_{1}]$ and $s_{1},\ldots ,s_{k}\in \lbrack t_{0},t]$,%
\begin{multline*}
\sum\limits_{x,z\in \mathfrak{L},|z|\leq 1}\underset{m=0}{\overset{\infty }{%
\sum }}\frac{\eta _{1}^{m}}{m!}\ \underset{\eta \in \lbrack -\varepsilon
,\varepsilon ]}{\sup }\left\vert \partial _{\eta }^{m}\left\{ \varrho
^{(\beta ,\omega ,\lambda )}\left( \left[ W_{s_{k}-t_{0},s_{k}}^{\eta
\mathbf{A}_{l}},\ldots \right. \right. \right. \right. \\
\left. \left. \left. \left. \ldots ,W_{s_{1}-t_{0},s_{1}}^{\eta \mathbf{A}%
_{l}},\tau _{t-t_{0}}^{(\omega ,\lambda )}(a_{x}^{\ast }a_{x+z})\right]
^{(k+1)}\right) \right\} \right\vert \leq \frac{Dl^{d}k!}{2^{k-1}\left(
t_{1}-t_{0}\right) ^{k-1}}\ .
\end{multline*}%
Now, use (\ref{energy increment bis}) together with fact that the $\eta $%
--derivative $\partial _{\eta }$ is a closed operator w.r.t. to the norm of
uniform convergence to arrive at the assertion.
\end{proof}

\appendix

\section{Relative Entropy -- Thermodynamic Limit\label{Section appendix}}

We give in the first subsection a concise account on the relative entropy in
$C^{\ast }$--algebras. In the second subsection we show that the properties
of the infinite fermion system result from features of the finite volume
one, at large volume.

\subsection{Quantum Relative Entropy\label{Section Quantum Relative Entropy}}

\subsubsection{Spacial Derivative Operator}

Although the relative entropy\ can be defined for states on general $C^{\ast
}$--algebras, it is natural to start with the special case of von Neumann
algebras, which are (generally) non--commutative analogues of the algebra of
bounded measurable functions. The definition of quantum relative entropy
also requires the concept of \emph{spacial derivative} operator. The latter
has been first introduced by Connes \cite{Connes} as a generalization of the
relative modular operator. It is the non--commutative analogue of the
Radon--Nikodym derivative of two measures defined as follows.

Let $\rho \in \mathfrak{M}^{\ast }$ be any normal state of a von Neumann
algebra $\mathfrak{M}$ acting on a Hilbert space $\mathcal{H}$. We denote
the so--called \emph{lineal} of $\rho $ by
\begin{equation}
\mathcal{D}_{\rho }:=\left\{ \psi \in \mathcal{H}:\left\langle \psi
,bb^{\ast }\psi \right\rangle _{\mathcal{H}}\leq D_{\psi }\rho \left(
bb^{\ast }\right) \text{ for all }b\in \mathfrak{M}\text{ and some }D_{\psi
}\in \mathbb{R}^{+}\right\} \ .  \label{lineal}
\end{equation}%
Similar to \cite[Lemma 2]{Connes} which is restricted to faithful states,
this subspace of $\mathcal{H}$ is dense in $\mathrm{supp}\left( \rho \right)
$. Here, by abuse of notation, $\mathrm{supp}\left( \rho \right) $ is
defined to be either the smallest projection $\mathrm{P}$ such that $\rho (%
\mathrm{P})=1$ or the range of this projection $\mathrm{P}$.

Let $(\mathcal{H}_{\rho },\pi _{\rho },\Psi _{\rho })$ be the GNS
representation of the state $\rho $. For any $\psi \in \mathcal{D}_{\rho }$,
there is a bounded operator $R_{\rho }(\psi ):\mathcal{H}_{\rho }\rightarrow
\mathcal{H}$ such that
\begin{equation}
R_{\rho }(\psi )\pi _{\rho }\left( b\right) \Psi _{\rho }=b\psi \ ,\qquad
b\in \mathfrak{M}\ .  \label{R operator}
\end{equation}%
Clearly, for any $b\in \mathfrak{M}$, $bR_{\rho }(\psi )=R_{\rho }(\psi )\pi
_{\rho }\left( b\right) $. This yields%
\begin{equation*}
\Theta _{\rho }(\psi ,\tilde{\psi}):=R_{\rho }(\psi )R_{\rho }(\tilde{\psi}%
)^{\ast }\in \mathfrak{M}^{\prime }\ ,\qquad \psi ,\tilde{\psi}\in \mathcal{D%
}_{\rho }\ .
\end{equation*}

Let $\varpi $ be a fixed normal state on $\mathfrak{M}^{\prime }$. For any $%
\psi ,\tilde{\psi}\in \mathcal{D}_{\rho }$ and $\psi _{\bot },\tilde{\psi}%
_{\bot }\in \mathcal{D}_{\rho }^{\bot }$, we define the quadratic form $q$ by%
\begin{equation}
q_{\varpi ,\rho }(\psi +\psi _{\bot },\tilde{\psi}+\tilde{\psi}_{\bot
}):=\varpi \left( \Theta _{\rho }(\psi ,\tilde{\psi})\right) \ .
\label{quadratic form}
\end{equation}%
Similar to what it is done in \cite[Lemmata 5 and 6]{Connes}, where the
state $\rho $ is faithful, $q_{\varpi ,\rho }$ is a positive densely defined
quadratic form. Moreover, it is closable. In particular, by \cite[Theorem
VIII.15]{ReedSimonI}, there is a unique positive self--adjoint operator $%
\partial _{\rho }\varpi $ acting on $\mathcal{H}$ such that the domain $%
\mathrm{Dom}\left( q\right) $ is a core for $\left( \partial _{\rho }\varpi
\right) ^{1/2}$and
\begin{equation*}
q_{\varpi ,\rho }\left( \psi ,\psi \right) =\left\langle \left( \partial
_{\rho }\varpi \right) \psi ,\psi \right\rangle _{\mathcal{H}}<\infty \
,\qquad \psi \in \mathrm{Dom}\left( q\right) \ .
\end{equation*}%
Let $\mathrm{supp}\left( \partial _{\rho }\varpi \right) $ be the orthogonal
projection on the range of $\partial _{\rho }\varpi $. By \cite[Eq. (4.4)]%
{OhyaPetz},
\begin{equation}
\mathrm{supp}\left( \partial _{\rho }\varpi \right) =\mathrm{supp}\left(
\varpi \right) \mathrm{supp}\left( \rho \right) \ .  \label{support 0}
\end{equation}

$\partial _{\rho }\varpi $\ is named the \emph{spacial derivative} operator
and can be seen as a \emph{non--commutative Radon--Nikodym derivative}, see
\cite{Connes}. For instance, at fixed state $\rho $, it is additive in $%
\varpi $. Since $\mathfrak{M}$ and $\mathfrak{M}^{\prime }$ have symmetric
roles, the spatial derivative $\partial _{\varpi }\rho $ can be defined as
well and one finds that
\begin{equation}
\partial _{\varpi }\rho =\left( \partial _{\rho }\varpi \right) ^{-1}\ ,
\label{truc cool}
\end{equation}%
under the convention that, for any operator $B$, $B^{-1}\equiv 0$ on the
subspace where $B=0$. Moreover, as it is explained in \cite[Chapter 4]%
{OhyaPetz}, for faithful states, $\partial _{\rho }\varpi $ is nothing else
than the \emph{relative modular operator} $\mathbf{\Delta }\left( \varpi
,\rho \right) $.

\subsubsection{Relative Entropy for States on $C^{\ast }$--Algebras}

Let $\mathcal{X}$ be a $C^{\ast }$--algebra and $\rho _{2}\in \mathcal{X}%
^{\ast }$ be any reference state with GNS representation $(\mathcal{H}_{\rho
_{2}},\pi _{\rho _{2}},\Psi _{\rho _{2}})$. Let $\tilde{\rho}_{2}\in
\mathfrak{M}^{\ast }$ be the normal state of the von Neumann algebra $%
\mathfrak{M}:=\pi _{\rho _{2}}\left( \mathcal{X}\right) ^{\prime \prime }$
that is defined by extension from $\rho _{2}\in \mathcal{X}^{\ast }$. Take
any state $\rho _{1}\in \mathcal{X}^{\ast }$ which is \emph{quasi--contained}
in $\rho _{2}$, that is, there exists a normal state $\tilde{\rho}_{1}\in
\mathfrak{M}^{\ast }$ such that
\begin{equation*}
\tilde{\rho}_{1}\left( \pi _{\rho _{2}}\left( B\right) \right) =\rho
_{1}\left( B\right) \ ,\qquad B\in \mathcal{X}\ .
\end{equation*}%
Then, by \cite[Theorems 2.4.21 and 2.5.31]{BratteliRobinsonI}, there is $%
\Psi _{\rho _{1}}\in \mathcal{H}_{\rho _{2}}$ such that
\begin{equation}
\tilde{\rho}_{1}\left( \pi _{\rho _{2}}\left( B\right) \right) =\left\langle
\Psi _{\rho _{1}},\pi _{\rho _{2}}\left( B\right) \Psi _{\rho
_{1}}\right\rangle _{\mathcal{H}_{\rho _{2}}}\ ,\qquad B\in \mathcal{X}\ .
\label{vector representing}
\end{equation}%
Moreover, $\Psi _{\rho _{1}}\in \mathcal{H}_{\rho _{2}}$ induces a vector
state $\tilde{\rho}_{1}^{\prime }$\ on the commutant $\mathfrak{M}^{\prime }$
of $\mathfrak{M}$. Then, from (\ref{lineal}) and (\ref{R operator}), observe
that $\mathcal{D}_{\tilde{\rho}_{1}^{\prime }}=\mathfrak{M}\Psi _{\rho _{1}}$%
,%
\begin{equation}
R_{\rho }(b\Psi _{\rho _{1}})\pi _{\rho }\left( b^{\prime }\right) \Psi
_{\rho }=b\left( b^{\prime }\Psi _{\rho _{1}}\right) \ ,\qquad b^{\prime
}\in \mathfrak{M}^{\prime }\ ,\ b\in \mathfrak{M}\ ,  \label{R operatorbis}
\end{equation}%
and the spacial derivative operator $\partial _{\tilde{\rho}_{1}^{\prime }}%
\tilde{\rho}_{2}$ is a well--defined positive self--adjoint operator acting
on $\mathcal{H}_{\rho _{2}}$. By (\ref{support 0}), its support, seen as an
orthogonal projection, equals
\begin{equation}
\mathrm{supp}\left( \partial _{\tilde{\rho}_{1}^{\prime }}\tilde{\rho}%
_{2}\right) =\mathrm{supp}\left( \tilde{\rho}_{1}^{\prime }\right) \ .
\label{definition support}
\end{equation}

Then, Araki's definition of relative entropy takes the following form:
\begin{equation}
S_{\mathcal{X}}\left( \rho _{1}|\rho _{2}\right) :=-\left\langle \ln
(\partial _{\tilde{\rho}_{1}^{\prime }}\tilde{\rho}_{2})\Psi _{\rho
_{1}},\Psi _{\rho _{1}}\right\rangle _{\mathcal{H}_{\rho _{2}}}=-\rho
_{1}(\ln (\partial _{\tilde{\rho}_{1}^{\prime }}\tilde{\rho}_{2}))\in
\mathbb{R}_{0}^{+}\ ,  \label{relative entropy general}
\end{equation}%
see \cite[Eq. (5.1)]{OhyaPetz}. This definition is sound because of (\ref%
{definition support}) and
\begin{equation*}
\Psi _{\rho _{1}}=\mathrm{supp}\left( \tilde{\rho}_{1}^{\prime }\right) \Psi
_{\rho _{1}}.
\end{equation*}
If the state $\rho _{1}\in \mathcal{U}^{\ast }$ is \emph{not}
quasi--contained in $\rho _{2}$, then the relative entropy of $\rho _{1}$
w.r.t. $\rho _{2}$ is by definition infinite, i.e., $S_{\mathcal{X}}\left(
\rho _{1}|\rho _{2}\right) :=+\infty $. However, this case never appears in
this paper. By the Uhlmann monotonicity theorem \cite[Theorem 5.3]{OhyaPetz}%
, note that this definition does not depend on the choice of the vector $%
\Psi _{\rho _{1}}\in \mathcal{H}_{\rho _{2}}$ representing $\tilde{\rho}_{1}$
via (\ref{vector representing}).

The \emph{quantum} relative entropy $S_{\mathcal{X}}$ is the analogue of the
relative entropy defined for probability measures on a Polish space. Compare
formally (\ref{truc cool}) and (\ref{relative entropy general}) with \cite[%
Eq. (6.2.8)]{Dembo}. The positivity of the relative entropy as well as the
equivalence relation between the two assertions $S_{\mathcal{X}}\left( \rho
_{1}|\rho _{2}\right) =0$ and $\rho _{1}=\rho _{2}$ both follow from \cite[%
Theorem 5.5]{OhyaPetz}. However, like for probability measures, neither $S_{%
\mathcal{X}}$ nor its symmetric version is a metric.

\subsubsection{Relative Entropy for States on Full Matrix Algebras}

In the case where $\mathcal{X}$ is a full matrix algebra $\mathcal{B}(%
\mathbb{C}^{n})$ for some $n\in \mathbb{N}$, the relative entropy $S_{%
\mathcal{X}}$ has a simple explicit expression. Note that any finite
dimensional $C^{\ast }$--algebra is isomorphic to a direct sum of full
matrix algebras and Lemma \ref{local AC-conductivity lemma copy(2)} has a
straighforward generalization to that case.

We denote by $\mathrm{tr}$ the normalized trace of $\mathcal{B}(\mathbb{C}%
^{n})$. For any state $\rho \in \mathcal{B}(\mathbb{C}^{n})^{\ast }$, there
is a unique adjusted density matrix $\mathrm{d}_{\rho }\in \mathcal{B}(%
\mathbb{C}^{n})$, that is, $\mathrm{d}_{\rho }\geq 0$, $\mathrm{tr}\left(
\mathrm{d}_{\rho }\right) =1$ and $\rho (A)=\mathrm{tr}\left( \mathrm{d}%
_{\rho }A\right) $ for all $A\in \mathcal{B}(\mathbb{C}^{n})$, see \cite[%
Lemma 3.1 (i)]{Araki-Moriya}. Then, by using an explicit GNS representation
of $\rho _{2}$ one can explicitly compute the spatial derivative operator $%
\partial _{\tilde{\rho}_{1}^{\prime }}\tilde{\rho}_{2}$ and, under the
convention $x\ln x|_{x=0}:=0$, one explicitly finds the relative entropy $S_{%
\mathcal{B}(\mathbb{C}^{n})}$ of any state $\rho _{1}\in \mathcal{B}(\mathbb{%
C}^{n})^{\ast }$ w.r.t. $\rho _{2}\in \mathcal{B}(\mathbb{C}^{n})^{\ast }$:

\begin{lemma}[Relative entropy - Finite dimensional case]
\label{local AC-conductivity lemma copy(2)}\mbox{
}\newline
Let $n\in \mathbb{N}$. For any state $\rho _{1},\rho _{2}\in \mathcal{B}(%
\mathbb{C}^{n})^{\ast }$, the relative entropy $S_{\mathcal{B}(\mathbb{C}%
^{n})}$ defined by (\ref{relative entropy general}) is equal to%
\begin{equation*}
S_{\mathcal{B}(\mathbb{C}^{n})}\left( \rho _{1}|\rho _{2}\right) =\left\{
\begin{array}{lll}
\rho _{1}\left( \ln \mathrm{d}_{\rho _{1}}-\ln \mathrm{d}_{\rho _{2}}\right)
\in \mathbb{R}_{0}^{+} & , & \qquad \text{if \ }\mathrm{supp}\left( \rho
_{2}\right) \geq \mathrm{supp}\left( \rho _{1}\right) \ . \\
+\infty & , & \qquad \text{otherwise}\ .%
\end{array}%
\right.
\end{equation*}
\end{lemma}

\begin{proof}
We give the proof for completeness and because it is instructive. Take two
states $\rho _{1},\rho _{2}\in \mathcal{B}(\mathbb{C}^{n})^{\ast }$. If $%
\rho _{1}$ is not quasi--contained in $\rho _{2}$ then clearly, $\mathrm{supp%
}\left( \rho _{2}\right) \ngeq \mathrm{supp}\left( \rho _{1}\right) $ and $%
S_{\mathcal{B}(\mathbb{C}^{n})}\left( \rho _{1}|\rho _{2}\right) =+\infty $.

Assume w.l.o.g. that $\rho _{2}$ is faithful. (Otherwise, one has to take a
subspace of $\mathcal{B}(\mathbb{C}^{n})$.) In particular, any state $\rho
_{1}$ is quasi--contained in $\rho _{2}$. The GNS representation $(\mathcal{H%
}_{\rho _{2}},\pi _{\rho _{2}},\Psi _{\rho _{2}})$ of $\rho _{2}$ is, in
this case, explicitly given as follows: $\mathcal{H}_{\rho _{2}}$
corresponds to the linear space $\mathcal{B}(\mathbb{C}^{n})$ endowed with
the Hilbert--Schmidt scalar product
\begin{equation}
\left\langle A,B\right\rangle _{\mathcal{H}_{\rho _{2}}}:=\mathrm{Trace}_{%
\mathbb{C}^{n}}(A^{\ast }B)\ ,\qquad A,B\in \mathcal{B}(\mathbb{C}^{n})\ .
\label{trace h at}
\end{equation}%
It is convenient to define left and right multiplication operators on $%
\mathcal{B}(\mathbb{C}^{n})$: For any $A\in \mathcal{B}(\mathbb{C}^{n})$ we
define the linear operators $\underrightarrow{A}$ and $\underleftarrow{A}$
acting on $\mathcal{B}(\mathbb{C}^{n})$ by
\begin{equation}
B\mapsto \underrightarrow{A}B:=AB\qquad \text{and}\qquad B\mapsto
\underleftarrow{A}B:=BA\ .  \label{left multi}
\end{equation}%
The representation $\pi _{\rho _{2}}$ is the left multiplication, i.e.,
\begin{equation*}
\pi _{\rho _{2}}\left( A\right) :=\underrightarrow{A}\ ,\qquad A\in \mathcal{%
B}(\mathbb{C}^{n})\ .
\end{equation*}%
The cyclic vector of the GNS representation of $\rho _{2}$ is defined by
using the density matrix $\mathrm{D}_{\rho _{2}}\in \mathcal{B}(\mathbb{C}%
^{n})$ of $\rho _{2}$ as
\begin{equation}
\Psi _{\rho _{2}}:=\mathrm{D}_{\rho _{2}}^{1/2}\in \mathcal{H}_{\rho _{2}}\ .
\label{defnition vaccum at}
\end{equation}%
The GNS representation $(\mathcal{H}_{\rho _{2}},\pi _{\rho _{2}},\Psi
_{\rho _{2}})$ is known in the literature as the \emph{standard
representation} of the state $\rho _{2}$. See \cite[Section 5.4]%
{DerezinskiFruboes2006}.

Let the \textquotedblleft left\textquotedblright\ and \textquotedblleft
right\textquotedblright\ von Neumann algebras be respectively defined by%
\begin{equation*}
\underrightarrow{\mathfrak{M}}:=\left\{ \underrightarrow{A}\ :A\in \mathcal{B%
}(\mathbb{C}^{n})\right\} =\pi _{\rho _{2}}\left( \mathcal{B}(\mathbb{C}%
^{n})\right)
\end{equation*}%
and
\begin{equation*}
\underleftarrow{\mathfrak{M}}:=\left\{ \underleftarrow{A}\ :A\in \mathcal{B}(%
\mathbb{C}^{n})\right\} =\underrightarrow{\mathfrak{M}}^{\prime }\ .
\end{equation*}%
For any state $\rho _{1}\in \mathcal{B}(\mathbb{C}^{n})^{\ast }$, there is $%
\Psi _{\rho _{1}}:=\mathrm{D}_{\rho _{1}}^{1/2}\in \mathcal{H}_{\rho _{2}}$
such that
\begin{equation}
\rho _{1}\left( B\right) =\left\langle \Psi _{\rho _{1}},\pi _{\rho
_{2}}\left( B\right) \Psi _{\rho _{1}}\right\rangle _{\mathcal{H}_{\rho
_{2}}}\ ,\qquad B\in \mathcal{B}(\mathbb{C}^{n})\ .  \label{state1}
\end{equation}%
In fact, $\mathrm{D}_{\rho _{1}}\in \mathcal{B}(\mathbb{C}^{n})$ is the
density matrix of $\rho _{1}$. Moreover, $\Psi _{\rho _{1}}\in \mathcal{H}%
_{\rho _{2}}$ induces a vector state $\rho _{1}^{\prime }$\ on the commutant
$\underleftarrow{\mathfrak{M}}$ of $\underrightarrow{\mathfrak{M}}$:
\begin{equation}
\rho _{1}^{\prime }(\underleftarrow{A}):=\langle \Psi _{\rho _{1}},%
\underleftarrow{A}\Psi _{\rho _{1}}\rangle _{\mathcal{H}_{\rho _{2}}}=\rho
_{1}(A)\ ,\qquad \underleftarrow{A}\in \underleftarrow{\mathfrak{M}}\ .
\label{state2}
\end{equation}%
Its GNS representation is obviously given by $\mathcal{H}_{\rho _{1}^{\prime
}}:=\mathrm{supp}\left( \rho _{1}\right) \mathcal{H}_{\rho _{2}}$ endowed
with the scalar product (\ref{trace h at}), $\pi _{\rho _{1}^{\prime }}:=%
\mathbf{1}_{\underleftarrow{\mathfrak{M}}}$ and $\Psi _{\rho _{1}^{\prime
}}=\Psi _{\rho _{1}}=\mathrm{D}_{\rho _{1}}^{1/2}$. Moreover, $\mathcal{D}%
_{\rho _{1}^{\prime }}=\underrightarrow{\mathfrak{M}}\Psi _{\rho _{1}}$ and,
for any $\underrightarrow{A}\in \underrightarrow{\mathfrak{M}}$, the bounded
operator $R_{\rho _{1}^{\prime }}(\underrightarrow{A}\Psi _{\rho _{1}})$
defined by (\ref{R operator}) equals in this case $\underrightarrow{A}$, see
(\ref{R operatorbis}). Note that $\Psi _{\rho _{2}}\in \mathcal{H}_{\rho
_{2}}$ induces a vector state $\rho _{2}^{\prime }$\ on the commutant $%
\underrightarrow{\mathfrak{M}}$ of $\underleftarrow{\mathfrak{M}}$:
\begin{equation*}
\rho _{2}^{\prime }(\underrightarrow{A}):=\langle \Psi _{\rho _{2}},%
\underrightarrow{A}\Psi _{\rho _{2}}\rangle _{\mathcal{H}_{\rho _{2}}}=\rho
_{2}(A)\ ,\qquad \underrightarrow{A}\in \underrightarrow{\mathfrak{M}}\ .
\end{equation*}%
Then, using the cyclicity of the trace we obtain that the quadratic form $q_{%
\tilde{\rho}_{2},\rho _{1}^{\prime }}$ defined by (\ref{quadratic form})
equals%
\begin{equation*}
q_{\tilde{\rho}_{2},\rho _{1}^{\prime }}(\psi +\psi _{\bot },\tilde{\psi}+%
\tilde{\psi}_{\bot })=\langle \tilde{\psi},\underleftarrow{\mathrm{D}_{\rho
_{1}}^{-1}}\underrightarrow{\mathrm{D}_{\rho _{2}}}\psi \rangle _{\mathcal{H}%
_{\rho _{2}}}
\end{equation*}%
for any $\psi ,\tilde{\psi}\in \mathcal{D}_{\rho _{1}^{\prime }}$ and $\psi
_{\bot },\tilde{\psi}_{\bot }\in \mathcal{D}_{\rho _{1}^{\prime }}^{\bot }$.
In particular, the spatial derivative $(\partial _{\rho _{1}^{\prime }}%
\tilde{\rho}_{2})$ on the subspace $\mathrm{supp}\left( \rho _{1}\right) =%
\underrightarrow{\mathfrak{M}}\Psi _{\rho _{1}}$ is equal to
\begin{equation*}
\partial _{\rho _{1}^{\prime }}\rho _{2}=\underleftarrow{\mathrm{D}_{\rho
_{1}}^{-1}}\ \underrightarrow{\mathrm{D}_{\rho _{2}}}\ .
\end{equation*}%
Since $\underleftarrow{\mathfrak{M}}=\underrightarrow{\mathfrak{M}}^{\prime
} $, we observe that, on the subspace $\mathrm{supp}\left( \rho _{1}\right) =%
\underrightarrow{\mathfrak{M}}\Psi _{\rho _{1}}$,
\begin{equation*}
\ln \left( \partial _{\rho _{1}^{\prime }}\rho _{2}\right) =\ln
\underrightarrow{\mathrm{D}_{\rho _{2}}}-\ln \underleftarrow{\mathrm{D}%
_{\rho _{1}}}=\underrightarrow{\ln \mathrm{D}_{\rho _{2}}}-\underleftarrow{%
\ln \mathrm{D}_{\rho _{1}}}\ .
\end{equation*}%
By combining this equality with (\ref{relative entropy general}), (\ref%
{state1}) and (\ref{state2}), we arrive at
\begin{equation*}
S_{\mathcal{B}(\mathbb{C}^{n})}\left( \rho _{1}|\rho _{2}\right) =\mathrm{%
Trace}_{\mathbb{C}^{n}}\left( \mathrm{D}_{\rho _{1}}\left( \ln \mathrm{D}%
_{\rho _{1}}-\ln \mathrm{D}_{\rho _{2}}\right) \right) \in \mathbb{R}%
_{0}^{+}\ .
\end{equation*}
\end{proof}

\subsection{Infinite System as Thermodynamic Limit\label{Section finite
volume system}}

We present here the infinite system considered above as the thermodynamic
limit of finite volume systems. The aim is to show that all properties of
the infinite model result from the corresponding ones of the finite volume
system, at large volume.

\subsubsection{Finite Volume Free Fermion Systems on the Lattice}

First, fix $L\in \mathbb{R}^{+}$ and recall that $\Lambda _{L}$ is the box (%
\ref{eq:def lambda n}) of side length $2[L]+1$. Let%
\begin{equation*}
\lbrack \Delta _{\mathrm{d}}^{(L)}(\psi )](x):=2d\psi
(x)-\sum\limits_{|z|=1,x+z\in \Lambda _{L}}\psi (x+z)\ ,\text{\qquad }x\in
\Lambda _{L},\ \psi \in \ell ^{2}(\Lambda _{L})\ ,
\end{equation*}%
be, up to a minus sign, the discrete Laplacian restricted to the box $%
\Lambda _{L}$. For any $\omega \in \Omega $, we denote by $V_{\omega }^{(L)}$
the restriction of $V_{\omega }$ to $\ell ^{2}(\Lambda _{L})\subset \ell
^{2}(\mathfrak{L})$:%
\begin{equation*}
V_{\omega }^{(L)}(\mathfrak{e}_{x}):=\mathbf{1}\left[ x\in \Lambda _{L}%
\right] V_{\omega }(\mathfrak{e}_{x})\ ,\qquad x\in \mathfrak{L}\ .
\end{equation*}%
Recall that $\left\{ \mathfrak{e}_{x}\right\} _{x\in \mathfrak{L}}$ is the
canonical orthonormal basis $\mathfrak{e}_{x}(y)\equiv \delta _{x,y}$ of $%
\ell ^{2}(\mathfrak{L})$. Then, for any $\omega \in \Omega $ and $\lambda
\in \mathbb{R}_{0}^{+}$, define the bounded self--adjoint operator
\begin{equation}
h_{L}^{(\omega ,\lambda )}:=\Delta _{\mathrm{d}}^{(L)}+\lambda V_{\omega
}^{(L)}\in \mathcal{B}(\ell ^{2}(\Lambda _{L}))\ .
\label{one-particle hamil restriction}
\end{equation}%
Obviously, this operator can also be extended to a bounded operator $\tilde{h%
}_{L}^{(\omega ,\lambda )}$ on $\ell ^{2}(\mathfrak{L})$ by defining%
\begin{equation*}
\tilde{h}_{L}^{(\omega ,\lambda )}(\mathfrak{e}_{x}):=\left\{
\begin{array}{lll}
h_{L}^{(\omega ,\lambda )}(\mathfrak{e}_{x}) &  & \text{for }x\in \Lambda
_{L}\ . \\
0 &  & \text{for }x\in \mathfrak{L}\backslash \Lambda _{L}\ .%
\end{array}%
\right.
\end{equation*}

Since $\mathcal{U}_{\Lambda _{L}}$ is isomorphic to the algebra of all
bounded linear operators on the fermion Fock space
\begin{equation*}
\mathcal{F}:=\bigwedge (\ell ^{2}(\Lambda _{L}))\ ,
\end{equation*}%
the Hamiltonian (\ref{def H loc}), that is,
\begin{equation}
H_{L}^{(\omega ,\lambda )}=\sum\limits_{x,y\in \Lambda _{L}}\langle
\mathfrak{e}_{x},h_{L}^{(\omega ,\lambda )}\mathfrak{e}_{y}\rangle
a_{x}^{\ast }a_{y}\in \mathcal{U}_{\Lambda _{L}}\ ,
\label{one-particle hamil restrictionbis}
\end{equation}%
can be seen as the \emph{second quantization} of $h_{L}^{(\omega ,\lambda )}$
for all $\omega \in \Omega $ and $\lambda \in \mathbb{R}_{0}^{+}$. It is
well--known in this case that the one--parameter (Bogoliubov) group $\tau
^{(\omega ,\lambda ,L)}:=\{\tau _{t}^{(\omega ,\lambda ,L)}\}_{t\in {\mathbb{%
R}}}$ of automorphisms uniquely defined by the condition
\begin{equation*}
\tau _{t}^{(\omega ,\lambda ,L)}(a(\psi ))=a(\mathrm{e}^{it\tilde{h}%
_{L}^{(\omega ,\lambda )}}(\psi ))\ ,\text{\qquad }t\in \mathbb{R},\ \psi
\in \ell ^{2}(\mathfrak{L})\ ,
\end{equation*}%
(cf. \cite[Theorem 5.2.5]{BratteliRobinson}) satisfies
\begin{equation*}
\tau _{t}^{(\omega ,\lambda ,L)}(B)=\mathrm{e}^{itH_{L}^{(\omega ,\lambda
)}}B\mathrm{e}^{-itH_{L}^{(\omega ,\lambda )}}\ ,\text{\qquad }B\in \mathcal{%
U}\ ,
\end{equation*}%
for each $L\in \mathbb{R}^{+}$ and all $\omega \in \Omega $ and $\lambda \in
\mathbb{R}_{0}^{+}$.

Let $\varrho ^{(\beta ,\omega ,\lambda ,L)}$ be the unique $(\tau ^{(\omega
,\lambda ,L)},\beta )$--KMS state for any $\omega \in \Omega $ and $\lambda
\in \mathbb{R}_{0}^{+}$ at fixed inverse temperature $\beta \in \mathbb{R}%
^{+}$. It is again well--known that this state is directly related with the
Gibbs state $\mathfrak{g}^{(\beta ,\omega ,\lambda ,L)}$ associated with the
Hamiltonian $H_{L}^{(\omega ,\lambda )}$ and defined by
\begin{equation}
\mathfrak{g}^{(\beta ,\omega ,\lambda ,L)}\left( B\right) :=\mathrm{Trace}_{%
\mathcal{F}}\left( B\frac{\mathrm{e}^{-\beta H_{L}^{(\omega ,\lambda )}}}{%
\mathrm{Trace}_{\mathcal{F}}(\mathrm{e}^{-\beta H_{L}^{(\omega ,\lambda )}})}%
\right) \ ,\text{\qquad }B\in \mathcal{U}_{\Lambda _{L}}\ ,
\label{Gibbs state H_n}
\end{equation}%
for any $L,\beta \in \mathbb{R}^{+}$, $\omega \in \Omega $ and $\lambda \in
\mathbb{R}_{0}^{+}$. Indeed,
\begin{equation}
\varrho ^{(\beta ,\omega ,\lambda ,L)}(B_{1}B_{2})=\mathfrak{g}^{(\beta
,\omega ,\lambda ,L)}(B_{1})\mathrm{tr}(B_{2})\ ,\text{\qquad }B_{1}\in
\mathcal{U}_{\Lambda _{L}}\ ,\;B_{2}\in \mathcal{U}_{\mathfrak{L}\backslash
\Lambda _{L}}\ ,  \label{Gibbs state H_nbis}
\end{equation}%
where $\mathrm{tr}$ is the normalized trace (state) on $\mathcal{U}$. Note
that $\mathrm{tr}$ is also named \emph{tracial state} and satisfies a
product property, see \cite[Section 4.2]{Araki-Moriya}. Here, $\mathcal{U}_{%
\mathfrak{L}\backslash \Lambda _{L}}\subset \mathcal{U}$ is the $C^{\ast }$%
--algebra generated by $\left\{ a_{x}\right\} _{x\in \mathfrak{L}\backslash
\Lambda _{L}}$ and the identity. In particular,
\begin{equation*}
\varrho ^{(\beta ,\omega ,\lambda ,L)}(B)=\mathfrak{g}^{(\beta ,\omega
,\lambda ,L)}(B)\ ,\text{\qquad }B\in \mathcal{U}_{\Lambda _{L}}\ .
\end{equation*}

Let $\mathbf{A}\in \mathbf{C}_{0}^{\infty }$. For any sufficiently large $%
L\in \mathbb{R}^{+}$, $W_{t}^{\mathbf{A}}\in \mathcal{U}_{\Lambda _{L}}$.
Therefore, consider the following finite dimensional initial value problem
on the space $\mathcal{B}(\mathcal{U}_{\Lambda _{L}})$ of bounded operators
on $\mathcal{U}_{\Lambda _{L}}$ for any sufficiently large $L\in \mathbb{R}%
^{+}$:
\begin{equation}
\forall s,t\in {\mathbb{R}},\ t\geq s:\ \partial _{t}\tau _{t,s}^{(\omega
,\lambda ,\mathbf{A},L)}=\tau _{t,s}^{(\omega ,\lambda ,\mathbf{A},L)}\circ
\delta _{t}^{(\omega ,\lambda ,\mathbf{A},L)},\quad \tau _{s,s}^{(\omega
,\lambda ,\mathbf{A},L)}:=\mathbf{1}\ ,  \label{cauchy bis}
\end{equation}%
with $\mathbf{1}$ being here the identity in $\mathcal{U}_{\Lambda _{L}}$.
Here, the infinitesimal generator $\delta _{t}^{(\omega ,\lambda ,\mathbf{A}%
,L)}$ of $\tau _{t,s}^{(\omega ,\lambda ,\mathbf{A},L)}$\ equals
\begin{equation}
\delta _{t}^{(\omega ,\lambda ,\mathbf{A},L)}(\cdot ):=i[H_{L}^{(\omega
,\lambda )}+W_{t}^{\mathbf{A}},\ \cdot \ ]  \label{cauchy bisbis}
\end{equation}%
and is of course a bounded operator acting on $\mathcal{U}_{\Lambda _{L}}$.
Therefore, using the Dyson--Phillips series one shows, analogously to
Section \ref{Section existence dynamics}, the existence of a strongly
continuous two--parameter (quasi--free) family $\{\tau _{t,s}^{(\omega
,\lambda ,\mathbf{A},L)}\}_{t\geq s}$ of automorphisms of the finite
dimensional $C^{\ast }$--algebra $\mathcal{U}_{\Lambda _{L}}$ satisfying (%
\ref{cauchy bis}). See, e.g., \cite[Sections 5.4.2., Proposition 5.4.26.]%
{BratteliRobinson} which, for the finite--volume dynamics, gives similar
results to Theorems \ref{bound incr 1 Lemma copy(4)}, \ref{Theo int pict},
and Proposition \ref{bound incr 1 Lemma copy(8)}.

\subsubsection{Heat Production and Internal Energy Increment}

Similar to Definition \ref{Heat production definition}, for any $L,\beta \in
\mathbb{R}^{+}$, $\omega \in \Omega $, $\lambda \in \mathbb{R}_{0}^{+}$ and $%
\mathbf{A}\in \mathbf{C}_{0}^{\infty }$, the heat production $\mathbf{Q}%
^{(\omega ,\mathbf{A},L)}\equiv \mathbf{Q}^{(\beta ,\omega ,\lambda ,\mathbf{%
A},L)}$ in the finite volume fermion system is defined, for any $t\geq t_{0}$%
, by%
\begin{equation}
\mathbf{Q}^{(\omega ,\mathbf{A},L)}\left( t\right) :=\beta ^{-1}S_{\mathcal{U%
}_{\Lambda }}\left( \mathfrak{g}^{(\beta ,\omega ,\lambda ,L)}\circ \tau
_{t,t_{0}}^{(\omega ,\lambda ,\mathbf{A},L)}|\mathfrak{g}^{(\beta ,\omega
,\lambda ,L)}\right) \in \left[ 0,\infty \right] \ .
\label{Heat production definition finite-volume}
\end{equation}%
Here, $S_{\mathcal{U}_{\Lambda }}$ is the quantum relative entropy defined
by (\ref{relative entropy finite dimensional}).

Like (\ref{entropic energy increment})--(\ref{electro free energy}), the
internal energy increment $\mathbf{S}^{(\omega ,\mathbf{A},L)}\equiv \mathbf{%
S}^{(\beta ,\omega ,\lambda ,\mathbf{A},L)}$ and the electromagnetic
potential energy $\mathbf{P}^{(\omega ,\mathbf{A},L)}\equiv \mathbf{P}%
^{(\beta ,\omega ,\lambda ,\mathbf{A},L)}$ in the finite volume fermion
system are respectively defined by%
\begin{eqnarray*}
\mathbf{S}^{(\omega ,\mathbf{A},L)}\left( t\right)  &:=&\mathfrak{g}^{(\beta
,\omega ,\lambda ,L)}(\tau _{t,t_{0}}^{(\omega ,\lambda ,\mathbf{A},L)}(H_{L}^{(\omega ,\lambda )}))-\mathfrak{g}^{(\beta ,\omega ,\lambda
,L)}(H_{L}^{(\omega ,\lambda )})\ , \\
\mathbf{P}^{(\omega ,\mathbf{A},L)}\left( t\right)  &:=&\mathfrak{g}^{(\beta
,\omega ,\lambda ,L)}(\tau _{t,t_{0}}^{(\omega ,\lambda ,\mathbf{A},L)}(W_{t}^{\mathbf{A}})) \ ,
\end{eqnarray*}%
for any $L,\beta \in \mathbb{R}^{+}$, $\omega \in \Omega $, $\lambda \in
\mathbb{R}_{0}^{+}$, $\mathbf{A}\in \mathbf{C}_{0}^{\infty }$ and $t\geq
t_{0}$. Using \cite[Lemma 5.4.27]{BratteliRobinson} one also obtains that
\begin{equation}
\mathbf{S}^{(\omega ,\mathbf{A},L)}\left( t\right) +\mathbf{P}^{(\omega ,%
\mathbf{A},L)}\left( t\right) =\int_{t_{0}}^{t}\mathfrak{g}^{(\beta ,\omega
,\lambda ,L)}(\tau _{s,t_{0}}^{(\omega ,\lambda ,\mathbf{A},L)}(\partial
_{s}W_{s}^{\mathbf{A}}))\mathrm{d}s  \label{workbis}
\end{equation}%
with
\begin{equation*}
\mathfrak{g}^{(\beta ,\omega ,\lambda ,L)}(\tau _{s,t_{0}}^{(\omega ,\lambda
,\mathbf{A},L)}(\partial _{s}W_{s}^{\mathbf{A}}))
\end{equation*}%
being, as in (\ref{work}), the infinitesimal work of the electromagnetic
field at time $t\in \mathbb{R}$ on the finite volume fermion system.

Similar to Theorem \ref{main 1 copy(1)} the internal energy increment and
the heat production also coincide at finite volume:

\begin{satz}[Heat production as internal energy increment]
\label{main 1 copy(1)bis}\mbox{
}\newline
For any $L,\beta \in \mathbb{R}^{+}$, $\omega \in \Omega $, $\lambda \in
\mathbb{R}_{0}^{+}$, $\mathbf{A}\in \mathbf{C}_{0}^{\infty }$ and all $t\geq
t_{0}$,
\begin{equation*}
\mathbf{Q}^{(\omega ,\mathbf{A},L)}\left( t\right) =\mathbf{S}^{(\omega ,%
\mathbf{A},L)}\left( t\right) \in \mathbb{R}_{0}^{+}\ \ .
\end{equation*}
\end{satz}

\begin{proof}
The arguments follow those of \cite{frohlichschwartzmerkli}. Note first that
\begin{equation}
\mathfrak{g}^{(\beta ,\omega ,\lambda ,L)}\circ \tau _{t,t_{0}}^{(\omega
,\lambda ,\mathbf{A},L)}\in \mathcal{U}_{\Lambda _{L}}^{\ast }
\label{passivity ohm2}
\end{equation}%
is a state with adjusted density matrix. Its von Neumann entropy is equal,
up to a minus sign, to
\begin{equation}
S_{\mathcal{U}_{\Lambda }}(\mathfrak{g}^{(\beta ,\omega ,\lambda ,L)}\circ
\tau _{t,t_{0}}^{(\omega ,\lambda ,\mathbf{A},L)}|\mathrm{tr})=S_{\mathcal{U}%
_{\Lambda }}(\mathfrak{g}^{(\beta ,\omega ,\lambda ,L)}|\mathrm{tr})
\label{passivity ohm3}
\end{equation}%
for all $t\geq t_{0}$ because $\tau _{t,t_{0}}^{(\omega ,\lambda ,\mathbf{A}%
,L)}$ is an automorphism on $\mathcal{U}_{\Lambda _{L}}$. Recall that we
denote by $\mathrm{tr}$ the normalized trace on $\mathcal{U}_{\Lambda }$
and, by finite dimensionality, the relative entropy equals (\ref{relative
entropy finite dimensional}), see also Lemma \ref{local AC-conductivity
lemma copy(2)}. Using (\ref{relative entropy finite dimensional}), (\ref%
{Gibbs state H_n}) and (\ref{passivity ohm3}), we directly derive the
equality%
\begin{equation*}
\mathbf{S}^{(\omega ,\mathbf{A},L)}\left( t\right) =\beta ^{-1}S_{\mathcal{U}%
_{\Lambda }}\left( \mathfrak{g}^{(\beta ,\omega ,\lambda ,L)}\circ \tau
_{t,t_{0}}^{(\omega ,\lambda ,\mathbf{A},L)}|\mathfrak{g}^{(\beta ,\omega
,\lambda ,L)}\right) =:\mathbf{Q}^{(\omega ,\mathbf{A},L)}\left( t\right) \ .
\end{equation*}
\end{proof}

Therefore, similar to Theorem \ref{Thm Heat production as power series} (i),
it is straightforward to write the heat production in terms of
multi--commutators: For any $L,\beta \in \mathbb{R}^{+}$, $\omega \in \Omega
$, $\lambda \in \mathbb{R}_{0}^{+}$, $\mathbf{A}\in \mathbf{C}_{0}^{\infty }$
and $t\geq t_{0}$,%
\begin{equation}
\mathbf{Q}^{(\omega ,\mathbf{A},L)}\left( t\right) =\sum\limits_{k\in {%
\mathbb{N}}}\int_{t_{0}}^{t}\mathrm{d}s_{1}\cdots \int_{t_{0}}^{s_{k-1}}%
\mathrm{d}s_{k}\ \mathbf{u}_{k}^{(\omega ,\mathbf{A},L)}\left( s_{1},\ldots
,s_{k},t\right) \ ,  \label{incr lim 3bis}
\end{equation}%
with the finite volume heat energy coefficient $\mathbf{u}_{k}^{(\omega ,%
\mathbf{A},L)}\equiv \mathbf{u}_{k}^{(\beta ,\omega ,\lambda ,\mathbf{A},L)}$
defined by%
\begin{align}
\mathbf{u}_{k}^{(\omega ,\mathbf{A},L)}\left( s_{1},\ldots ,s_{k},t\right) &
:=\sum\limits_{x,y\in \Lambda _{L},|x-y|\leq 1}i^{k}\langle \mathfrak{e}%
_{x},h_{L}^{(\omega ,\lambda )}\mathfrak{e}_{y}\rangle
\label{incr lim 3bis+1} \\
& \times \mathfrak{g}^{(\beta ,\omega ,\lambda ,L)}\left(
[W_{s_{k}-t_{0},s_{k}}^{(\mathbf{A},L)},\ldots ,W_{s_{1}-t_{0},s_{1}}^{(%
\mathbf{A},L)},\tau _{t-t_{0}}^{(\omega ,\lambda ,L)}(a_{x}^{\ast
}a_{y})]^{(k+1)}\right)  \notag
\end{align}%
for any $k\in {\mathbb{N}}$, $L,\beta \in \mathbb{R}^{+}$, $\omega \in
\Omega $, $\lambda \in \mathbb{R}_{0}^{+}$, $\mathbf{A}\in \mathbf{C}%
_{0}^{\infty }$, $t\geq t_{0}$ and $s_{1},\ldots ,s_{k}\in \lbrack t_{0},t]$%
. Similar to the definition (\ref{def LA}) of $W_{t,s}^{\mathbf{A}}$, note
that we use above the notation
\begin{equation*}
W_{t,s}^{(\mathbf{A},L)}\equiv W_{t,s}^{(\omega ,\lambda ,\mathbf{A}%
,L)}:=\tau _{t}^{(\omega ,\lambda ,L)}(W_{s}^{\mathbf{A}})\in \mathcal{U}
\end{equation*}%
for any $t,s\in \mathbb{R}$ and $\mathbf{A}\in \mathbf{C}_{0}^{\infty }$.
Theorem \ref{Thm Heat production as power series} (ii) also holds at finite
volume, uniformly w.r.t. $L\in \mathbb{R}^{+}$.

\subsubsection{Thermodynamic Limit of the Finite Volume System}

We first summarize well--known results on the infinite volume dynamics and
thermal state:

\begin{satz}[Infinite volume dynamics and thermal state]
\label{conv Gibbs}\mbox{
}\newline
Let $\beta \in \mathbb{R}^{+}$, $\omega \in \Omega $ and $\lambda \in
\mathbb{R}_{0}^{+}$. Then: \newline
\emph{(i)} For any $t\in \mathbb{R}$, the localized (quasi--free)
automorphism $\tau _{t}^{(\omega ,\lambda ,L)}$ strongly converges to $\tau
_{t}^{(\omega ,\lambda )}$, as $L\rightarrow \infty $.\newline
\emph{(ii)} The $(\tau ^{(\omega ,\lambda ,L)},\beta )$--KMS state $\varrho
^{(\beta ,\omega ,\lambda ,L)}$ converges to the $(\tau ^{(\omega ,\lambda
)},\beta )$--KMS state $\varrho ^{(\beta ,\omega ,\lambda )}$ in the weak$%
^{\ast }$--topology, as $L\rightarrow \infty $.
\end{satz}

\begin{proof}
See \cite[Chapters 5.2 and 5.3]{BratteliRobinson}.
\end{proof}

Then, from Equation (\ref{workbis}), Theorems \ref{conv Gibbs} and Lebes%
\-%
gue's dominated convergence theorem, it is clear that the energy increments $%
\mathbf{S}^{(\omega ,\mathbf{A})}$ and $\mathbf{P}^{(\omega ,\mathbf{A})}$
respectively defined by (\ref{entropic energy increment}) and (\ref{electro
free energy}) result from the finite volume energy increments $\mathbf{S}%
^{(\omega ,\mathbf{A},L)}$ and $\mathbf{P}^{(\omega ,\mathbf{A},L)}$:

\begin{koro}[Energy increments as thermodynamic limits]
\mbox{
}\newline
For any $\beta \in \mathbb{R}^{+}$, $\omega \in \Omega $, $\lambda \in
\mathbb{R}_{0}^{+}$ and all $t\geq t_{0}$,
\begin{equation*}
\mathbf{S}^{(\omega ,\mathbf{A})}\left( t\right) =\lim_{L\rightarrow \infty }%
\mathbf{S}^{(\omega ,\mathbf{A},L)}\left( t\right) \qquad \text{and}\qquad
\mathbf{P}^{(\omega ,\mathbf{A})}\left( t\right) =\lim_{L\rightarrow \infty }%
\mathbf{P}^{(\omega ,\mathbf{A},L)}\left( t\right) \ .
\end{equation*}
\end{koro}

By combining this with Theorems \ref{main 1 copy(1)} and \ref{main 1
copy(1)bis} we show the same property for the heat production:

\begin{koro}[Heat production as thermodynamic limit]
\label{coro Thm limit heat prioduction}\mbox{
}\newline
For any $\beta \in \mathbb{R}^{+}$, $\omega \in \Omega $, $\lambda \in
\mathbb{R}_{0}^{+}$ and all $t\geq t_{0}$,
\begin{equation*}
\mathbf{Q}^{(\omega ,\mathbf{A})}\left( t\right) =\lim_{L\rightarrow \infty }%
\mathbf{Q}^{(\omega ,\mathbf{A},L)}\left( t\right) \ .
\end{equation*}
\end{koro}

By Theorem \ref{Thm Heat production as power series}, recall that, for any $%
\mathbf{A}\in \mathbf{C}_{0}^{\infty }$, there is a constant $\eta _{0}\in
\mathbb{R}^{+}$ such that, for all $|\eta |\in \lbrack 0,\eta _{0}]$, $\beta
\in \mathbb{R}^{+}$, $\omega \in \Omega $, $\lambda \in \mathbb{R}_{0}^{+}$
and $t\geq t_{0}$,%
\begin{equation}
\mathbf{Q}^{(\omega ,\eta \mathbf{A})}\left( t\right) =\sum\limits_{k\in {%
\mathbb{N}}}\int_{t_{0}}^{t}\mathrm{d}s_{1}\cdots \int_{t_{0}}^{s_{k-1}}%
\mathrm{d}s_{k}\ \mathbf{u}_{k}^{(\omega ,\eta \mathbf{A})}\left(
s_{1},\ldots ,s_{k},t\right) \ .  \label{incr lim 3bis0}
\end{equation}%
Here, the heat energy coefficient $\mathbf{u}_{k}^{(\omega ,\mathbf{A}%
)}\equiv \mathbf{u}_{k}^{(\beta ,\omega ,\lambda ,\mathbf{A})}$ is defined,
for any $k\in {\mathbb{N}}$, $\beta \in \mathbb{R}^{+}$, $\omega \in \Omega $%
, $\lambda \in \mathbb{R}_{0}^{+}$, $\mathbf{A}\in \mathbf{C}_{0}^{\infty }$%
, $t\geq t_{0}$ and $s_{1},\ldots ,s_{k}\in \lbrack t_{0},t]$, by
\begin{multline*}
\mathbf{u}_{k}^{(\omega ,\mathbf{A})}\left( s_{1},\ldots ,s_{k},t\right)
:=\sum\limits_{x,y\in \mathfrak{L},|x-y|\leq 1}i^{k}\langle \mathfrak{e}%
_{x},\left( \Delta _{\mathrm{d}}+\lambda V_{\omega }\right) \mathfrak{e}%
_{y}\rangle \\
\times \varrho ^{(\beta ,\omega ,\lambda )}\left( [W_{s_{k}-t_{0},s_{k}}^{%
\mathbf{A}},\ldots ,W_{s_{1}-t_{0},s_{1}}^{\mathbf{A}},\tau
_{t-t_{0}}^{(\omega ,\lambda )}(a_{x}^{\ast }a_{y})]^{(k+1)}\right)
\end{multline*}%
with $W_{t,s}^{\mathbf{A}}:=\tau _{t}^{(\omega ,\lambda )}(W_{s}^{\mathbf{A}%
})\in \mathcal{U}$ for any $t,s\in \mathbb{R}$, see (\ref{def LA}). The
series (\ref{incr lim 3bis0}) absolutely converges for the above range of
parameters.

Then, by combining these last series with (\ref{incr lim 3bis})--(\ref{incr
lim 3bis+1}), Theorem \ref{conv Gibbs} and Corollary \ref{coro Thm limit
heat prioduction} one directly obtains the following result:

\begin{satz}[Taylor coefficients of $\mathbf{Q}^{(\protect\omega ,\protect%
\eta \mathbf{A})}$ as thermodynamic limit]
\label{conv Gibbs copy(1)}\mbox{
}\newline
For any $\beta \in \mathbb{R}^{+}$, $\omega \in \Omega $, $\lambda \in
\mathbb{R}_{0}^{+}$, $\mathbf{A}\in \mathbf{C}_{0}^{\infty }$, $t\geq t_{0}$
and $m\in \mathbb{N}$,%
\begin{eqnarray*}
\partial _{\eta }^{m}\mathbf{Q}^{(\omega ,\eta \mathbf{A})}\left( t\right)
|_{\eta =0} &=&\underset{L\rightarrow \infty }{\lim }\partial _{\eta }^{m}%
\mathbf{Q}^{(\omega ,\eta \mathbf{A,}L)}\left( t\right) |_{\eta =0} \\
&=&\sum\limits_{k\in {\mathbb{N}}}\int_{t_{0}}^{t}\mathrm{d}s_{1}\cdots
\int_{t_{0}}^{s_{k-1}}\mathrm{d}s_{k} \\
&&\qquad \qquad \underset{L\rightarrow \infty }{\lim }\left\{ \partial
_{\eta }^{m}\mathbf{u}_{k}^{(\omega ,\eta \mathbf{A},L)}\left( s_{1},\ldots
,s_{k},t\right) |_{\eta =0}\right\} \ ,
\end{eqnarray*}%
where the above series is absolutly convergent.
\end{satz}

\begin{proof}
The proof uses similar arguments to those showing Lemma \ref{bound incr 1
Lemma copy(1)}. We omit the details.
\end{proof}

\bigskip

\noindent \textit{Acknowledgments:} We would like to thank Volker Bach,
Horia Cornean, Abel Klein and Peter M\"{u}ller for relevant references and
interesting discussions as well as important hints. JBB and WdSP are also
very grateful to the organizers of the Hausdorff Trimester Program entitled
\textquotedblleft \textit{Mathematical challenges of materials science and
condensed matter physics}\textquotedblright\ for the opportunity to work
together on this project at the Hausdorff Research Institute for Mathematics
in Bonn. This research is supported by the ICTP South American Institute for Fundamental Research, 
a grant of the {}\textquotedblleft Inneruniversit{\"{a}%
}re Forschungsf{\"{o}rde%
\-%
rung}\textquotedblright {} of the Johannes Gutenberg University, 
by the agency FAPESP under the grant 2013/13215-5 as well as by the Basque Government through the grant IT641-13 
and the BERC 2014-2017 program and 
by the Spanish Ministry of Economy and Competitiveness MINECO: BCAM Severo Ochoa accreditation SEV-2013-0323, MTM2010-16843.

\frenchspacing
\bibliographystyle{plain}

\begin{thebibliography}{99}
\bibitem[A1]{Araki1} \textsc{H. Araki}, Relative Entropy of States of von
Neumann Algebras, \textit{Publ. Res. Inst. Math. Sci. Kyoto Univ.} \textbf{11%
} (1976) 809--833.

\bibitem[A2]{Araki2} \textsc{H. Araki}, Relative Entropy of States of von
Neumann Algebras II, \textit{Publ. Res. Inst. Math. Sci. Kyoto Univ.}
\textbf{13} (1977) 173--192.

\bibitem[AM]{Araki-Moriya} \textsc{H. Araki and H. Moriya}, Equilibrium
Statistical Mechanics of Fermion Lattice Systems, \textit{Rev. Math. Phys.}
\textbf{15} (2003) 93--198.

\bibitem[AJP]{AttalJoyePillet2006a} \textsc{C.-A. Pillet}, Quantum Dynamical
Systems, in \textit{Open Quantum Systems I: The Hamiltonian Approach,}
Volume 1880 of Lecture Notes in Mathematics, editors: S.\ Attal, A.\ Joye,
C.-A. Pillet. Springer-Verlag, 2006.

\bibitem[BGKS]{jfa} \textsc{J.-M. Bouclet, F. Germinet, A. Klein and J.H.
Schenker}, Linear response theory for magnetic Schr\"{o}dinger operators in
disordered media, \textit{Journal of Functional Analysis} \textbf{226}
(2005) 301--372.

\bibitem[BR1]{BratteliRobinsonI} \textsc{O. Bratteli and D.W. Robinson},
\textit{Operator Algebras and Quantum Statistical Mechanics, Vol. I, 2nd ed.}
Springer-Verlag, New York, 1996.

\bibitem[BR2]{BratteliRobinson} \textsc{O. Bratteli and D.W. Robinson},
\textit{Operator Algebras and Quantum Statistical Mechanics, Vol. II, 2nd ed.%
} Springer-Verlag, New York, 1996.

\bibitem[B]{Cornean} \textsc{M.H. Brynildsen, H.D. Cornean}, On the Verdet
constant and Faraday rotation for graphene-like materials, \emph{Preprint}
arXiv:1112.2613 (2011).

\bibitem[C]{thermo0} \textsc{R. Clausius}, Ueber die bewegende Kraft der W%
\"{a}rme und die Gesetze, welche sich daraus f\"{u}r die W\"{a}rmelehre
selbst ableiten lassen, Annalen der Physik und Chemie (Poggendorff, Leipzig)
\textbf{155}(3) (1850) 368--394, page 384.

\bibitem[C]{Connes} \textsc{A. Connes}, On the spatial theory of von Neumann
algebras, \textit{Journal of Functional Analysis} \textbf{35}(2) (1980),
153--164.

\bibitem[DF]{DerezinskiFruboes2006} \textsc{J.\ Derezinski and R.\ Fr\"{u}%
boes}, Fermi {G}olden {R}ule and {O}pen {Q}uantum {S}ystems. In \textit{Open
Quantum Systems III}, Volume 1882, pages 67--116. Springer-Verlag, 2006.

\bibitem[DZ]{Dembo} \textsc{A. Dembo, O. Zeitouni}, \textit{Large Deviations
Techniques and Applications}, Springer-Verlag, New York, 1998.

\bibitem[EL]{Thermo1} \textsc{G. Emch, C. Liu}, \textit{The Logic of
Thermostatistical Physics}, Springer-Verlag, New York, 2002.

\bibitem[EN]{EngelNagel} \textsc{K.-J. Engel and R. Nagel}, \textit{%
One--Parmeter Semigroups for Linear Evolution Equations}, Springer New York,
2000.

\bibitem[FMSU]{frohlichschwartzmerkli} \textsc{J. Fr\"{o}hlich, M. Merkli,
S. Schwartz and D. Ueltschi}, Statistical Mechanics of Thermodynamic
Processes, A garden of Quanta, 345--363, World Scientific Publishing, River
Edge, 2003.

\bibitem[FMU]{FroehlichMerkliUeltschi} \textsc{J. Fr\"{o}hlich, M. Merkli
and D. Ueltschi}, Dissipative Transport: Thermal Contacts and Tunnelling
Junctions, \textit{Ann. Henri Poincar\'{e}} \textbf{4} (2003) 897--945.

\bibitem[JP]{JaksicPillet} \textsc{V. Jaksic and C.-A. Pillet}, A Note on
the Entropy Production Formula, \textit{Contemp. Math.} \textbf{327} (2003)
175--181.

\bibitem[J]{Joulesup} \textsc{J. P. Joule}, On the Production of Heat by Voltaic Electricity. \textit{Philosophical
Transactions of the Royal Society of London} \textbf{4} (1840) 280–-282.

\bibitem[K]{Werner Kirsch} \textsc{Werner Kirsch}, An invitation to random
Schr\"{o}dinger operators. In \textit{Random Schr\"{o}dinger Operators},
Panoramas et Synth\`{e}ses (Soci\'{e}t\'{e} Math\'{e}matique de France)
\textbf{25} (2008).

\bibitem[KLM]{Annale} \textsc{A. Klein, O. Lenoble, and P. M\"{u}ller}, On
Mott's formula for the ac-conductivity in the Anderson model, \textit{Annals
of Mathematics} \textbf{166} (2007) 549--577.

\bibitem[KM]{JMP-autre} \textsc{A. Klein and P. M\"{u}ller}, The
Conductivity Measure for the Anderson Model, \textit{Journal of Mathematical
Physics, Analysis, Geometry} \textbf{4} (2008) 128--150.

\bibitem[OP]{OhyaPetz} \textsc{M. Ohya and D. Petz}, \textit{Quantum Entropy
and Its Use}, Springer--Verlag Berlin Heidelberg New York, 1993 (corrected
second printing 2004).

\bibitem[RS1]{ReedSimonI} \textsc{M. Reed and B. Simon}, \textit{Methods of
Modern Mathematical Physics, Vol. I: Fourier Analysis, Self--Adjointness},
Academic Press, New York-London, 1975.

\bibitem[RS2]{ReedSimonII} \textsc{M. Reed and B. Simon}, \textit{Methods of
Modern Mathematical Physics, Vol. II: Fourier Analysis, Self--Adjointness},
Academic Press, New York-London, 1975.

\bibitem[SF]{thermo2} \textsc{W.K. Abou Salem and J\"{u}rg Fr\"{o}hlich},
Status of the Fundamental Laws of Thermodynamics, \textit{J. Stat. Phys.}
\textbf{126} (2007) 1045--1068.
\end{thebibliography}

\end{document}